\documentclass[11pt, a4paper]{article}

\usepackage{amsthm,booktabs,integrable,jheparxiv,xparse}

\title{On the associativity of 1-loop corrections to the celestial operator product in gravity}

\author{Roland Bittleston}

\affiliation{Perimeter Institute for Theoretical Physics,\\ 51 Caroline Street, Waterloo, Ontario, Canada}

\emailAdd{rbittleston@perimeterinstitute.ca}

\begin{document}
	
\abstract{The question of whether the holomorphic collinear singularities of graviton amplitudes define a consistent chiral algebra has garnered much recent attention. We analyse a version of this question for infinitesimal perturbations around the self-dual sector of 4d Einstein gravity. The singularities of tree amplitudes in such perturbations do form a consistent chiral algebra, however at 1-loop its operator products are corrected by the effective graviton vertex. We argue that the chiral algebra can be interpreted as the universal holomorphic surface defect in the twistor uplift of self-dual gravity, and show that the same correction is induced by an anomalous diagram in the bulk-defect system. The 1-loop holomorphic collinear singularities do not form a consistent chiral algebra. The failure of associativity can be traced to the existence of a recently discovered gravitational anomaly on twistor space. It can be restored by coupling to an unusual 4\textsuperscript{th}-order gravitational axion, which cancels the anomaly by a Green-Schwarz mechanism. Alternatively, the anomaly vanishes in certain theories of self-dual gravity coupled to matter, including in self-dual supergravity.}
	
\maketitle
\flushbottom


\section{Introduction}

The celestial holography program suggests that theories of gravity in 4d asymptotically flat spacetimes admit the action of an infinite dimensional chiral algebra consisting of asymptotic symmetries arising from soft theorems \cite{Strominger:2013jfa,He:2014laa}. Universal collinear singularities of amplitudes in the gravitational theory are controlled by operator products in the corresponding chiral algebra. (See, e.g., the reviews \cite{Strominger:2017zoo,Raclariu:2021zjz,Pasterski:2021rjz}.) It was recently argued that at tree level the tower of positive-helicity soft graviton symmetries in Einstein gravity generate the loop algebra of a certain infinite dimensional Lie algebra closely related to the wedge subalgebra of $w_{1+\infty}$ \cite{Guevara:2021abz,Strominger:2021lvk}. We shall denote the chiral algebra of positive-helicity symmetries by $\cV$. It reduces to one considered by Penrose \cite{Penrose:1968me,Penrose:1976js} in the twistor setting \cite{Adamo:2021lrv}.

We study the 4d theory of self-dual gravity arising as the $\kappa\to0$ limit of the (self-dual) Palatini action \cite{Plebanski:1975wn,Capovilla:1991qb,Smolin:1992wj}. It's field equations supply a self-dual vacuum Einstein metric on spacetime. (For an introduction to self-dual gravity we refer the reader to the survey \cite{Krasnov:2016emc} and references therein.) The theory includes states of both helicities, with the negative-helicity field acting as a Lagrange multiplier.\\

In the beautiful paper \cite{Ball:2021tmb} it was argued that the celestial chiral algebra of self-dual gravity is undeformed by quantum corrections. We seek to extend this result by asking a considerably more general question: do the collinear singularities of amplitudes in generic 1\textsuperscript{st}-order deformations of quantum self-dual gravity form a consistent chiral algebra?\\

This question is motivated by the results of \cite{Costello:2022wso,Costello:2022upu}, in which the authors study an extension of the celestial chiral algebra of self-dual Yang-Mills whose operator products encode the collinear limits of tree form factors, i.e., tree gluon amplitudes in the presence of local operators. Integrating a local operator over spacetime gives a 1\textsuperscript{st}-order deformation of self-dual Yang-Mills, so their results show that the collinear limits of tree amplitudes in such deformations are universal. This extended celestial chiral algebra receives quantum corrections, the simplest of which can be attributed to the 1-loop all-plus gluon splitting amplitude discovered in \cite{Bern:1994zx}. Associativity of the operator product is violated in the deformed chiral algebra, but can be restored by cancelling an anomaly in the twistor description of the theory.\\

Self-dual gravity has no BRST invariant local operators, at least of vanishing ghost number, but 1\textsuperscript{st}-order deformations of the theory still make sense. There does exist an extension of its celestial chiral algebra whose operator products describe the collinear limits of tree amplitudes in 1\textsuperscript{st}-order deformations of self-dual gravity. It consists of a split extension of $\cV$ by its adjoint module, encoding the states of the negative-helicity field \cite{Costello:2022wso}.

We show that the operator products of this extended celestial chiral algebra receive 1-loop corrections. This might seem to conflict with the well known result of \cite{Bern:1998sv} that there are no 1-loop splitting amplitudes in gravity. The resolution lies in the fact that this result applies only to the \emph{true} collinear limit, whereas operator products in the chiral algebra describe \emph{holomorphic} collinear singularities. In fact, 1-loop graviton amplitudes can acquire double poles in the holomorphic collinear limit from the diagram illustrated in figure \ref{fig:splitting} \cite{Brandhuber:2007up,Alston:2015gea}. One first order pole is generated by the loop integration, and a second arises from the internal propagator. Understanding the holomorphic collinear singularities, and more generally the factorization at complex kinetic points, of amplitudes is central to the application of recursion methods \cite{Britto:2005fq,Bern:2005hs,Alston:2012xd}.

In \cite{Brandhuber:2007up} an effective 1-loop graviton vertex was introduced describing this double pole. We massage the effective vertex into a 1-loop holomorphic splitting amplitude
\be \label{eq:splitting-intro} \mathrm{Split}^\mathrm{1-loop}_+(1^+,2^+;t) = \frac{4t(1-t)}{180(4\pi)^2}\frac{[12]^4}{\la12\ra^2}\,. \ee
(For details of our conventions see section \ref{sec:splitting}.) In principle, it can be computed as the partially off-shell 3-point amplitude obtained by stripping off the tree in figure \ref{fig:splitting}. We emphasise that in the true limit $[12]\propto\overline{\la12\ra}\to0$ it vanishes. For completeness, we explicitly recover this splitting amplitude by taking the holomorphic collinear limit of the 5-point mostly-plus amplitude originally computed in \cite{Dunbar:2010xk}.

\begin{figure}[h!]
	\centering
	\includegraphics[scale=0.3]{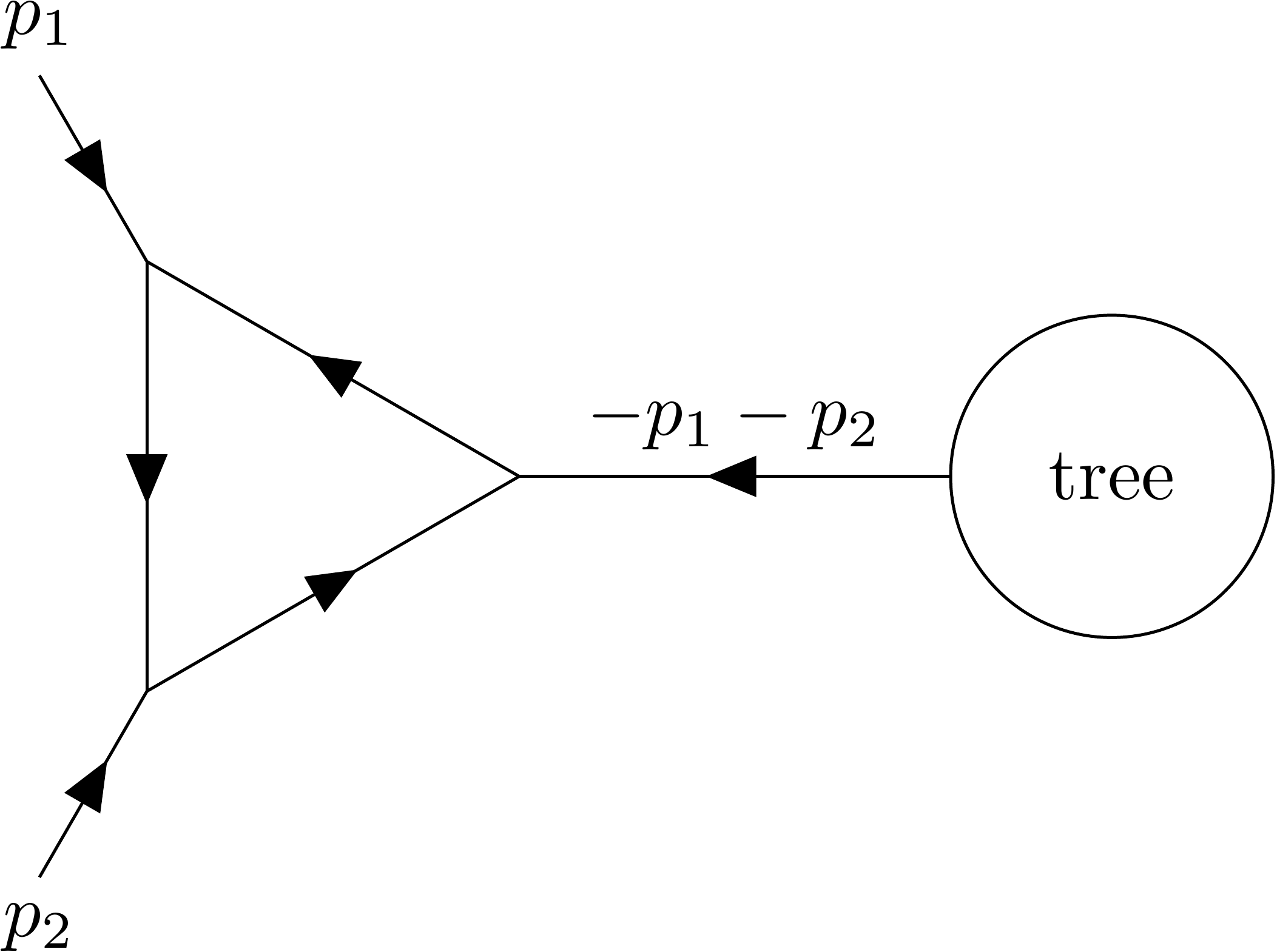}
	\caption{\emph{This diagram is responsible for double poles in the holomorphic collinear limits of graviton amplitudes. Arrows indicate the flow of momenta and helicity through the diagram.}} \label{fig:splitting}
\end{figure}

The above diagram does not contribute in the holomorphic collinear limits of the 1-loop all-plus amplitudes, on account of the vanishing of the 1-minus trees, and so does not modify the operator products of the celestial chiral algebra of self-dual gravity. However the 1-minus tree amplitudes in a generic 1\textsuperscript{st}-order deformation of self-dual gravity will not vanish, so the diagram in figure \ref{fig:splitting} corrects the operator products of our extended celestial chiral algebra. It's important here that our algebra includes states associated to both the positive- and negative-helicity modes, since the 1-loop effective vertex describes a 2-plus to 1-minus scattering process.\\

We then determine whether this 1-loop deformation of the extended celestial chiral algebra is consistent, i.e., whether it has an associative operator product. In order to check this, we first need to characterise the potential 1-loop corrections in general, which we do through symmetry arguments. The leading double poles in the corrected operator products necessarily take the same form as those induced by the 1-loop holomorphic splitting amplitude, but there can also be subleading simple poles which are bilinear in generators. We find that the holomorphic collinear singularities \emph{do not} form an associative chiral algebra. This failure signals that the holomorphic collinear singularities of amplitudes in 1\textsuperscript{st}-order deformations are not universal: they depend on the choice of deformation. The non-vanishing 1-loop all-plus amplitudes in self-dual gravity are responsible such non-universal behaviour.\\

To understand how to correct this failure, we appeal to twistor theory. Penrose's non-linear graviton construction realises self-dual vacuum Einstein spacetimes in terms of the complex geometry of a 6d twistor space \cite{Penrose:1968me,Penrose:1976js}. This can be exploited to obtain an uplift of self-dual gravity to a holomorphic Poisson-BF theory on the twistor space of flat spacetime \cite{Mason:2007ct,Sharma:2021pkl}.

As a holomorphic theory, the twistor uplift admits surface defects supporting chiral algebras. Classically, the universal holomorphic surface defect coincides with the extended celestial chiral algebra describing the collinear singularities of tree amplitudes in 1\textsuperscript{st}-order deformations of self-dual gravity. We argue that this identification persists at the quantum level. As evidence, we identify the anomalous 1-loop diagrams in the coupled bulk-defect system necessitating corrections to operator products. Evaluating these for particular configurations of external legs, we recover precisely the deformation induced by the 1-loop holomorphic splitting amplitude. We further find that certain subleading terms bilinear in generators are non-vanishing.

From the twistor perspective, the failure of associativity can be traced to a recently discovered gravitational anomaly in Poisson-BF theory \cite{Bittleston:2022nfr}. The presence of this anomaly obstructs the existence of the twistorial theory, and there's no reason to expect its holomorphic surface defects support consistent chiral algebras. This is compatible with the spacetime interpretation: the anomalous box diagram on twistor space can be identified with the 1-loop 4-point all-plus amplitude on spacetime.\\

This perspective presents a natural method of correcting the associativity failure: by cancelling the twistorial gravitational anomaly.

One method of doing so is using a kind of Green-Schwarz mechanism. This requires coupling to a $(2,1)$-form field on twistor space, describing an unusual 4\textsuperscript{th}-order gravitational axion on spacetime. Incorporating the axion states into the chiral algebra we find that associativity is restored, conditional on suitable normalizations of certain 1-loop corrected operator products. These match those determined by the splitting amplitude, and the results of our direct calculation in the bulk-defect system on twistor space.

Alternatively, the twistorial anomaly vanishes in certain theories of self-dual gravity coupled to matter, including in self-dual supergravity. We briefly discuss how associativity is restored in these examples.\\

Even in cases where the anomaly is cancelled, at 1-loop operator products of generators involve terms bilinear in generators. These do not directly correspond to holomorphic splitting amplitudes, but nonetheless are anticipated to describe the subleading holomorphic collinear singularities of 1-minus amplitudes in Einstein gravity coupled to a 4\textsuperscript{th}-order gravitational axion.

The presence of non-linearities in the chiral algebra mean that its mode algebra has been deformed as an associative algebra, not a Lie algebra. In this sense it's a kind of quantum group, whose relationship to self-dual gravity is somewhat analogous to that of the Yangian to the principal chiral model \cite{Luscher:1977uq,Luscher:1977rq,Brezin:1979am,Bernard:1990jw}. (In fact, the twistor formulation of self-dual Yang-Mills theory is known to be closely related to the 4d Chern-Simons description of the principal chiral model \cite{Costello:2019tri,Bittleston:2020hfv}.) One consequence of this is that the quantum deformation we find here doesn't seem to be related to the symplecton deformation of the wedge subalgebra of $w_{1+\infty}$ \cite{Biedenharn:1971et}, which instead appears as the celestial chiral algebra of the Moyal deformation of self-dual gravity \cite{Monteiro:2022lwm,Bu:2022iak}.\footnote{The wedge subalgebra of $w_{1+\infty}$ admits a continuous family of inequivalent Lie algebra deformations labelled by $s\geq -1/2$ \cite{Pope:1991ig,Shen:1992dd}. One corollary of our results in appendix \ref{app:extended-sl2C} is that only one of these, the symplecton with $s=-1/4$, defines a deformation of $\cV$. This was also noted in \cite{Bu:2022iak}.}\\

The paper is organized as follows:
\begin{itemize}
	\item[-] In section \ref{sec:background} we review relevant background material. Much of this concerns the twistor formulation of self-dual gravity, which is only essential for sections \ref{sec:defect} and \ref{sec:cure}.
	\item[-] In section \ref{sec:splitting} we obtain the 1-loop holomorphic splitting amplitude from the effective 1-loop 3-point vertex. We explain why it deforms the extended celestial chiral algebra, and determine the resulting 1-loop corrected operator products.
	\item[-] In section \ref{sec:associativity} we characterise the 1-loop corrections to the extended celestial chiral algebra in general. We then show that associativity of the operator product is violated. This failure signals that the holomorphic collinear singularities of amplitudes in 1\textsuperscript{st}-order deformations are not universal.
	\item[-] In section \ref{sec:defect} we argue that the extended celestial chiral algebra can be identified with the universal holomorphic surface defect in the twistor uplift of self-dual gravity. We identify the anomalous 1-loop diagrams in the coupled bulk-defect system necessitating corrections to operator products, and evaluate them for particular choices of external fields. The failure of associativity can be traced to the recently discovered twistorial gravitational anomaly.
	\item[-] In section \ref{sec:cure} we consider self-dual gravitational theories for which the twistorial anomaly vanishes. One example is given by coupling self-dual gravity to an unusual 4\textsuperscript{th}-order gravitational axion, which cancels the twistorial anomaly by a Green-Schwarz mechanism. We verify that the previously identified associativity failure is cured in its extended celestial chiral algebra. We also briefly discuss minimally coupling to matter in such a way that the anomaly vanishes directly.
	\item[-] In section \ref{sec:discussion} conclude by discussing possible extensions of this work.
\end{itemize}

Our conventions are as follows. We let $\alpha,\beta,\ldots\in\{1,2\}$ and $\dot\alpha,\dot\beta,\ldots\in\{1,2\}$ denote left- and right-handed Weyl spinor indices respectively. Spinor indices can be raised or lowered using the Levi-Civita symbols $\eps^{\da\db}$ and $\eps_{\da\db}$ with the conventions $u_\da = \eps_{\da\db}u^\db$, along with $\eps^{\da\dc}\eps_{\dc\db} = \delta^\da_{~\,\db}$. $\gSL_2(\bbC)$ invariant spinor contractions will be denoted  $[uv] = u^\da v_\da$ and $\langle ab\rangle = a^\alpha b_\alpha$. Greek letters $\mu,\nu,\dots$ from the middle of the alphabet denote spacetime indices.


\section{Background}
\label{sec:background}

In this section we sum up essential background material. First, we cover the classical twistor action for self-dual (SD) gravity as originally proposed in \cite{Mason:2007ct}. Next, we recall the celestial chiral algebra (CCA) of SD gravity and its extension to include negative-helicity states. It arises as the universal holomorphic surface defect in the twistor uplift of the theory \cite{Costello:2022wso}. Finally, we review the newly discovered anomaly in the twistor description of SD gravity, and two methods of cancelling it, from \cite{Bittleston:2022nfr}.


\subsection{Twistor action for self-dual gravity} \label{subsec:classical-SDGR}

The twistor space of flat Euclidean spacetime is denoted $\bbPT$. As a complex manifold it's the total space of the rank 2 vector bundle
\be \cO(1)\oplus\cO(1)\mapsto\bbCP^1\,. \ee
It admits a natural action of the (double cover of the) complexified Lorentz group $\gSL_2(\bbC)_+\times\gSL_2(\bbC)_-$ where the first factor mixes the fibres and the second acts by M\"{o}bius transformations on the base. We use holomorphic coordinates $v^\da$ for $\da=1,2$ on the fibres and the inhomogeneous coordinate $z$ on the base. Note that the $v^\da$ coordinates are singular at $z=\infty$ since they represent sections of $\cO(1)$.

Complexified spacetime is recovered as the space of degree 1 curves $\bbCP^1\hookrightarrow\bbPT$. Explicitly we can parametrize such curves by a quadruple $(u^1,u^2,\tilde u^1,\tilde u^2)\in\bbC^4$ as
\be \label{eq:incidence} v^1 = u^1 - z\tilde u^2\,,\qquad v^2 = u^2 + z\tilde u^1\,. \ee
The equations \eqref{eq:incidence} are known as incidence relations, and we refer to the curve they define as a twistor line.

To recover Euclidean spacetime we introduce the antipodal map $\sigma:z\mapsto -1/\bar z$ which extends to act on the $v^\da$ coordinates by $(v^1,v^2)\mapsto(\bar v^2/\bar z,-\bar v^1/\bar z)$. This automorphism has no fixed points, so there's a unique degree 1 curve passing through a point $Z\in\bbPT$ and $\sigma(Z)\in\bbPT$. Twistor lines of this type correspond to points in the Euclidean slice $\bbR^4\subset\bbC^4$, and we refer to them as real twistor lines. Explicitly they are determined by a pair $(u^1,u^2)\in\bbC^2\cong\bbR^4$ with $(\tilde u^1,\tilde u^2)=(\bar u^1,\bar u^2)$.

The choice of involution $\sigma$ breaks $\gSL_2(\bbC)_+\times\gSL_2(\bbC)_-$ to $\gSU(2)_+\times\gSU(2)_-$, which induces an action of $\gSO_4(\bbR)$ on $\bbR^4\cong\bbC^2$ preserving the metric
\be \dif s^2 = \dif u^1\dif\bar u^1 + \dif u^2\dif\bar u^2\,. \ee

This construction furnishes twistor space with a smooth fibration over spacetime $\bbPT\to\bbR^4$ with $S^2$ fibres. Assuming the reality condition $(\tilde u^1,\tilde u^2)=(\bar u^1,\bar u^2)$, we can solve the incidence relation \eqref{eq:incidence} for $u=(u^1,u^2)$ to get
\be u^1 = \frac{v^1 + z\bar v^2}{1+|z|^2}\,,\qquad u^2 = \frac{v^2 - z\bar v^1}{1+|z|^2}\,. \ee
We emphasise that although we are choosing to parametrize Euclidean spacetime $\bbR^4\cong\bbC^2$ with complex coordinates $(u^1,u^2)$, twistor space is indifferent to this choice. Indeed, at fixed $u$ the coordinate $z$ can be viewed as parametrizing the space of K\"{a}hler structures on the tangent space compatible with the metric $\dif s^2$ (and a choice of orientation).

Although we will largely be working in Euclidean signature on spacetime, we note that an important result of \cite{Costello:2021bah} is that complex symmetries of classical twistorial theories persist at the quantum level. In particular, translation invariant twistorial theories have analytic correlation functions with poles on the complexified light cone. In this way our results extend to Lorentzian and ultrahyperbolic signatures.\\

It is a remarkable fact that solutions to chiral equations on spacetime are equivalent to holomorphic structures on twistor space. Linearly, solutions to spin $s$ free field equations on spacetime can be identified with Dolbeault cohomology classes $H^1(\bbPT,\cO(2s-2))$ on twistor space \cite{Penrose:1968me}. Non-linearly, solutions to the SD Yang-Mills equations on spacetime are equivalent to holomorphic bundles on twistor space which are trivial on twistor lines \cite{Ward:1977ta}. In this work we will be concerned with curved spacetimes obeying the SD vacuum Einstein equations, i.e., which are Ricci flat and have SD Weyl curvature. (We will always assume a vanishing cosmological constant.) These can be characterised by twistor data using Penrose's non-linear graviton construction \cite{Penrose:1976js,Atiyah:1978wi}.

The non-linear graviton construction provides a correspondence between 4 dimensional manifolds $\cM$ with a conformal class of metrics $[g]$ whose Weyl curvature is SD and 3 dimensional complex manifolds $\cPT$ possessing a 4 parameter family of rational curves $\cL_u$ with normal bundle $\cN_u=\cO(1)\oplus\cO(1)$.\footnote{Strictly this correspondence applies to conformal classes of holomorphic metrics on complexified spacetimes. To get Riemannian metrics on real slices we also require a free antiholomorphic involution of $\cPT$ acting as the antipodal map on $\cL_u$.} The moduli space of such curves is identified with $\cM$, and $T_u\cM\cong H^0(\cL_u,\cN_u)$. $\cPT$ is the curved twistor space of $\cM$.

In perturbation theory we can view the complex structure on $\cPT$ as a deformation of that on $\bbPT$, which can be explicitly parametrized by a Beltrami differential $V\in\Omega^{0,1}(\bbPT,T^{1,0}\bbPT)$. The antiholomorphic Dolbeault operator $\nbar$ on $\cPT$ is related to $\bar\p$ on $\bbPT$ by
\be \nbar = \bar\p + \cL_V\,. \ee
We emphasise that the undeformed Dolbeault operator $\dbar$ is not trivial. In this work, whenever the Lie derivative $\cL$ appears we mean the $(1,0)$ Lie derivative $\cL_V = [V\ip\,,\p\,]$. This makes no difference when acting on Dolbeault forms with only antiholomorphic degree, but on generic $(p,q)$-forms the distinction is important. Often when $(1,0)$ vector fields act on antiholomorphic forms we suppress the Lie derivative, since there is no ambiguity. The Beltrami differential determines a deformation of the almost complex structure, and for this to be an integrable deformation the Nijenhuis tensor must also vanish
\be N = \nbar^2 = \bar\p V + \frac{1}{2}[V,V] = 0\,. \ee

The non-linear graviton construction can be further refined to obtain SD vacuum metrics by requiring that the curved twistor space admit a fibration $\cPT\to\bbCP^1$ with a $\cO(2)$-valued symplectic form on the fibres. In particular, if this condition holds there is a unique Ricci flat metric $g$ in the conformal class determined by $\cPT$. The vertical tangent bundle with respect to the holomorphic fibration over $\bbCP^1$ is denoted by $\cN$, and its pullback to a twistor line can be identified with the normal bundle $\cN_u$. In the case of flat space this fibration clearly exists and the required symplectic form is $\dif v^2\wedge\dif v^1$. The almost complex structure deformation determined by $V$ preserves this form, and so determines an SD vacuum metric if it's Hamiltonian with respect to the bivector $\Pi = \p_{v^1}\vee\p_{v^2} = \frac{1}{2}\epsilon^{\db\da}\p_\da\vee\p_\db\in\wedge^2\cN\otimes\cO(-2)$ for $\p_\da = \p_{v^\da}$. Writing $\{\ ,\ \}$ for the corresponding Poisson bracket this is the condition
\be \label{eq:PB} V = \{h,\ \} = \epsilon^{\db\da}\p_\da h\,\p_\db\,. \ee
where $h\in\Omega^{0,1}(\bbPT,\cO(2))$ is the Hamiltonian. $h$ must take values in the line bundle $\cO(2)$ to compensate the twisting of the bivector. For such almost complex structure deformations the Nijenhuis tensor is itself Hamiltonian, i.e., $N = \{T,\ \}$, where
\be T = \bar\p h + \frac{1}{2}\{h,h\}\,. \ee
In this way we can characterise SD vacuum fluctuations of spacetime by a field $h\in\Omega^{0,1}(\bbPT,\cO(2))$ obeying $T=0$.\\

An action on $\bbPT$ with field equations imposing integrability of a Hamiltonian complex structure deformation was provided by Mason and Wolf in \cite{Mason:2007ct}. (See also the discussions in \cite{Adamo:2013tja,Adamo:2013cra}.) They take for dynamical fields $h\in\Omega^{1,0}(\bbPT,\cO(2))$ and a Lagrange multiplier field $g\in\Omega^{3,1}(\bbPT,\cO(-2))$. The action is then
\be \label{eq:S-PBF} S_\mathrm{PBF}[g,h] = \frac{1}{2\pi\im}\int_\bbPT g\wedge\bigg(\dbar h + \frac{1}{2}\{h,h\}\bigg)\,. \ee
Note that the twisting of $g$ has been chosen to ensure that Lagrangian takes values in $\Omega^{3,3}(\bbPT)$. We refer to this as (holomorphic) Poisson-BF theory. It can be interpreted as a truncation of the $\cN=8$ twistor string \cite{Skinner:2013xp} to the constant map sector. The equations of motion are
\be T = \dbar h + \frac{1}{2}\{h,h\} = 0\,,\qquad \nbar g = \dbar g + \{h,g\} = 0\,. \ee
As per the above discussion, the first ensures that the almost complex structure determined by $V = \{h,\ \}$ is integrable, and the second tells us that $g$ is closed in this deformed complex structure. When coupling to defects we will find it convenient to introduce $\tilde h\in\Omega^{0,1}(\bbPT,\cO(-6))$ defined by
\be g = \Dif^3Z\wedge\tilde h\,, \ee
where
\be \Dif^3Z = \frac{1}{2}\,\dif z\,\dif v^\da\,\dif v_\da \ee
generates $H^{3,0}(\bbPT,\cO(4))\cong\bbC$.

The action \eqref{eq:S-PBF} has two families of gauge symmetries. The first is given by
\be \delta h = \nbar\chi = \dbar\chi + \{h,\chi\}\,,\qquad \delta g = \{g,\chi\}\,, \ee
where $\chi$ can be interpreted as the Hamiltonian for a vector field $\xi = \{\chi,\ \}$. This gauge symmetry therefore identifies complex structure deformations related by Poisson diffeomorphisms. The second is
\be \label{eq:gauge-phi} \delta h = 0\,,\qquad \delta g = \nbar\phi = \dbar\phi + \{h,\phi\}\,, \ee
which tells us that $g$ is determined modulo the addition of exact forms in the deformed almost complex structure. On-shell we therefore find that $g\in H^{3,1}(\cPT,\cO(-2))$.

This theory admits a number of symmetries which play an important role in what follows. Firstly, we have the action of $\gSL_2(\bbC)_+\times\gSL_2(\bbC)_-$, which includes a $z$ scaling symmetry $z\mapsto s z$ for $s\in\bbC^*$. Secondly, we can scale the fibres as $v^\da\mapsto tv^\da$ for $t\in\bbC^*$ corresponding to a spacetime dilation. These also appear as symmetries of the CCA, and we will use them to constrain OPEs.\\

The spacetime counterpart of the twistor action \eqref{eq:S-PBF} can be recovered by following a rather technical partial gauge fixing procedure as outlined in \cite{Sharma:2021pkl}. (See also \cite{Bittleston:2022nfr} for an alternative reduction leading to a superficially rather different, but ultimately equivalent, action for SD gravity originally appearing in \cite{Krasnov:2021cva}.) It depends on a vierbein $e^{\da\alpha}$ and a 1-form Lagrange multiplier field $\Gamma_{\alpha\beta}$ symmetric in its spinor indices, and has action
\be \label{eq:S-SDGR} S_\mathrm{SDGR}[\Gamma,e] = \frac{1}{2}\int_{\bbR^4}\Gamma_{\alpha\beta}\wedge\dif(e^{\alpha\da}\wedge e^\beta_{~\,\da})\,. \ee
This is the $\kappa\to0$ limit of the SD Palatini action \cite{Ashtekar:1987qx,Capovilla:1991qb,Smolin:1992wj}.\footnote{The SD Palatini action differs from the tetradic Palatini action by a topological Nieh-Yan term, and so is equivalent to Einstein gravity in perturbation theory. It coincides with the Holst action if the Barbero-Immirzi parameter is specialised to $\beta=-\im$. The theory of SD gravity we study here is then the $\kappa\to0$ limit of the SD Palatini action.} The field $\Gamma_{\alpha\beta}$ encodes the negative-helicity states. Since it appears linearly in the action, it can be treated as a loop counting parameter in place of $\hbar$. We will often refer to this theory simply as SD gravity.\\

We will find it convenient to employ the BV formalism in order to simplify perturbative calculations. This is a powerful generalization of the standard Fadeev-Popov gauge fixing. Here we briefly review the elements needed for this paper, which are fairly minimal. For a more extensive description of the BV formulation of holomorphic theories we refer the reader to \cite{Williams:2018ows}, and for the particular case of Poisson-BF theory to \cite{Elliott:2020jai}.

We extend the dynamical fields $h,g$ to polyform fields $\bh\in\Omega^{0,\bullet}(\bbPT,\cO(2))[1]$ and $\bg\in\Omega^{3,\bullet}(\bbPT,\cO(-2))[1]$. Here the symbol $\bullet$ indicates, e.g., that $\bh$ is not merely a $(0,1)$-form, but rather a sum of $(0,q)$-forms for $0\leq q\leq3$. The symbol $[1]$ stipulates a shift in cohomological degree so that it coincides with minus the ghost number. Expanding $\bh$ into components of definite degree we have
\be \bh = \chi + h + g^\vee + \phi^\vee\,. \ee
In degrees $-1,0$ are the Poisson diffeomorphism ghost $\chi$ and the physical field $h$ respectively. (We abuse notation by denoting the ghost with the same symbol as the gauge parameter above. The difference is that the ghost is graded odd.) The components with positive degree are the antifield $g^\vee$ to the physical field $g$, and the antifield $\phi^\vee$ to the ghost $\phi$ associated to the gauge symmetry \eqref{eq:gauge-phi}. We can similarly expand $\bg = \phi + g + h^\vee + \chi^\vee$ to obtain the remaining fields. For coupling to defects it's again convenient to define $\tilde\bh\in\Omega^{0,\bullet}(\bbPT,\cO(-6))[1]$ by $\bg = \Dif^3Z\,\tilde\bh$, where the physical field corresponding to $\tilde\bh$ is $\tilde h$.

The extended space of fields is equipped with a canonical $(-1)$-shifted symplectic pairing
\be \frac{1}{2\pi\im}\int_\bbPT \delta\bg\,\delta\bh\,, \ee
where here alone $\delta$ denotes the exterior derivative on the space of fields. This induces a BV bracket on functionals $\{\ ,\ \}_\mathrm{BV}$, which should not be confused with Poisson bracket introduced around equation \eqref{eq:PB}.

The classical action \eqref{eq:S-PBF} is replaced with its BV counterpart
\be \label{eq:S-PBF-BV} S_\mathrm{PBF}[\bg,\bh] = \frac{1}{2\pi\im}\int_\bbPT\bg\,\bigg(\bar\p\bh + \frac{1}{2}[\bh,\bh]\bigg)\,, \ee
where the integral is understood to extract the $(3,3)$-form part of the Lagrangian. Although this takes exactly the same form as \eqref{eq:S-PBF}, it receives contribution from more terms. These encode the action of the BRST operator, which is
\be \delta = \{S_\mathrm{PBF},\ \}_\mathrm{BV}\,. \ee
In particular
\be \delta\bh = T(\bh) = \bar\p\bh + \frac{1}{2}\{\bh,\bh\}\,,\qquad \delta\bg = \nbar\bg = \bar\p\bg + \{\bh,\bg\}\,. \ee
Restricting the above formulae to definite cohomological degree we recover the structure constants of the gauge algebra, the BRST transformations of the physical fields and the classical equations of motion. Nilpotence of the BRST operator follows from the classical master equation $\{S_\mathrm{PBF},S_\mathrm{PBF}\}_\mathrm{BV}=0$.

When performing computations in the BV formalism we should employ a `gauge fixing', which in this context refers to a choice of Lagrangian subspace in the space of fields. For us this will always be $\bar\p^\dag\bh=\bar\p^\dag\bg=0$ where $\bar\p^\dag = - \star\bar\p\,\star$ for $\star$ the anti-linear Hodge star associated to some choice of Hermitian metric on $\bbPT$.

The BV formalism encodes the data of the theory efficiently, but its real power lies in perturbative computations.


\subsection{Celestial chiral algebra of self-dual gravity}
\label{subsec:classical-chiral-algebra}

Here we write down an extended CCA for SD gravity incorporating states of both positive- and negative-helicity \cite{Guevara:2021abz,Strominger:2021lvk}. We also review how this extended algebra arises as the universal holomorphic surface defect in the twistor uplift of the theory \cite{Costello:2022wso}.\\

The extended CCA of SD gravity is a vertex algebra consisting of two infinite towers of states $w[m,n]$, $\tilde w[m,n]$ for $m,n\geq0$, corresponding to positive- and negative-helicity gravitons respectively.\footnote{In this work the term `chiral algebra' refers to the mathematical notion of a vertex algebra. Physically chiral algebras arise as holomorphic subsectors of 2d CFTs.} These are listed in table \ref{tab:classical-generators}, together with their conformal spins.

\begin{table}[h!]
	\centering
	\begin{tabular}{c c c c c}
		\toprule
		Generator & Field & Spin & $\fsl_2(\bbC)_+$ representation & Dimension \\ \midrule
		$w[m,n]\,,~m,n\geq0$ & $h$ & $2-(m+n)/2$ & $\mathbf{m+n+1}$ & $2-m-n$ \\
		$\tilde w[m,n]\,,~m,n\geq0$ & $g$ & $-2-(m+n)/2$ & $\mathbf{m+n+1}$ & $-4-m-n$ \\
		\bottomrule
	\end{tabular}
	\caption{Quantum numbers of $w,\tilde w$}
	\label{tab:classical-generators}
\end{table}

The defining OPEs of the chiral algebra are \cite{Strominger:2021lvk,Costello:2022wso}
\bea \label{eq:chiral-algebra}
w[p,q](z)w[r,s](0)&\sim \frac{ps-qr}{z}w[p+r-1,q+s-1](0)\,, \\
w[p,q](z)\tilde w[r,s](0)&\sim \frac{ps-qr}{z}\tilde w[p+r-1,q+s-1](0)\,, \\
\tilde w[p,q](z)\tilde w[r,s](0)&\sim 0\,. \\
\eea
These can be better understood by introducing the Lie algebra of Hamiltonian vector fields on $\bbC^2$ equipped with the standard holomorphic symplectic structure, which we denote by $\Ham(\bbC^2)$. As a vector space $\Ham(\bbC^2)\cong\bbC[x,y]$ for $x,y$ formal parameters, with basis $t_{m,n} = x^my^n$ for $m,n\in\bbZ_{\geq0}$. The Lie algebra structure is induced by the Poisson bracket associated to the bivector $\p_x\vee\p_y$
\be [t_{p,q},t_{r,s}] = \p_x(x^py^q)\p_y(x^ry^s) - \p_y(x^py^q)\p_x(x^ry^s) = (ps-qr)t_{p+r-1,q+s-1}\,. \ee
The vertex algebra generated by the $w$ states, which we denote by $\cV$, is isomorphic to the loop algebra of $\Ham(\bbC^2)$, i.e., $\cV\cong L(\Ham(\bbC^2))$. The $\tilde w$ states generate the adjoint module of this vertex algebra, which we denote by $\widetilde\cV$. The extended CCA of SD gravity is then the split extension of $\cV$ by the module $\widetilde\cV$, and we shall denote it by $\cU$.\\

We note that much of the celestial holography literature concerns the algebra generated by just the positive-helicity soft symmetries of gravity, and so is $\cV$ alone. To avoid confusion, we will refer to this as `the celestial chiral algebra of self-dual gravity', and to its extension by the adjoint module as `the extended celestial chiral algebra of self-dual gravity'. This distinction is more important in the quantum setting, where we will interpret the algebras differently.

$\Ham(\bbC^2)$ is closely related to $w_{1+\infty}$: in particular the wedge subalgebra of $w_{1+\infty}$ is isomorphic to $\Ham(\bbC^2/\bbZ_2)$, where the $\bbZ_2$ acts by $(x,y)\mapsto(-x,-y)$. Pulling back by quotient map $\bbC^2\to\bbC^2/\bbZ_2$ induces an embedding into $\Ham(\bbC^2)$. Under the natural action of $\gSL_2(\bbC)$ we have
\bea \Ham(\bbC^2)\cong \mathbf{1}\oplus\mathbf{2}\oplus\mathbf{3}\oplus\dots\,,\qquad \Ham(\bbC^2/\bbZ_2)\cong \mathbf{1}\oplus\mathbf{3}\oplus\mathbf{5}\oplus\dots\,. \eea
($\mathfrak{sl}_2(\bbC)$ representations of dimension $d$ are denoted $\mathbf{d}$.) One significant difference between these is that the wedge subalgebra of $w_{1+\infty}$ admits a continuous family of inequivalent linear deformations \cite{Pope:1991ig}, however only one of these defines a deformation of $\Ham(\bbC^2)$. This point was noted in \cite{Bu:2022iak}, and also follows from our calculations in appendix \ref{app:extended-sl2C}.

$\gSL_2(\bbC)_-$ transformations of the complexified spacetime can be identified with M\"{o}bius maps on the celestial sphere acting as conformal symmetries of the chiral algebra. $\gSL_2(\bbC)_+$ acts in the natural way on $\Ham(\bbC^2)$. Dilations on spacetime provide an action of $\bbC^*$, and we refer to the weight of a state under this action as its dimension. The $\gSL_2(\bbC)_+$ representations and dimensions of states are listed in table \ref{tab:classical-generators}. Writing $s$ for spin and $d$ for dimension, the combination $s-d/2$ takes the values $1,0$ for the states $w,\tilde w$ respectively. These are the charges under simultaneous dilations $z\mapsto r^{-1}z$ on the celestial sphere and $x\mapsto r^{1/2}x$ on spacetime, and are particularly useful in constraining quantum corrections to OPEs.  \\

We gain insight into the chiral algebra $\cU$ by considering its low lying states, i.e., those transforming in $\fsl_2(\bbC)_+$ representations with small dimension:
\begin{itemize}
	\item[-] $w[0,0]$ has trivial OPEs with all other generators and transforms trivially under $\fsl_2(\bbC)_+$. We could therefore safely remove $w[0,0]$ from the chiral algebra, though we will find it convenient not to do so. Similar statements can be made regarding $\tilde w[0,0]$.
	\item[-] $w[p,q]$ for $p+q=1$ are distinguished as the only generators which can raise dimension. When interpreted as asymptotic symmetry generators they correspond to supertranslations.
	\item[-] $w[p,q]$ for $p+q=2$ generate the loop algebra of $\fsl_2(\bbC)_+$. All other states belong to affine primaries for this loop algebra. It is akin to the $\hat\fg_0$ subalgebra present in gauge theory case \cite{Costello:2022wso}. When interpreted as asymptotic symmetry generators these correspond to superrotations.
	\item[-] The extended CCA is generated (in the strong sense) by $w[p,q]$ with $p+q\leq3$ and $\tilde w[r,s]$ with $r+s\leq 2$.
\end{itemize}

In \cite{Costello:2022wso} two different realisations of $\cU$ are presented utilizing the twistor uplift of SD gravity, which we've seen is Poisson-BF theory. First, as the Lie algebra of infinitesimal gauge transformations on $\bbPT$ with the fibres $z=0,\infty$ removed which preserve the vacuum, and second, as the universal holomorphic surface defect. In this work we will employ the second description. Mathematically, the chiral algebra $\cU$ is Koszul dual to the algebra of local operators in classical Poisson-BF theory on $\bbPT$. For more details on the role of Koszul duality in quantum field theory we refer the reader to \cite{Paquette:2021cij}.

Consider a holomorphic surface defect supported on the real twistor line $\cL_0 = \{u^\da=0\}\subset\bbPT$. The most general holomorphic coupling to the bulk theory is through
\be \label{eq:defect-gh} \frac{1}{2\pi\im}\sum_{m,n\geq0}\int_{\cL_0}\dif z\,(w[m,n](z)\cD_{m,n}\bh + \tilde w[m,n](z)\cD_{m,n}\tilde\bh)\,, \ee
for undetermined operators on the defect suggestively denoted $w[m,n],\tilde w[m,n]$. Here we've introduced the notation $\cD_{m,n} = \p_{v^1}^m\p_{v^2}^n/m!n!$. Although we're using the BV fields $\bh,\tilde\bh$, we could replace them with their physical counterparts because only their $(0,1)$-form parts contribute.

BRST invariance of the coupled bulk-defect system imposes relations on the operators $w[m,n],\tilde w[m,n]$. Working perturbatively, we can interpret the BRST variation of the combined system using Feynman diagrams. The classical relations are determined by trees, two of which are illustrated in figure \ref{fig:BRST-tree}.

\begin{figure}[h!]
	\centering
	\begin{subfigure}[t]{0.49\textwidth}
		\centering
		\includegraphics[scale=0.3]{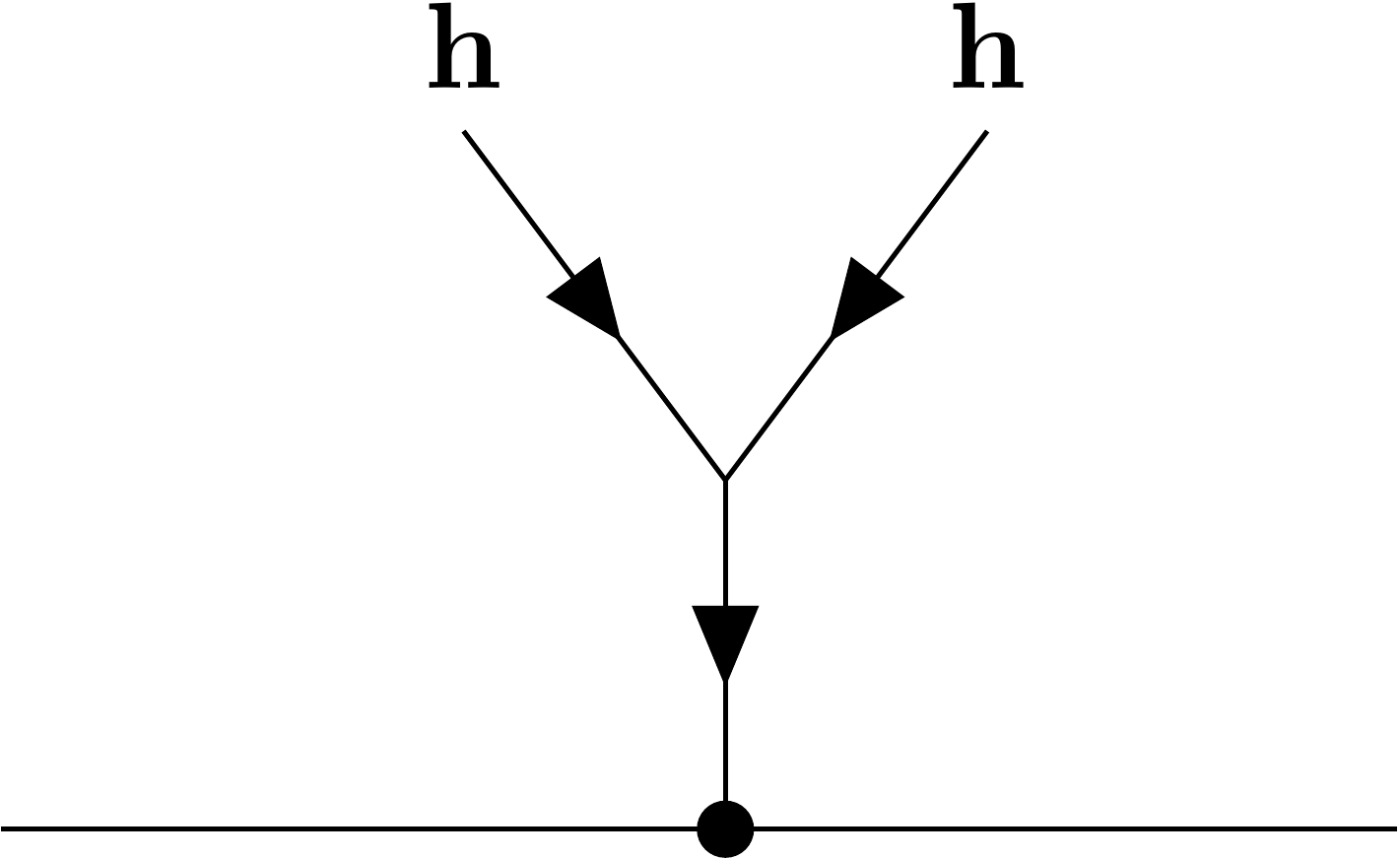}
	\end{subfigure}
	\begin{subfigure}[t]{0.49\textwidth}
		\centering
		\includegraphics[scale=0.3]{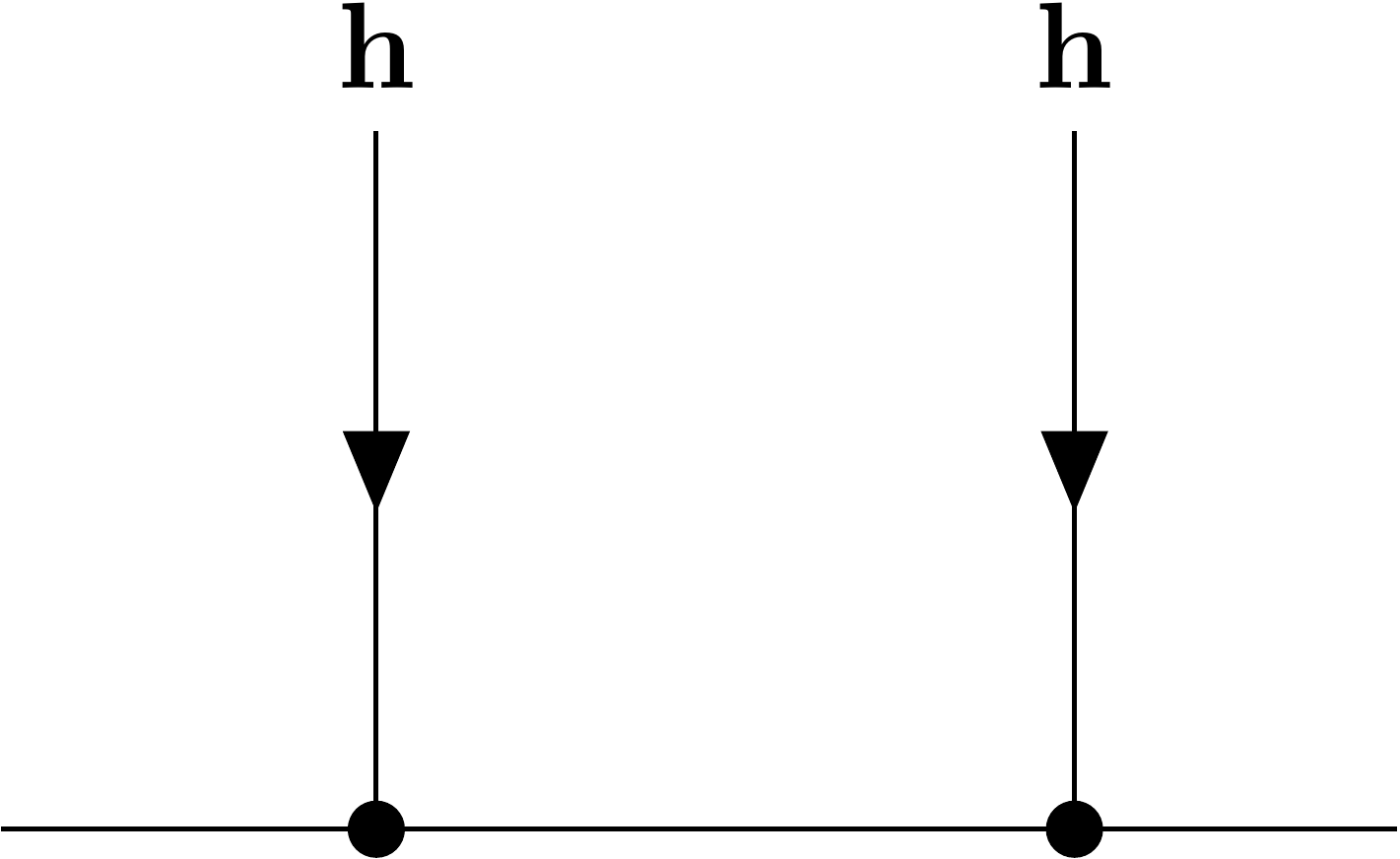}
	\end{subfigure}
	\caption{\emph{Feynman diagrams whose BRST variation determines the classical $w,w$ OPEs. The horizontal line represents the defect supported on $\cL_0$, and the solid dots represent couplings to $w$}.}
	\label{fig:BRST-tree}
\end{figure}

Working classically, we can also directly take the BRST variation of the defect, but it's useful to keep the diagrammatic representation in mind as we do so. Ignoring the $\tilde\bh$ states for the moment, the classical BRST variation of the defect coupling \eqref{eq:defect-gh} is
\bea \label{eq:BRST-tree-nonlinear}
&\delta\Bigg(\sum_{m,n\in\bbZ_{\geq0}}\frac{1}{2\pi\im}\int_{\cL_0}\dif z\,w[m,n](z)\cD_{m,n}\bh\Bigg) \\
&= -\frac{1}{2\pi\im}\sum_{m,n\in\bbZ_{\geq0}}\sum_{p=0}^{m+1}\sum_{q=0}^{n+1}\int_{\cL_0}\dif z\,p(n-q+1)w[m,n](z)\cD_{p,q}\bh\cD_{m+1-p,n+1-q}\bh \\
&= -\frac{1}{4\pi\im}\sum_{p,q,r,s\in\bbZ_{\geq0}}\int_{\cL_0}\dif z\,(ps-qr)w[p+r-1,q+s-1](z)\cD_{p,q}\bh\cD_{r,s}\bh\,.
\eea
In going from the penultimate to final line we've used the fact that $\bh$ is graded odd. The expression above can readily be identified with the BRST variation of the first diagram in figure \ref{fig:BRST-tree}.

This is cancelled by the BRST variation of the second diagram in figure \ref{fig:BRST-tree}, which represents a bilocal term on the defect
\be \delta\Bigg(\frac{1}{2(2\pi\im)^2}\sum_{p,q,r,s\in\bbZ_{\geq0}}\int_{\cL_{0,1}\times\cL_{0,2}}\dif z_1\,w[p,q](z_1)\cD_{p,q}\bh_1\dif z_2\,w[r,s](z_2)\cD_{r,s}\bh_2\Bigg)\,.
\ee
The factor of $1/2$ is present because this term is obtained by expanding an exponential to quadratic order in the path integral. Taking the linearised part of the BRST variation here gives
\be \label{eq:BRST-tree-bilocal} \frac{1}{2(2\pi\im)^2}\sum_{p,q,r,s\geq0}\int_{\cL_{0,2}}\bigg(\lim_{\epsilon\to0}\oint_{|z_{12}|=\epsilon}\dif z_{12}\,w[p,q](z_1)w[r,s](z_2)\bigg)\dif z_2\,\cD_{p,q}\bh_2\cD_{r,s}\bh_2\,. \ee
(The non-linear BRST variation cancels against trilocal terms on the defect.) Comparing \eqref{eq:BRST-tree-nonlinear} to \eqref{eq:BRST-tree-bilocal} we see that they cancel precisely if
\be \frac{1}{2\pi\im}\lim_{\epsilon\to0}\oint_{|z_{12}|=\epsilon}\dif z_{12}\,w[p,q](z_1)w[r,s](z_2) = (ps-qr)w[p+r-1,q+s-1](z_2)\,, \ee
which holds when
\be w[p,q](z)w[r,s](0)\sim\frac{ps-qr}{z}w[p+r-1,q+s-1](0)\,. \ee
By an identical argument we recover the remaining OPEs in \eqref{eq:chiral-algebra}.

The charges in table \ref{tab:classical-generators} can be read off from the defect coupling \eqref{eq:defect-gh}. In particular, dilations on spacetime lift to twistor space as scaling of the $v^\da$ coordinates at fixed $z$. On the other hand, the difference $s-d/2$ gives the eigenvalue under scaling $z$ at fixed $v^\da$.\\

The advantage of viewing the chiral algebra as being the universal surface defect supported on a real twistor line is that it can readily be generalized to loop level. Indeed, in section \ref{sec:defect} we will identify a 1-loop diagram in the coupled bulk-defect system whose BRST variation is non-vanishing, necessitating corrections to the OPEs. We will argue that these can be identified with 1-loop corrections to the holomorphic collinear singularities of amplitudes in 1\textsuperscript{st}-order deformations of SD gravity.


\subsection{Twistorial gravitational anomaly}
\label{subsec:quantum-SDGR}

Thus far our discussion has been completely classical. In this subsection we review the results of \cite{Bittleston:2022nfr} concerning the twistor uplift of SD gravity, which we've seen is Poisson-BF theory, as a quantum field theory.\\

Since Poisson-BF theory is 1-loop exact and holomorphic, it's finite \cite{Williams:2018ows}. However, it can in principle suffer from an anomaly. The dangerous diagram is the box illustrated in figure \ref{fig:PBF-anomaly}.

\begin{figure}[h!]
\centering
  \includegraphics[scale=0.3]{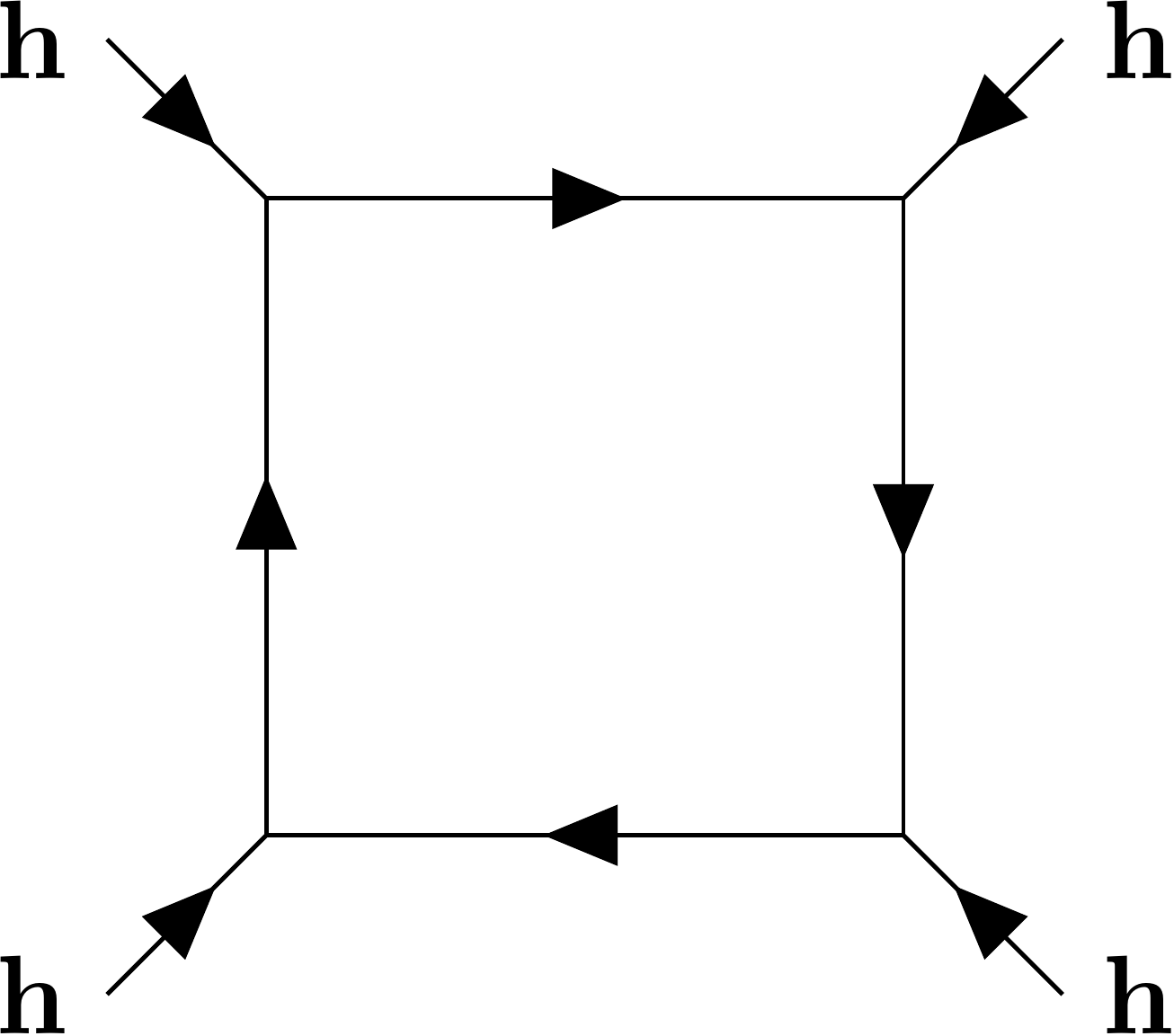}
\caption{\emph{This diagram fails to be BRST invariant in Poisson-BF theory.}} \label{fig:PBF-anomaly}
\end{figure}

In \cite{Bittleston:2022nfr} the BRST variation of this diagram was evaluated using both the family index theorem and by a direct perturbative calculation, and found to be
\be \label{eq:PBF-anomaly} \cA = \frac{1}{4\cdot 5!}\bigg(\frac{\im}{2\pi}\bigg)^3\int_\bbPT\p^\dd\p_\da\bh\,\p^\da\p_\db\p\bh\,\p^\db\p_\dc\p\bh\,\p^\dc\p_\dd\p\bh\,. \ee
This can be more directly understood by introducing
\be \bs^\da_{~\,\db} = - \p^\da\p_\db\bh\,, \ee
which we view as a $(0,\bullet)$-polyform with values in $\End\,\cN$. Under a BRST transformation
\be \delta\bs = \nbar\bs + \frac{1}{2}[\bs,\bs] = \bar\p\bs + \{\bh,\bs\} + \frac{1}{2}[\bs,\bs]\,, \ee
so that $\bs$ is an $\fsl_2(\bbC)$ connection on $\cN$ over the deformed spacetime $\cPT$. In terms of $\bs$ the anomaly cocycle \eqref{eq:PBF-anomaly} is
\be \label{eq:PBF-anomaly-s} \cA = \frac{1}{4\cdot 5!}\bigg(\frac{\im}{2\pi}\bigg)^3\int_\bbPT\tr(\bs\p\bs^3)\,. \ee
This takes a similar form to the anomaly cocycle in holomorphic BF theory for the gauge algebra $\fsl_2(\bbC)$ \cite{Costello:2021bah,Costello:2019jsy}, although the overall coefficient differs.

Following \cite{Costello:2021bah}, we do not interpret this as an anomaly in SD gravity. Indeed, we expect that it can be cancelled by a counterterm which is non-local on $\bbPT$, but which is local on spacetime. Instead it should be viewed as obstructing the integrability of SD gravity. The spacetime counterpart of the box diagram in figure \ref{fig:PBF-anomaly} is the 1-loop 4-point all-plus amplitude, consistent with the proposal of Bardeen that the all-plus amplitudes break integrability \cite{Bardeen:1995gk}.\\

Having identified an anomaly, it's natural to seek a means of cancelling it. A particularly simple method is by coupling to appropriate matter. In general, the twistorial gravitational anomaly is proportional to the difference between the number of bosonic and fermionic degrees of freedom in the theory, so there are many different ways of achieving this. Certainly the anomaly vanishes in theories of SD supergravity, but probably the simplest choice is to minimally couple to a single Weyl fermion on spacetime. The disadvantage of this approach is that it cancels one quantum effect against another, and so cannot be used to compare loop computations against easier tree calculations.

A second method identified in \cite{Bittleston:2022nfr} involves coupling to a field $\bseta\in\Omega^{2,\bullet}(\bbPT)[1]$ obeying $\p\bseta=0$. (We will sometimes write this as $\bseta\in\Omega^{2,\bullet}_\mathrm{cl}(\bbPT)[1]$.) This is motivated by the case of SD Yang-Mills, where a kind of Green-Schwarz mechanism allows for anomaly cancellation \cite{Costello:2021bah}. A variant of this applies in the gravitational case, though it is somewhat more complicated. The kinetic term for $\bseta$ and its classical coupling to $\bh$ is
\be \label{eq:S-eta} S_\mathrm{kin.}[\bseta;\bh] = \frac{1}{4\pi\im}\int_\bbPT\p^{-1}\bseta\,\nbar\bseta = \frac{1}{4\pi\im}\int_\bbPT\p^{-1}\bseta\,\bar\p\bseta + \bseta\,\{\bh,\ \}\ip\bseta\,, \ee
where
\be \nbar \bseta = \bar\p\bseta + \cL_{\{\bh,\ \}}\bseta = \bar\p\bseta - \p(\{\bh,\ \}\ip\bseta)\,. \ee
Here it's essential that we use the $(1,0)$ Lie derivative $\cL_V = [V\ip\,,\p\,]$, and in the second equality we've employed the identity $\p\bseta=0$. The Lie derivative term has no counterpart in the gauge theory case, but is essential to ensure that together \eqref{eq:S-eta} and the Poisson-BF action \eqref{eq:S-PBF-BV} obey the classical master equation.

So far, this gets us no closer to cancelling the anomaly. Indeed, it makes the situation worse. Now the field $\bseta$ can run through the loop in figure \ref{fig:PBF-anomaly}, doubling the coefficient of the anomaly cocycle. This is because $\bseta$ carries twice the degrees of freedom of $\bh$, but its propagator is undirected so the anomaly diagram is accompanied by an extra $1/2$ symmetry factor.

However, we can also introduce a further coupling between $\bh$ and $\bseta$ taking the form
\be S_\mathrm{count.}[\bseta;\bh] = \frac{\mu}{4\pi\im}\int_\bbPT\bseta\,\p^\db\p_\da\bh\,\p^\da\p_\db\p\bh = \frac{\mu}{4\pi\im}\int_\bbPT\bseta\,\tr(\bs\p\bs)\,.
\ee
We refer to this as the Poisson-BF counterterm, since we will use it to cancel the Poisson-BF anomaly. The full action $S = S_\mathrm{PBF} + S_\mathrm{kin.} + S_\mathrm{count.}$ is not invariant under the modified BRST transformations
\bea
&\delta\bh = T(\bh) = \bar\p\bh + \frac{1}{2}\{\bh,\bh\}\,, \\
&\delta\bg = \nbar\bg + \frac{1}{2}\cL_{\p^\da}(\bseta\,\p_\da\ip\bseta) - \mu\,\cL_{\p^\db}\cL_{\p_\da}\bseta\,\p^\da\p_\db\p\bh \,, \\
&\delta\bseta = \nbar\bseta - \frac{\mu}{2}\tr(\p\bs^2)\,.
\eea
Instead we find that
\be \delta S = - \frac{\mu^2}{8\pi\im}\int_{\bbPT}\tr(\bs\p\bs)\,\tr(\p\bs^2)\,. \ee
By tuning the value of $\mu$ we can use this failure to precisely cancel the Poisson-BF anomaly \eqref{eq:PBF-anomaly-s}. This relies on the following trace identity for the fundamental of $\fsl_2(\bbC)$
\be \tr(X^4) = \frac{1}{2}\tr(X^2)^2\,, \ee
and the required value of $\mu$ is
\be \mu^2 = \frac{1}{5!}\bigg(\frac{\im}{2\pi}\bigg)^2 \,. \ee
The advantage of this approach is that it involves cancelling a 1-loop effect against a tree level counterterm. In \cite{Costello:2022upu} this same feature was exploited in the gauge theory case to determine 1-loop corrections to the extended CCA without performing explicit diagram computations. Unfortunately the covariant $\bseta,\bh$ coupling makes this more challenging in the gravitational case, though we will see in section \ref{sec:cure} that it can still be leveraged to constrain the chiral algebra.\\

In full, the twistor action is
\be \label{eq:S-full} S[\bseta;\bg,\bh] = \frac{1}{2\pi\im}\int_\bbPT\bg\,T(\bh) + \frac{1}{4\pi\im}\int_\bbPT\big(\p^{-1}\bseta\,\nbar\bseta + \mu\,\bseta\,\tr(\bs\p\bs)\big)\,. \ee
It was shown in \cite{Mason:FAv1} that the field $\bseta$ descends to a 4\textsuperscript{th}-order scalar on spacetime, which at the linearised level can be recovered as
\be \rho = \frac{1}{2\pi\im}\int_\bbPT\p^{-1}\bseta\,. \ee
The freedom to add $\p$-exact terms to $\p^{-1}\bseta$ leads to an ambiguity $\rho\sim\rho + C$ for $C$ constant, so the spacetime theory can only depend on $\rho$ through $\dif\rho$.

It was further argued in \cite{Bittleston:2022nfr} that, upon restriction to the physical fields, the twistor action \eqref{eq:S-full} descends to the following action on spacetime
\be \label{eq:SDGR+axion} S_{\mathrm{SDGR}+\rho}[\rho;\Gamma,e] = S_\mathrm{SDGR}[\Gamma,e] + \int_{\bbR^4}\bigg(\vol_g\,\frac{1}{2}(\Delta_g\rho)^2 + \frac{\mu}{\sqrt{2}}\,\rho \,R^\mu_{~\,\nu}\wedge R^\nu_{~\,\mu}\bigg)\,. \ee
Here $R^\mu_{~\nu}$ is the Riemann curvature tensor, viewed as a 2-form with values in endomorphisms of the tangent bundle. $R^\mu_{~\,\nu}\wedge R^\nu_{~\,\mu}$ is the Pontryagin class, revealing $\rho$ to be a kind of gravitational axion.\footnote{On the support of the classical equations of motion the Pontryagin and Euler classes coincide, so this coupling is unambiguous modulo field redefinitions.} Note that this coupling respects the redundancy $\rho\sim\rho+C$.

As is well known, the only non-vanishing amplitudes in SD gravity, aside from the 3-point tree, are the 1-loop all-plus amplitudes. We've already seen that the anomalous box diagram on twistor space can be identified with the 1-loop 4-point all-plus amplitude on spacetime. It should therefore come as no surprise that by cancelling the anomaly we eliminate this amplitude. It was further argued in \cite{Bittleston:2022nfr} that the only non-vanishing amplitudes of the theory \eqref{eq:SDGR+axion} are 3-point trees, with the 1-loop all-plus amplitudes all being cancelled by axion exchange.


\section{1-loop holomorphic splitting amplitudes in gravity} \label{sec:splitting}

In this section we massage the effective 1-loop graviton vertex of \cite{Brandhuber:2007up} into a 1-loop holomorphic splitting amplitude. It characterises the physical double pole arising in the holomorphic collinear limit of a generic 1-loop graviton amplitude. We argue that it deforms the extended CCA at 1-loop, and then go on to compute the resulting corrections to operator products.


\subsection{Holomorphic collinear limit} \label{subsec:collinear-limit}

First, let us explain what we mean by the holomorphic collinear limit, which appears in \cite{Ball:2021tmb}.\\

Consider an $n$-point graviton amplitude $\cM_n(\{p_i^{J_i}\}_{i=1}^n)$ where $p_i,J_i$ are the (complexified) momenta and helicities of the external gravitons. Since gravity is a massless theory the $p_i$ are null, and we can always decompose
\be p_i^{\alpha\da} = \lambda_i^\alpha{\tilde\lambda}_i^\da\,, \ee
where we're implicitly using the Van der Waerden symbols. $\lambda_i,{\tilde\lambda}_i$ are known as spinor-helicity variables. It's standard to express amplitudes with definite helicity configurations in terms of the $\gSL_2(\bbC)_+\times\gSL_2(\bbC)_-$ invariant combinations
\be \la ij\ra = \sqrt{2}\epsilon_{\alpha\beta}\lambda_i^\alpha\lambda_j^\beta\,,\qquad [ij] = -\sqrt{2}\epsilon_{\da\db}{\tilde\lambda}_i^\da{\tilde\lambda}_j^\db\,. \ee
These have been normalized so that $(p_i + p_j)^2 = 2p_i\cdot p_j = \la ij\ra[ji]$, matching standard conventions.\\

In the true collinear limit of an amplitude we specify two external momenta, without loss of generality $p_1,p_2$, and take
\be p_1\to tP_{12}\,,\qquad p_2\to (1-t)P_{12} \ee
for $P_{12} = p_1 + p_2$ and some parameter $t\in\bbC$. (For real momenta in Lorentzian or ultrahyperbolic signature we have $t\in\bbR$.) In terms of the spinor-helicity variables
\bea \label{eq:true}
&\lambda_1\to \sqrt{2t}\lambda_{P_{12}}\,,\qquad \lambda_2\to\sqrt{2(1-t)}\lambda_{P_{12}}\,, \\ 
&{\tilde\lambda}_1\to\sqrt{\frac{t}{2}}{\tilde\lambda}_{P_{12}}\,,\qquad {\tilde\lambda}_2\to\sqrt{\frac{(1-t)}{2}}{\tilde\lambda}_{P_{12}}\,,
\eea
so that $\la12\ra,[12]\to0$ separately, and importantly these limits are taken at the same rate. (Our slightly unconventional normalizations here have been chosen for later convenience.) Usually the true limit is taken in Lorentzian signature, where $[12] \propto \overline{\la12\ra}$. In the true collinear limit a tree graviton amplitude factorizes as
\be \cM^\mathrm{tree}_n(\{p_i^{J_i}\}_{i=1}^n) \to \mathrm{Split}^\mathrm{tree}_\pm(p_1^{J_1},p_2^{J_2};t)\,\cM^\mathrm{tree}_{n-1}(P_{12}^\mp,\{p_i^{J_i}\}_{i=3}^n) \ee
for the tree graviton splitting amplitudes
\be \label{eq:tree-splitting} \mathrm{Split}^\mathrm{tree}_-(1^+,2^+;t) = - \frac{1}{4t(1-t)}\frac{[12]}{\la12\ra}\,,\qquad \mathrm{Split}^\mathrm{tree}_+(1^-,2^+;t) = - \frac{4t^3}{(1-t)}\frac{[12]}{\la12\ra} \ee
originally computed in \cite{Bern:1998xc}.\footnote{In this work we strip the coupling constant $\kappa^2=32\pi G_\mathrm{Newton}$ from amplitudes. It can easily be reintroduced: $n$-point trees are accompanied by $(\kappa/2)^{n-2}$ and $n$-point loops by $(\kappa/2)^n$.} In the true limit these are not actually singular, but acquire a phase as $p_1,p_2$ rotate around each other at fixed $P_{12}$. It's a remarkable result of \cite{Bern:1998sv} that a 1-loop graviton amplitude also decomposes as
\be
\cM^\mathrm{1-loop}_n(\{p_i^{J_i}\}_{i=1}^n) \to \mathrm{Split}^\mathrm{tree}_\pm(p_1^{J_1},p_2^{J_2};t)\,\cM^\mathrm{1-loop}_{n-1}(P_{12}^\mp,\{p_i^{J_i}\}_{i=3}^n)\,.
\ee
In other words, there are no 1-loop splitting amplitudes in gravity, at least in the true limit.\\

In the holomorphic collinear limit we again specify two external momenta, $p_1,p_2$, but instead take only
\be \label{eq:lambda-P} \lambda_1\to\sqrt{2t}\lambda_{P_{12}}\,,\qquad \lambda_2\to\sqrt{2(1-t)}\lambda_{P_{12}}\,, \ee
so that $\la12\ra\to0$, whilst holding $[12]$ fixed. Note that in this limit the momenta $p_1,p_2$ do not become collinear, but instead span a totally null 2-plane. Only the spinor-helicity variables $\lambda_1,\lambda_2$ become linearly dependent. In Lorentzian signature the reality condition $[12]\propto\overline{\la12\ra}$ prevents us from taking the holomorphic limit, but there is no such issue in the complexified or ultrahyperbolic setting. Let 
\be \label{eq:mu-P} {\tilde\lambda}_{P_{12}} = \sqrt{2t}{\tilde\lambda}_1 + \sqrt{2(1-t)}{\tilde\lambda}_2\,, \ee
so that in the limit $P^{\alpha\da}_{12} = \lambda_{P_{12}}^\alpha{\tilde\lambda}_{P_{12}}^\da$. Our conventions have been chosen so that for $t=1/2$ we have ${\tilde\lambda}_{P_{12}} = {\tilde\lambda}_1 + {\tilde\lambda}_2$, and to ensure consistency with the true limit \eqref{eq:true}.

In the next subsection we'll see that as the momenta $p_1,p_2$ of two positive helicity gravitons become holomorphically collinear, a generic 1-loop graviton amplitude acquires a double pole
\be
\cM^\mathrm{1-loop}_n(p_1^+,p_2^+,\{p_i^{J_i}\}_{i=2}^n) \sim \mathrm{Split}^\mathrm{1-loop}_+(p_1^+,p_2^+;t)\,\cM^\mathrm{tree}_{n-1}(P_{12}^-,\{p_i^{J_i}\}_{i=3}^n)\,,
\ee
encoded in a `1-loop holomorphic splitting amplitude' $\mathrm{Split}^\mathrm{1-loop}_+(p_1^+,p_2^+;t)$ which can be identified with the 1-loop effective vertex. In order for this to be consistent with \cite{Bern:1998sv}, in the true limit we must have
\be \mathrm{Split}^\mathrm{1-loop}_+(p_1^+,p_2^+;t)\to 0\,. \ee


\subsection{1-loop holomorphic splitting amplitude} \label{subsec:effective-vertex}

The holomorphic collinear singularities in a generic 1-loop graviton amplitude arise either from an internal propagator going on-shell, or from the loop integration. As the momenta $p_1,p_2$ of two positive helicity massless particles become homomorphically collinear, it's possible for both of these sources to contribute, leading to a double pole. The relevant diagram is illustrated in figure \ref{fig:splitting-again} \cite{Brandhuber:2007up,Dunbar:2010xk,Dunbar:2011xw,Alston:2012xd}.

\begin{figure}[h!]
	\centering
	\includegraphics[scale=0.3]{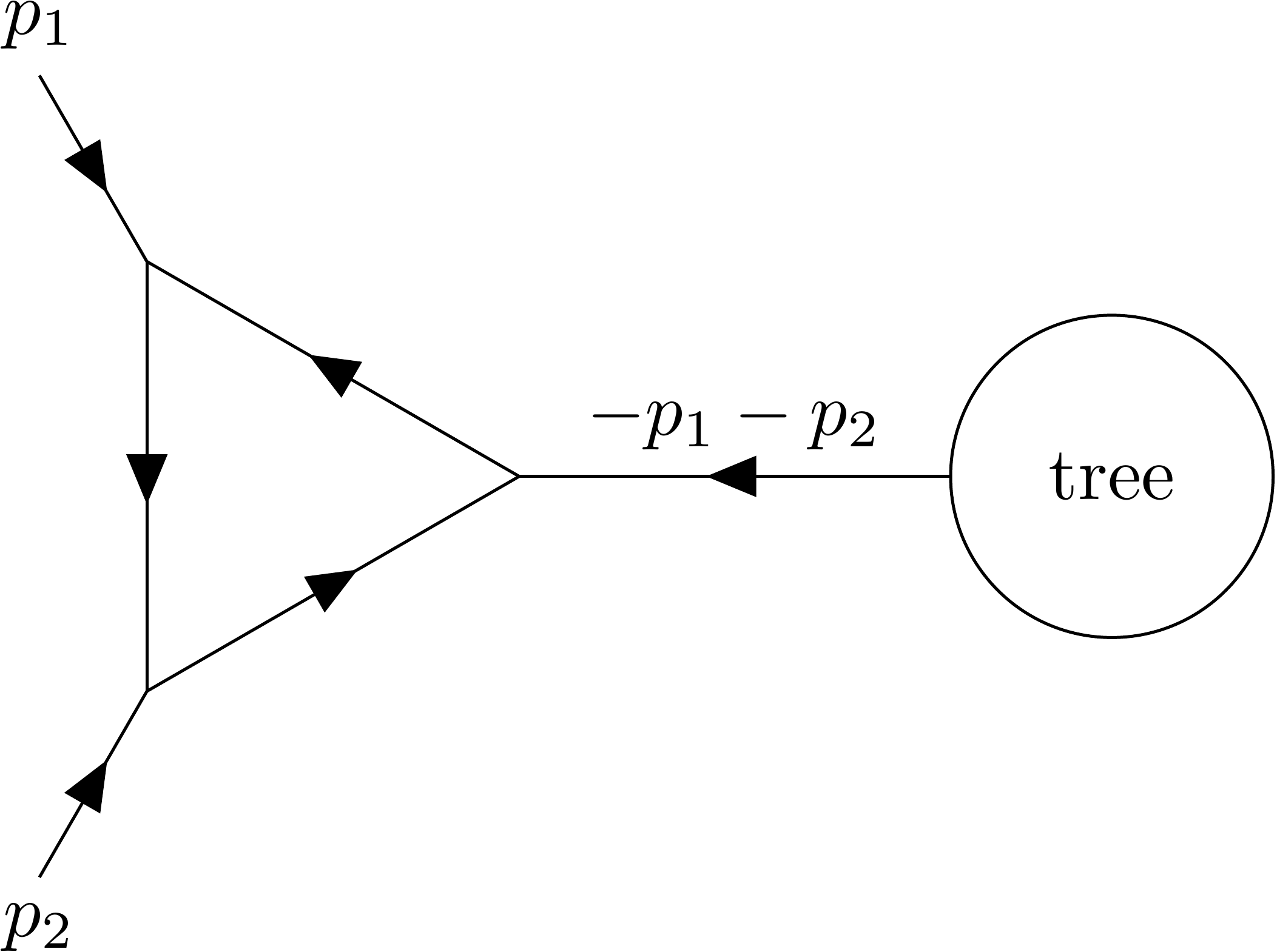}
	\caption{\emph{This diagram is responsible for double poles in 1-loop amplitudes in the holomorphic collinear limit. Arrows indicate the flow of momenta and helicity through the diagram.}} \label{fig:splitting-again}
\end{figure}

Here the internal propagator with momentum $-p_1-p_2$ is off-shell. This is why it gives a non-vanishing result, whereas the 1-loop 3-point all-plus amplitude vanishes.

The only non-vanishing amplitudes in SD gravity, aside from the 3-point trees, are the 1-loop all-plus amplitudes. These do not acquire double poles in their holomorphic collinear limits: although the off-shell triangle in figure \ref{fig:splitting-again} involves only vertices present in the SD theory, the 1-minus trees all vanish.

Instead, double poles first arise in the rational 1-loop mostly-plus amplitudes. For more general helicity configurations, 1-loop graviton amplitudes have discontinuities at branch cuts. These can always be collected in a set of box, triangle and bubble integrals with rational coefficients, which themselves can be determined by generalized unitarity. The remaining rational part of the amplitude then has unphysical poles compensating those in the cut constructible piece \cite{Bern:1994cg,Dunbar:2011xw}. This complicated cancellation can obscure the physical singularities, and for this reason the rational 1-loop mostly-plus amplitudes are those in which the double pole is most clearly identifiable.\\

The factorization of amplitudes at complex kinematic points can be exploited to reconstruct them recursively. Indeed, this is essentially the basis of BCFW recursion for Yang-Mills tree amplitudes \cite{Britto:2005fq}, which can be extended to the case of gravity.\footnote{It is notable that the MHV rules for gravity \cite{Bjerrum-Bohr:2005xoa,Mason:2009afn}, which can be obtained from a particular BCFW shift of the negative helicity momentum variables, fail for NMHV amplitudes at 12-points and above \cite{Benincasa:2007qj,Bianchi:2008pu,Conde:2012ik} due to the presence of poles at infinity.}

Double poles complicate the application of recursion methods to 1-loop amplitudes, because they hide poorly understood simple poles. However, in this work we will be most interested the double poles themselves. These have been characterised using a 1-loop effective graviton vertex describing the (off-shell) 1-loop 3-point all-plus amplitude \cite{Brandhuber:2007up} (see also \cite{Dunbar:2010xk,Alston:2012xd})
\be \label{eq:effective-graviton-vertex} \cM_3^\mathrm{1-loop}(1^+,2^+,3^+) = - \frac{\im}{180(4\pi)^2}\frac{[12]^2[23]^2[31]^2}{P_{12}^2} = \frac{\im}{180(4\pi)^2}\frac{[12][23]^2[31]^2}{\la12\ra}\,. \ee
The simple pole is generated by the loop integration. This vertex determines a `1-loop holomorphic splitting amplitude'
\be \label{eq:holomorphic-splitting} \Split^\mathrm{1-loop}_+(1^+,2^+;t) = \frac{\im V^\mathrm{1-loop}_3(1^+,2^+,P_{12}^+)}{P_{12}^2} = \frac{4t(1-t)}{180(4\pi)^2}\frac{[12]^4}{\la12\ra^2}\,, \ee
where we've included the propagator connecting the triangle and tree in figure \ref{fig:splitting-again}. As expected, it vanishes in the true limit. For completeness, we show explicitly in appendix \ref{app:5pt-amplitude} that this splitting amplitude arises in the holomorphic collinear limit of the 5-point 1-loop mostly-plus amplitude. We will explain shortly why this splitting amplitude deforms the extended celestial algebra.\\

Although this work is largely concerned with the double poles in the 1-loop corrected extended CCA, we will also find simple poles which are bilinear in generators. We expect these are related to the simple poles beneath double poles, which are responsible for much of the difficulty in applying recursion techniques to 1-loop amplitudes. In the case of the 1-loop mostly plus amplitudes these subleading simple poles have been computed in \cite{Dunbar:2010xk,Alston:2015gea}.\\

In subsection \ref{subsec:splitting-deformation}, we compute the 1-loop corrections to operator products of $\cU$ induced by the 1-loop holomorphic splitting amplitude \eqref{eq:holomorphic-splitting}. Before doing so, let's address why the deformation occurs.


\subsection{Why is the extended celestial chiral algebra deformed?} \label{subsec:interpretation}

As shown in \cite{Ball:2021tmb}, only tree splitting amplitudes contribute in the holomorphic collinear limits of 1-loop all-plus graviton amplitudes. The CCA of SD gravity is therefore undeformed by quantum corrections.\\

Why then, does the 1-loop holomorphic splitting amplitude \eqref{eq:holomorphic-splitting} deform the extended CCA? To answer this question, we first need to understand which collinear singularities the extended CCA is describing.\footnote{We would particularly like to thank Kevin Costello for many useful discussions on this point.}\\

We can argue by analogy with the case of SD Yang-Mills, as analysed in detail in \cite{Costello:2022wso,Costello:2022upu}. It was proven in \cite{Costello:2022wso} that the operator products of the classical extended CCA of SD Yang-Mills describe the universal holomorphic collinear singularities of tree form factors, i.e., tree scattering amplitudes in the presence of local operators. Here by universal we mean that the singularities are independent of the choice of local operator. In a general field theory the holomorphic collinear singularities of form factors will certainly not be universal, let's see why this is the case at tree level in SD Yang-Mills.

The collinear singularities of $n$-point tree form factors arise from internal propagators going on-shell. They can be attributed to factorizing diagrams in which the momenta becoming collinear are attached to a 3-point vertex connected by a singular propagator to an $(n-1)$-point tree, as illustrated in figure \ref{fig:tree-factorization}. However, SD Yang-Mills has only trivial tree amplitudes, so in any such diagram the $(n-1)$-point tree must involve the local operator. The holomorphic collinear singularities are therefore universal. Since the classical OPEs of the extended CCA simply encode the tree vertices present in the action, they characterise this behaviour.

\begin{figure}[h!]
	\centering
	\begin{subfigure}[t]{0.49\textwidth}
		\centering
		\includegraphics[scale=0.3]{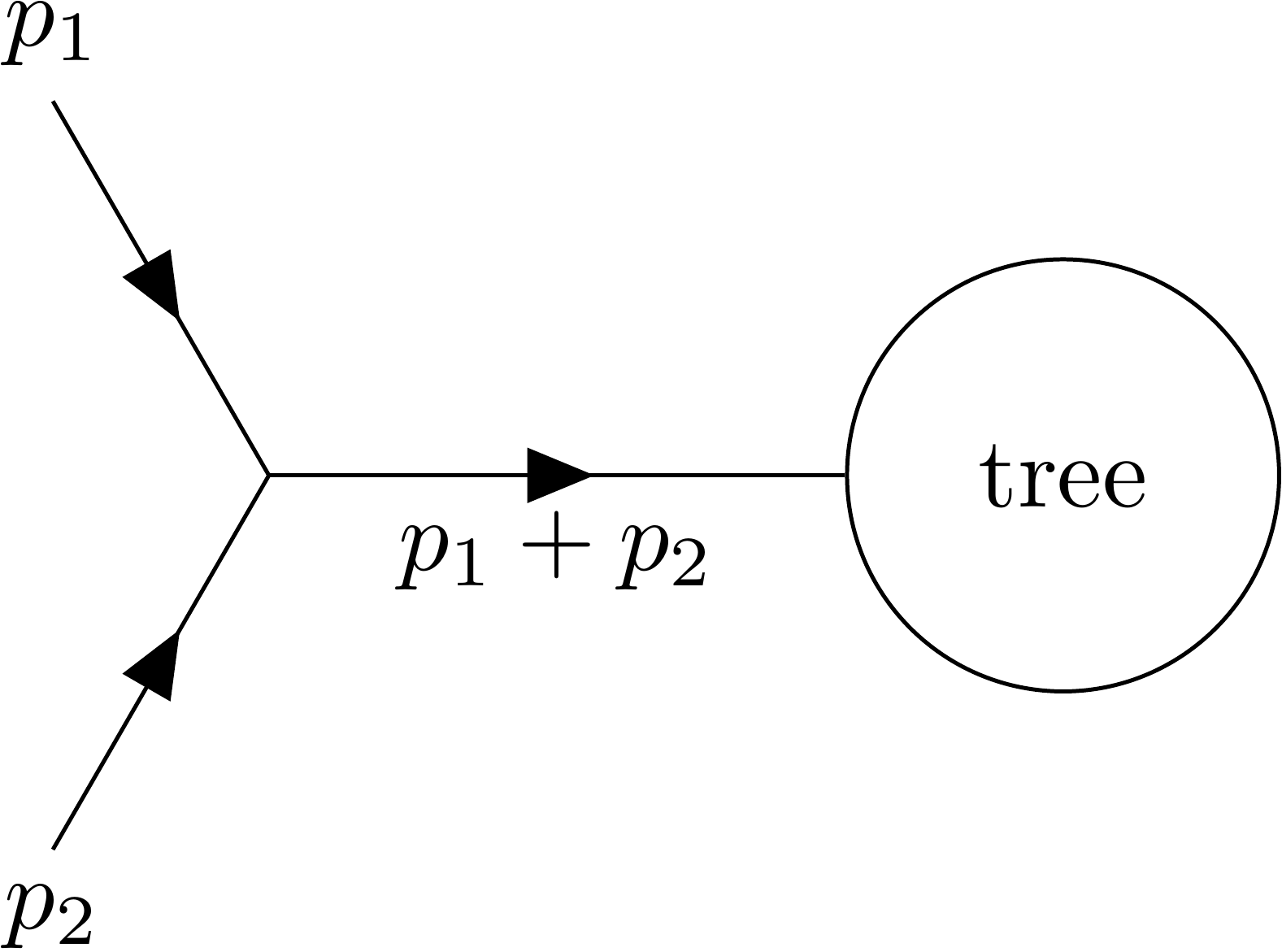}
		\label{subfig:plusplustoplus}
	\end{subfigure}
	\begin{subfigure}[t]{0.49\textwidth}
		\centering
		\includegraphics[scale=0.3]{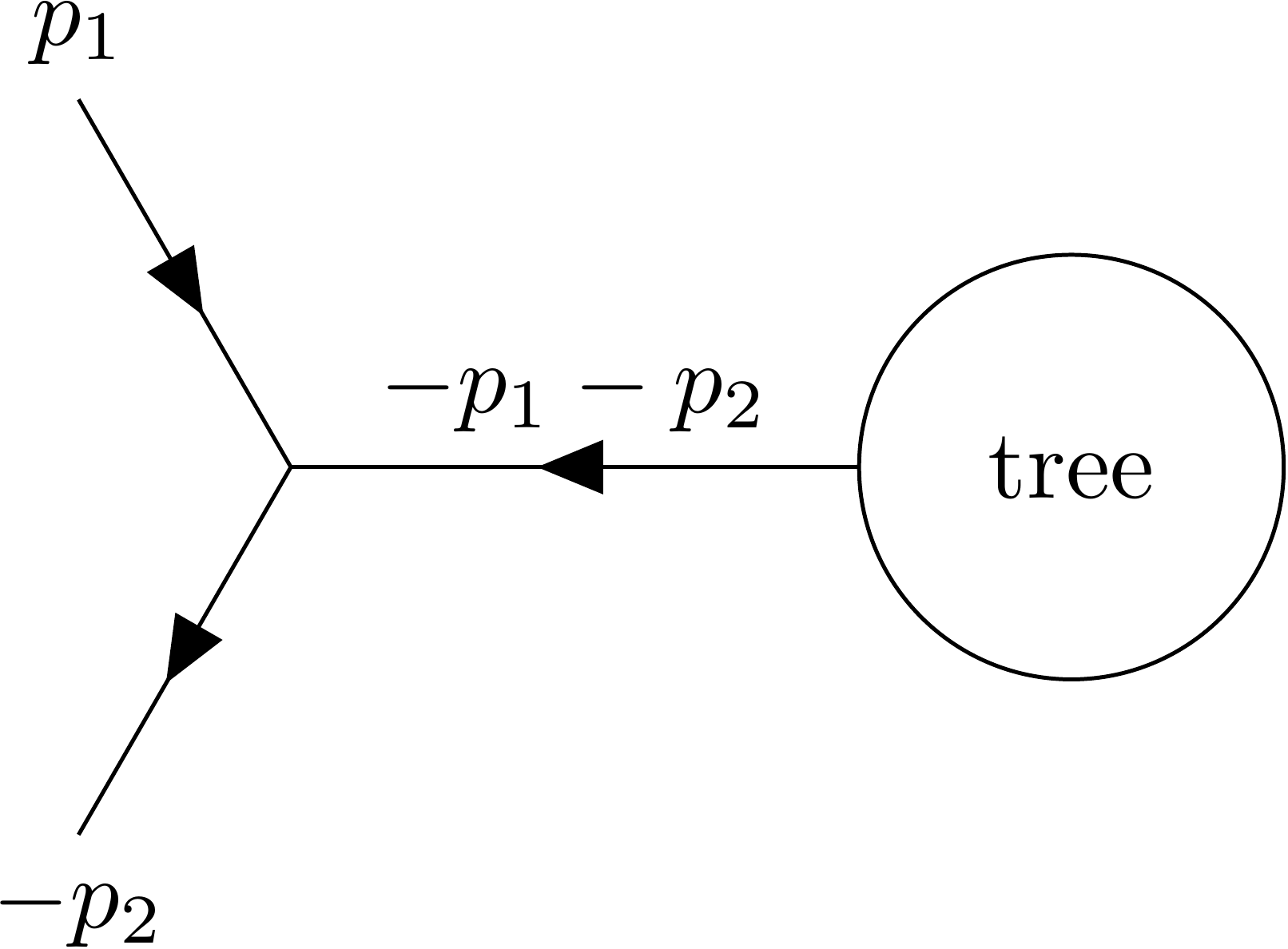}
		\label{subfig:plusminustominus}
	\end{subfigure}
	\caption{\emph{These diagrams are responsible for the holomorphic collinear singularities of tree form factors in SD Yang-Mills, and tree amplitudes in 1\textsuperscript{st}-order deformations of SD gravity. Since the tree amplitudes in both SD Yang-Mills and SD gravity are trivial, the blob labelled `tree' must involve the local operator or deformation. The collinear singularities are therefore universal.}}
	\label{fig:tree-factorization}
\end{figure}

Full Yang-Mills can be written as a perturbation around its SD sector, so the integrals of certain tree form factors constitute tree amplitudes in the non-SD theory. This explains why the classical extended CCA describes the holomorphic collinear singularities of tree gluon amplitudes which do not arise in SD Yang-Mills.\\

In SD gravity the above interpretation of the classical extended CCA is complicated by the fact that there are no BRST invariant local operators, at least of ghost number zero. We can instead consider 4-form local operators whose BRST variation is de Rham exact, which are related to BRST invariant local operators of higher ghost number by anomaly ascent. (These can also be thought of as local operators with BRST variation a total derivative.) Such an operator, $\cO$ say, determines a 1\textsuperscript{st}-order deformation of SD gravity
\be S_\mathrm{SDGR}[\Gamma,e] \mapsto S_\mathrm{SDGR}[\Gamma,e] + \epsilon\int_{\bbR^4}\cO\,. \ee
Expanding to 1\textsuperscript{st}-order in the infinitesimal parameter $\epsilon$ inside the path integral, the operator $\cO$ plays the role of an interaction vertex which appears in Feynman diagrams once. The operator products in the classical extended CCA of SD gravity describe the universal holomorphic collinear limits of tree amplitudes in the presence of such 1\textsuperscript{st}-order deformations. The universality of this behaviour again rests on the vanishing of all non-trivial tree amplitudes. Full Einstein gravity can be written as a perturbation around its SD sector, so the classical extended CCA describes the holomorphic collinear singularities of tree graviton amplitudes which are not present in the SD theory.\\

In \cite{Costello:2022upu}, the authors explored whether the collinear singularities of form factors in quantum SD Yang-Mills define a consistent chiral algebra. By analogy, we should seek a quantum deformation of the classical extended CCA whose OPEs encode the holomorphic collinear limits of amplitudes in 1\textsuperscript{st}-order deformations of quantum SD gravity. Diagrams of the form illustrated in figure \ref{fig:splitting-again} can certainly contribute in these limits, where the blob labelled `tree' now represents an amplitude in the 1\textsuperscript{st}-order deformation. The 1-loop effective vertex therefore deforms the extended classical CCA.

The non-vanishing of the 1-loop all-plus amplitudes in SD gravity mean that the holomorphic collinear singularities of amplitudes in its 1\textsuperscript{st}-order deformations are not universal. We will ultimately find that these 1-loop amplitudes are the source of the associativity failure in the quantum deformation of the extended CCA. This point will be addressed in detail in section \ref{sec:associativity}.\\

Motivated by the results of \cite{Costello:2022wso}, in section \ref{sec:defect} we will argue that the extended CCA has a concrete realisation as the universal holomorphic surface defect in the twistor formulation of SD gravity. We will see that an anomalous diagram in the coupled bulk-defect system necessitates precisely the corrections to OPEs induced by the 1-loop holomorphic splitting amplitude, together with subleading corrections.


\subsection{1-loop corrections to operator products from the splitting amplitude} \label{subsec:splitting-deformation}

In this subsection we determine the corrections to the OPEs of the classical extended CCA induced by the 1-loop holomorphic splitting amplitude from equation \eqref{eq:holomorphic-splitting}.\\

Following \cite{Costello:2022upu}, to obtain the OPEs of a CCA from splitting amplitudes we first assemble the states into generating functions
\be \label{eq:hard-modes}
w[{\tilde\lambda}](z) = \sum_{m,n\in\bbZ^+_0}\frac{({\tilde\lambda}^1)^m({\tilde\lambda}^2)^n}{m!\,n!}w[m,n](z)\,,\qquad \tilde w[{\tilde\lambda}](z) = \sum_{m,n\in\bbZ^+_0}\frac{({\tilde\lambda}^1)^m({\tilde\lambda}^2)^n}{m!\,n!}\tilde w[m,n](z)\,.
\ee
The coordinate $z$ on the celestial sphere determines a left-handed spinor-helicity variable by $\lambda = (1\ z)$, and ${\tilde\lambda}$ can be interpreted as a right-handed spinor-helicity variable. Since we've fixed a normalization for $\lambda$, the scale of ${\tilde\lambda}$ can be identified with the energy of $p^{\alpha\da} = \lambda^\alpha{\tilde\lambda}^\da$. Then, up to appropriate powers of the energy, the generating functions in equation \eqref{eq:hard-modes} are the positive- and negative-helicity hard graviton operators dual to states with momentum $p$.

The tree splitting amplitudes \eqref{eq:tree-splitting} determine the chiral algebra $\cU$ through the identities
\bea
w[{\tilde\lambda}_1](z_1)w[{\tilde\lambda}_2](z_2)&\sim \Split^\mathrm{tree}_-(1^+,2^+;1/2) w[{\tilde\lambda}_1+{\tilde\lambda}_2](z_2)\,, \\
w[{\tilde\lambda}_1](z_1)\tilde w[{\tilde\lambda}_2](z_2)&\sim \Split^\mathrm{tree}_+(1^+,2^-;1/2) \tilde w[{\tilde\lambda}_1+{\tilde\lambda}_2](z_2)\,,
\eea
where $\la12\ra=\sqrt{2}z_{12}$, $[12] = -\sqrt{2}\epsilon_{\da\db}{\tilde\lambda}_1^\da{\tilde\lambda}_2^\db$. Indeed, expanding out the above identities readily reproduces the OPEs from equation \eqref{eq:chiral-algebra}.\\

The holomorphic 1-loop splitting amplitude deforms these by
\be w[{\tilde\lambda}_1](z_1)w[{\tilde\lambda}_2](z_2)\sim \Split_+^\mathrm{1-loop}(1^+,2^+;1/2)\tilde w[{\tilde\lambda}_1+{\tilde\lambda}_2]((z_1+z_2)/2)\,, \ee
where symmetry forces us to evaluate the state $\tilde w[{\tilde\lambda}_1+{\tilde\lambda}_2]$ at $(z_1+z_2)/2$ on the r.h.s. Decomposing into soft modes using equations \eqref{eq:hard-modes} we find that 
\bea &\sum_{p,q,r,s\in\bbZ^+_0}\frac{({\tilde\lambda}_1^1)^p({\tilde\lambda}_1^2)^q({\tilde\lambda}_2^1)^r({\tilde\lambda}_2^2)^s}{p!\,q!\,r!\,s!}w[p,q](z_1)w[r,s](z_2) \\
&\sim \frac{2}{5\pi^2z_{12}^2}\sum_{p,q,r,s\in\bbZ_0^+}\frac{R_4(p,q,r,s)}{(4!)^2}\frac{({\tilde\lambda}_1^1)^p({\tilde\lambda}_1^2)^q({\tilde\lambda}_2^1)^r(\tilde\lambda_2^2)^s}{p!\,q!\,r!\,s!}\tilde w[p+r-4,q+s-4]((z_1+z_2)/2)\,.
\eea
Here
\be \label{eq:Rpqrs}
R_\ell(p,q,r,s) = \sum_{k=0}^\ell(-)^k {\ell\choose k} [p]_{\ell-k}[q]_k[r]_k[s]_{\ell-k}\,,
\ee
for $[x]_n = x(x-1)\dots(x-n+1)$ the descending Pochhammer symbol. This object can be understood as intertwining $\fsl_2(\bbC)_+$ representations
\be (\mathbf{p+q+1})\otimes(\mathbf{r+s+1})\to (\mathbf{p+q+r+s+1-2\boldsymbol{\ell}})\,. \ee
We therefore find that the classical OPEs receive a 1-loop correction
\bea \label{eq:full-double-pole}
&w[p,q](z)w[r,s](0) \\
&\sim\frac{2}{5\pi^2}\frac{R_4(p,q,r,s)}{(4!)^2}\bigg(\frac{1}{z^2}\tilde w[p+r-4,q+s-4] + \frac{1}{2z}\p_z\tilde w[p+r-4,q+s-4]\bigg)(0)\,.
\eea
The first non-trivial deformation occurs with $p+q=r+s=4$, and is determined up to the action of $\fsl_2(\bbC)_+$ by
\be \label{eq:first-correction} w[4,0](z)w[0,4](0)\sim
\frac{2}{5\pi^2}\bigg(\frac{1}{z^2}\tilde w[0,0] + \frac{1}{2z}\p_z\tilde w[0,0]\bigg)(0)\,. \ee
In fact, this single OPE necessitates the full deformation \eqref{eq:full-double-pole}. This is because the $\fsl_2(\bbC)_+$ symmetry admits a split extension by the Heisenberg algebra composed of the zero-modes of $w[m,n]$ with $m+n\leq 1$. Each tower of states furnishes a representation of this algebra, and the $R_\ell(p,q,r,s)$ intertwine these representations. Details can be found in appendix \ref{app:extended-sl2C}.\\

We emphasise that \eqref{eq:full-double-pole} is not necessarily the full 1-loop correction to the chiral algebra. Indeed, we will argue in the next section that there can be subleading simple poles which are bilinear in generators. The interpretation of these terms is subtle, and will be addressed in section \ref{sec:cure}.


\section{Associativity of the operator product}
\label{sec:associativity}

In this section analyse the most general 1-loop corrections to the operator products of the extended CCA compatible with the symmetries of SD gravity. In addition to double poles of the same form as those induced by the 1-loop splitting amplitude, we identify subleading simple poles which can be bilinear in generators.

We then analyse associativity of the operator product in the quantum deformation of the extended CCA, and find that it is violated. We argue that this failure can be attributed to the presence of the 1-loop all-plus amplitudes in SD gravity.

It is essential that we identify the most general 1-loop corrections to OPEs before checking associativity, to ensure that subleading simple poles do not contribute.


\subsection{Characterising 1-loop corrections} \label{subsec:possible-corrections}

In section \ref{sec:splitting} we argued the extended CCA of SD gravity, whose OPEs describe the holomorphic collinear singularities of amplitudes in first 1\textsuperscript{st}-order deformations of the theory, receives quantum corrections. The simplest of these can be attributed to the 1-loop effective vertex characterising double poles in the holomorphic collinear limits of 1-loop graviton amplitudes. In this section we determine the most general 1-loop corrections to the singular parts of OPEs compatible with the symmetries of the theory.\\

Consider the 1-loop correction to the OPE of two generators $\cO_1,\cO_2$. We assume that it takes the form
\be \label{eq:general-correction} \cO_1(z)\cO_2(0) \sim \sum_{N\geq n>0} \frac{1}{z^n}\cO^{(n)}_3(0)\,, \ee
for $N\in\bbZ_{\geq0}$ and $\cO^{(n)}_3$ a sum of normal ordered products of generators and their holomorphic derivatives. We will argue that these objects are tightly constrained by symmetry.\\

All terms in $\cO^{(n)}_3$ must contain one more copy of $\tilde w$ or its derivatives than appear in that pair $\cO_1,\cO_2$. This is most easily seen by reintroducing $\hbar$ into the theory. Since the $\cO^{(n)}_3$ represent 1-loop corrections, they are all accompanied by an explicit factor of $\hbar$. SD gravity is invariant under simultaneous rescalings of $\hbar$ and the negative-helicity Lagrange multiplier field $\Gamma$. Demanding equivariance of equation \eqref{eq:general-correction} under this symmetry gives the claimed result.\\

Next let's restrict to particular cases. First we take both $\cO_1,\cO_2$ to be $w$ states, so that the $\cO^{(3)}_n$ are linear in $\tilde w$ and its derivatives. The number of occurrences of $w$ and its derivatives can be determined by performing simultaneous dilations $z\mapsto r^{-1}z$ on the celestial sphere and $x\mapsto r^{1/2}x$ on spacetime, corresponding to rescaling the coordinate $z$ on twistor space at fixed $v^\da$. Under this symmetry $w,\tilde w$ have charges $1,0$ respectively, and the holomorphic derivative $\p_z$ has charge $1$. equivariance of \eqref{eq:general-correction} requires that the only non-vanishing terms are $\cO^{(2)}_3$, which must be a linear combination of $\tilde w$ states, and $\cO^{(1)}_3$, which is a linear combination of the derivatives $\p_z\tilde w$ and normal ordered products $\normord{w\tilde w}$.

We now impose $\fsl_2(\bbC)_+$ equivariance and dimension matching, i.e., spacetime dilation equivariance. These impose stringent restrictions. For example, they imply that 1-loop corrections to OPEs involving $w[m,n]$ for $m+n\leq3$ involve at worst simple poles, and for $m+n\leq2$ necessarily vanish. In particular, OPEs involving the generators of supertranslations and superrotations are unmodified. Proofs of these results can be found in appendix \ref{app:extended-sl2C}. From here we proceed by identifying the possible 1-loop corrections to the OPEs of low-lying generators, i.e., those transforming in $\fsl_2(\bbC)_+$ representations of small dimension.

The first potential 1-loop corrections are to the OPEs between generators $w[m,n]$ for $m+n=3$. They're determined up to the action of $\fsl_2(\bbC)_+$ by
\be \label{eq:first-simple-pole} w[3,0](z)w[0,3](0)\sim \frac{\beta_{4,4}}{z}w[0,0]\tilde w[0,0](0)\,, \ee
for some $\beta_{4,4}\in\bbC$. The next are between $w[p,q]$ with $p+q=4$ and $w[r,s]$ with $r+s=3$, and are determined by
\be \label{eq:second-simple-pole} w[4,0](z)w[0,3](0)\sim \frac{1}{z}\big(\beta^{2,1}_{5,4}w[1,0]\tilde w[0,0] + \beta_{5,4}^{1,2}w[0,0]\tilde w[1,0]\big)(0) \ee
for constants $\beta_{5,4}^{2,1},\beta_{5,4}^{1,2}\in\bbC$. The first potential double poles appear in the OPEs between generators $w[m,n]$ for $m+n=4$. They're determined by
\bea \label{eq:first-double-pole}
&w[4,0](z)w[0,4](0)\sim \alpha\bigg(\frac{1}{z^2}\tilde w[0,0] + \frac{1}{2z}\p_z\tilde w[0,0]\bigg)(0) + \frac{1}{z}\big(\beta^{3,1}_{5,5}w[1,1]\tilde w[0,0] \\
&+ \beta^{2,2}_{5,5}(\normord{w[1,0]\tilde w[0,1]} + \normord{w[0,1]\tilde w[1,0]}) + \beta^{1,3}_{5,5}w[0,0]\tilde w[1,1]\big)(0)
\eea
for constants $\alpha,\beta_{5,5}^{3,1},\beta_{5,5}^{2,2},\beta_{5,5}^{1,3}\in\bbC$. Note that the coefficient of $\p_z\tilde w[0,0]$ is fixed in terms of $\alpha$ by symmetry. We will find that these simple 1-loop OPEs already capture the essential features of the deformed algebra.\\

Second let's take $\cO_1$ to be a $w$ state, and $\cO_2$ to be a $\tilde w$ state. Equivariance under the combined dilation discussed above shows that the OPE can involve at worst simple poles, and that $\cO^{(1)}_3$ is a linear combination of products $\tilde w\tilde w$. (Since the classical OPEs between $\tilde w$ generators vanish, there is no need to normal order.) Arguments from appendix \ref{app:extended-sl2C} can be adapted to show that OPEs involving $\tilde w[m,n]$ for $m+n\leq 2$ are unmodified. Therefore the first non-trivial 1-loop corrections are determined by
\bea \label{eq:tilde-first-simple-poles}
&w[3,0](z)\tilde w[0,3](0)\sim \frac{1}{z}\tilde\beta_{4,4}\tilde w[0,0]\tilde w[0,0](0)\,, \\
&w[4,0](z)\tilde w[0,3](0)\sim \frac{1}{z}\tilde\beta^{2,1}_{5,4}\tilde w[1,0]\tilde w[0,0](0)\,, \\
&w[3,0](z)\tilde w[0,4](0)\sim \frac{1}{z}\tilde\beta^{2,1}_{4,5}\tilde w[0,1]\tilde w[0,0](0)
\eea
for constants $\tilde\beta_{4,4},\tilde\beta_{5,4}^{2,1},\tilde\beta_{4,5}^{2,1}\in\bbC$.\\

Finally consider taking both $\cO_1,\cO_2$ to be $\tilde w$ states. Under the combined dilation both have vanishing charge, so their OPE cannot be deformed. This same argument applies at any loop order.\\

We saw in section \ref{sec:splitting} that the 1-loop holomorphic splitting amplitude induces a double pole in the $w[4,0],w[0,4]$ OPE of the form appearing in equation \eqref{eq:first-double-pole}. We argue in appendix \ref{app:extended-sl2C} that invariance under the split extension of $\fsl_2(\bbC)_+$ by the Heisenberg algebra necessitates that \eqref{eq:first-double-pole} is completed to
\bea
w[p,q](z)w[r,s](0)\sim \alpha\frac{R_4(p,q,r,s)}{(4!)^2}\bigg(\frac{1}{z^2}\tilde w[p+r-4,r+s-4] \\ 
+ \frac{1}{z}\p_z\tilde w[p+r-4,q+s-4]\bigg)(0) + \dots\,,
\eea
where $+\dots$ denotes subleading terms which are bilinear in generators. This matches the form of the full deformation induced by the 1-loop holomorphic splitting amplitude. It is striking that symmetry is sufficient to determine this structure up to an overall constant, which we saw in section \ref{sec:splitting} takes the value $\alpha=2/5\pi^2$.\\

The bilinear terms accompanying simple poles on the r.h.s. of equations 
\eqref{eq:first-simple-pole}, \eqref{eq:second-simple-pole} and \eqref{eq:first-double-pole} cannot be straightforwardly interpreted as splitting amplitudes. These are subtle, and will be interpreted in section \ref{sec:cure} once we understand the associativity failure in the chiral algebra.\\

We noted in subsection \ref{subsec:classical-chiral-algebra} that $\cU$ is generated in the strong sense by $w[p,q]$ with $p+q\leq3$ and $\tilde w[r,s]$ with $r+s\leq2$. Does the correction to the $w[3,0],w[0,3]$ OPE therefore determine the quantum deformation of the extended CCA in full? In section \ref{sec:associativity} we show that associativity of the operator product is violated so this question is no longer meaningful, however we also present methods of overcoming this failure. Even in these cases it may not be sufficient to determine the corrected OPEs of a generating set. It could be that the generators obey relations involving repeated operator products, and that these are deformed. Nonetheless, we've seen in this subsection and appendix \ref{app:extended-sl2C} that CCAs are highly constrained by symmetry, so there are grounds for hoping they're completely determined by the OPEs of low-lying generators.


\subsection{Failure of associativity} \label{subsec:inconsistency}

Here we show that the 1-loop corrections to OPEs derived in the previous subsection are incompatible with associativity of the operator product if the constant $\alpha$ is non-vanishing, and in particular if it takes the value determined by the 1-loop effective vertex. In the case of SD Yang-Mills, this kind of associativity check was performed in \cite{Costello:2022upu}. Similar calculations also appear in \cite{Ren:2022sws,Bhardwaj:2022anh}.\\

To test associativity for the operators $\cO_i$, $i=1,\dots 3$, we consider the identity
\bea \label{eq:contour-identity}
&\oint_{|z_2|=2}\dif z_2\,\bigg(\oint_{|z_{12}|=1}\dif z_{12}\,\cO_1(z_1)\cO_2(z_2)\bigg)\cO_3(0) \\
&\qquad= \oint_{|z_1|=2}\dif z_1\,\cO_1(z_1)\bigg(\oint_{|z_2|=1}\dif z_2\,\cO_2(z_2)\cO_3(0)\bigg) \\ 
&\qquad\qquad - \oint_{|z_2|=2}\dif z_2\,\cO_2(z_2)\bigg(\oint_{|z_1|=1}\dif z_1\,\cO_1(z_1)\cO_3(0)\bigg)\,,
\eea
which follows from the equivalence of contours illustrated below.

\begin{figure}[h!]
	\centering
	\includegraphics[scale=0.3]{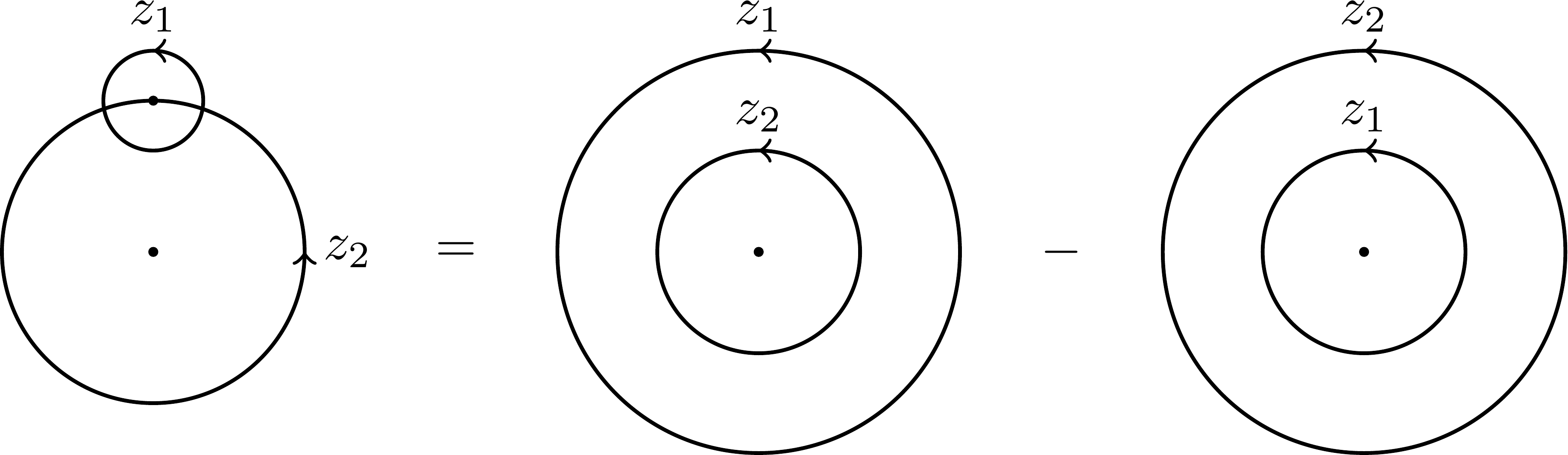}
	\caption{\emph{Equivalence of contours used to test for associativity.}}
	\label{fig:contours}
\end{figure}

Anticipating the conclusions of section \ref{sec:defect}, our choices for the operators $\cO_i$ are motivated by our expectation that the failure of associativity can be traced to the twistoral gravitational anomaly. The minimal dangerous product should be in a sense `dual' to the cocycle \eqref{eq:PBF-anomaly}. Since it depends only on $\bh$, we take $\cO_i\propto w[p_i,q_i]$. It involves five $v^1$ and five $v^2$ derivatives, so we take $\sum_ip_i=\sum_iq_i=5$. The single $z$ derivative leads us to include an explicit factor of $z$ in $\cO_2$.

This prompts us to take
\be \label{eq:operators-445} \cO_1(z) = w[3,0](z)\,,\qquad \cO_2(z) = zw[0,3](z)\,,\qquad \cO_3(z) = w[2,2](z)\,. \ee 

\begin{proposition} \label{prop:associativity-failure}
	The identity \eqref{eq:contour-identity} fails to hold for the choice of operators \eqref{eq:operators-445} if the $w[4,0],w[0,4]$ OPE has a double pole, that is, if the coefficient $\alpha$ from equation \eqref{eq:first-double-pole} is non-vanishing.
\end{proposition}

\begin{proof}
	We proceed by evaluating the three terms in \eqref{eq:contour-identity} separately. Let's start by considering the l.h.s. It gets contributions from the tree OPEs \eqref{eq:chiral-algebra} and 1-loop correction \eqref{eq:first-simple-pole}
	\bea
	&\frac{1}{2\pi\im}\oint_{|z_{12}|=1}\dif z_{12}\,w[3,0](z_1)w[0,3](z_2)z_2 = z_2\big(9w[2,2] + \beta_{4,4}w[0,0]\tilde w[0,0]\big)(z_2)\,. \eea
	The exterior integral extracts the residue in the OPE with $w[2,2](0)$. Since $w[0,0],\tilde w[0,0]$ are central, the only possible contribution is from
	\be z_2w[2,2](z_2)w[2,2](0) \sim \frac{\alpha}{6z_2}\tilde w[0,0](0)\,. \ee
	The factor of $1/6$ relative to the double pole in the $w[4,0],w[0,4]$ follows from $\fsl_2(\bbC)_+$ invariance. The l.h.s. therefore evaluates to
	\be \frac{3}{2}\alpha\tilde w[0,0](0)\,. \ee
	
	Next consider the first term on the r.h.s. The OPE of $w[0,3]$ with $w[2,2]$ has only a simple pole. The explicit factor of $z_2$ therefore ensures that the interior contour integral vanishes.
	
	This leaves the second term on the r.h.s. The interior integral over $z_1$ gives
	\be \frac{1}{2\pi\im}\oint_{|z_1|=1}\dif z_1\,w[3,0](z_1)w[2,2](0) = 6w[4,1](0)\,. \ee
	We emphasise that the 1-loop correction from equation \eqref{eq:second-simple-pole} cannot contribute here, since it factors through the representation $\mathbf{2}$, but the $w[3,0],w[2,2]$ OPE has weight $3$ under the Cartan. The exterior contour integral then vanishes,
	\be - \frac{6}{2\pi\im}\oint_{|z_2|=2}\dif z_2\,z_2w[0,3](z_2)w[4,1](0) = 0\,, \ee
	since $\fsl_2(\bbC)_+$ invariance precludes a double pole appearing at 1-loop in the $w[0,3],w[4,1]$ OPE. We conclude that associativity fails unless $\alpha=0$.
\end{proof}

It's tempting to try and bypass this failure by simply discarding the central element $\tilde w[0,0]$ from the chiral algebra. However, a parallel calculation shows that taking
\be \label{eq:operators-545} \cO_1(z)= w[4,0](z)\,,\qquad \cO_2(z) = zw[0,3](z)\,,\qquad \cO_3(z) = w[2,2](z)\,, \ee
leads to a similar associativity failure. In this case the l.h.s. of equation \eqref{eq:contour-identity} gives $6\alpha\tilde w[1,0]$, whereas the r.h.s. gives $0$. The difference is therefore proportional to $\tilde w[1,0]$, which cannot simply be discarded from the algebra.

Instead taking
\be \label{eq:operators-553} \cO_1(z) = w[4,0](z)\,,\qquad \cO_2(z) = zw[0,4](z)\,,\qquad \cO_3(z) = w[1,1](z) \ee
we find that associativity of the operator product necessitates
\be \beta_{5,5}^{2,2} = 2\alpha\,. \ee
This demonstrates the importance of characterising the 1-loop corrected OPEs in full before performing the associativity check, since the subleading simple poles can contribute.\\


\subsection{Non-universality of holomorphic collinear limits} \label{subsec:non-universal}

In summary, quantum corrections to the extended CCA of SD gravity break associativity of the operator product.\\

From the discussion in subsection \ref{subsec:interpretation}, OPEs in the extended CCA are expected to describe the holomorphic collinear singularities of amplitudes in 1\textsuperscript{st}-order deformations of SD gravity. However, the presence of the 1-loop all-plus amplitudes mean that these singularities are not universal. For example, consider the 1\textsuperscript{st}-order deformation\footnote{This deformation is of particular interest as its finite counterpart gives the SD Palatini action for full Einstein gravity. Expanding around the SD sector in this way can be shown to reproduce the MHV vertices for gravity \cite{Mason:2009afn}.}

\be \label{eq:deform-to-gravity} \delta S_\mathrm{SDGR}[\Gamma,e] = \frac{1}{4}\int_{\bbR^4}\Gamma_{\alpha\gamma}\wedge\Gamma^\gamma_{~\,\beta}\wedge e^{\alpha\da}\wedge e^\beta_{~\,\da}\,. \ee

Working perturbatively around flat space, this deformation introduces a quadratic, cubic and quartic vertex. Representing the cubic vertex by a crossed dot, the mostly-plus amplitudes now acquire a non-universal holomorphic collinear singularity which can be attributed to the diagram illustrated in figure \ref{fig:non-universal}. The blob indicates the insertion of a 1-loop all-plus amplitude.

\begin{figure}[h!]
	\centering
	\includegraphics[scale=0.3]{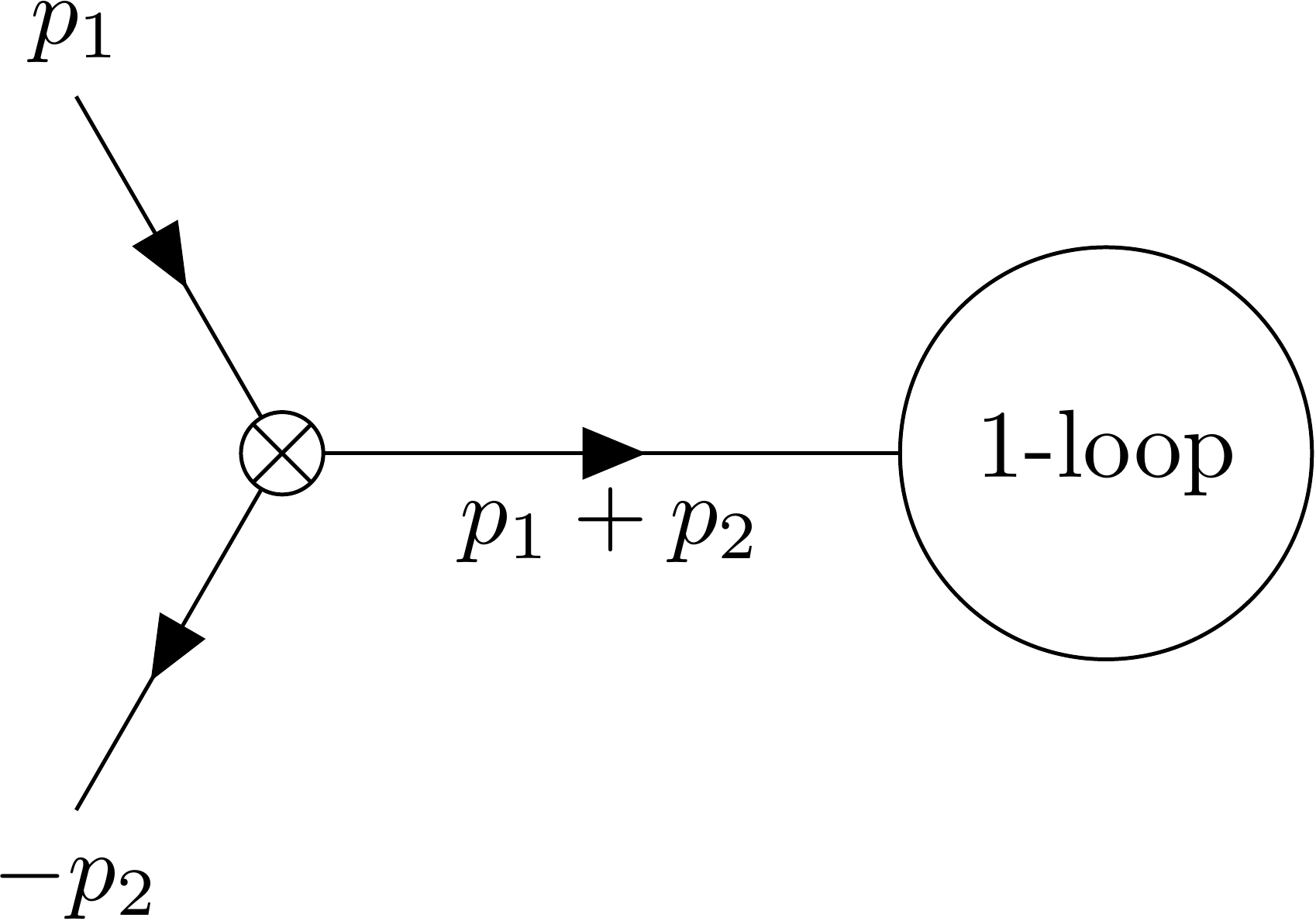}
	\caption{\emph{This diagram is responsible for non-universal holomorphic collinear singularities in the 1-minus amplitudes of a particular 1\textsuperscript{st}-order deformation of SD gravity.}} \label{fig:non-universal}
\end{figure}

The failure of associativity can be attributed to this non-universality. A necessary condition for obtaining a consistent extended CCA is therefore the cancellation of the 1-loop all-plus amplitudes.\\

There is also a second interpretation of this inconsistency in the chiral algebra. The extended CCA is closely related to the hidden symmetry of the 4d hyper-K\"{a}hler hierarchy \cite{Penrose:1976js,Park:1989vq,Takasaki:1989cg,Dunajski:2000iq}. The associativity failure can be understood as a global anomaly in this hidden symmetry, consistent with the proposal of Bardeen that the 1-loop all-plus amplitudes obstruct integrability \cite{Bardeen:1995gk}.\footnote{This interpretation perhaps applies more naturally to the case of the K\"{a}hler scalar. See the discussion in section \ref{sec:discussion} for further details.}


\section{Chiral algebras supported on holomorphic surface defects} \label{sec:defect}

In subsection \ref{subsec:classical-chiral-algebra} we reviewed how the classical extended CCA arises from surface defects in the twistor uplift of SD gravity. In this section we propose that this identification persists at the quantum level, i.e., that the universal chiral algebra supported on holomorphic surface defects is the quantum deformation of the extended CCA we seek.

This is motivated the example of SD QCD, for which the results of \cite{Costello:2022wso,Costello:2022upu} show that the universal holomorphic surface defect in the twistor uplift coincides with the extended CCA describing the holomorphic collinear singularities of form factors.\footnote{At least in cases in which the twistorial gauge anomaly vanishes.}

As evidence, we identify the anomalous 1-loop diagrams in the coupled bulk-defect system which necessitate corrections to OPEs. We evaluate these in some simple cases, and obtain a precise match with the coefficient of the double poles induced by the 1-loop effective vertex in section \ref{sec:splitting}. We also find that some of the subleading bilinear terms identified in subsection \ref{subsec:possible-corrections} are non-vanishing.

We then interpret the associativity failure from the previous section in terms of the recently discovered twistorial gravitational anomaly from \cite{Bittleston:2022nfr}.


\subsection{Anomalous 1-loop diagrams in the bulk-defect system} \label{subsec:anomalous-diagram}

The universal holomorphic surface defect wrapping the real twistor line $\cL_0=\{u^\da=0\}\subset\bbPT$ couples to Poisson-BF theory via
\be \frac{1}{2\pi\im}\sum_{m,n\geq0}\int_{\cL_0}\dif z\,(w[m,n](z)\cD_{m,n}\bh + \tilde w[m,n](z)\cD_{m,n}\tilde\bh)\,, \ee
for two towers of operators $w[m,n],\tilde w[m,n]$ with $m,n\in\bbZ_{\geq0}$. We saw in subsection \ref{subsec:classical-chiral-algebra} that tree level BRST invariance of the coupled bulk-defect system necessitates that these operators obey the relations of the classical extended CCA of SD gravity \eqref{eq:chiral-algebra}. Quantum corrections can be understood as arising from anomalous loop diagrams in the coupled bulk-defect system. Here we identify the anomalous diagrams at 1-loop.\\

Let's being by concentrating on the $w,w$ OPEs. These are determined by Feynman diagrams with two external bulk legs both of which are $\bh$. The only potentially anomalous 1-loop diagrams are illustrated in figure \ref{fig:BRST-1loop}. The symmetry arguments from subsection \ref{subsec:possible-corrections} apply equally well here, and so the simplest 1-loop corrections to OPEs must take the form given in equations \eqref{eq:first-double-pole}, \eqref{eq:second-simple-pole} and \eqref{eq:first-double-pole}.

\begin{figure}[h!]
	\centering
	\begin{subfigure}[t]{0.49\textwidth}
		\centering
		\includegraphics[scale=0.3]{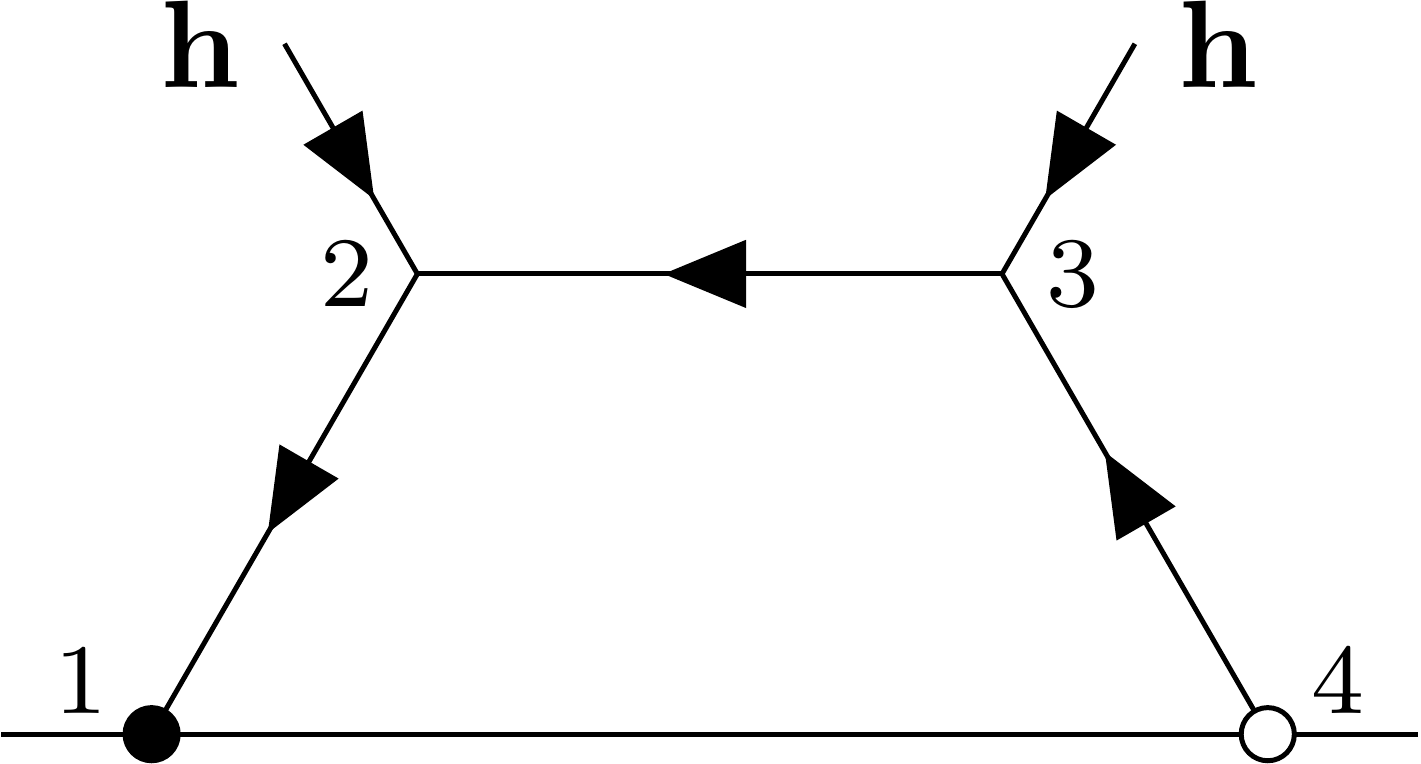}
		\subcaption{Defect box}
		\label{subfig:box}
	\end{subfigure}
	\begin{subfigure}[t]{0.49\textwidth}
		\centering
		\includegraphics[scale=0.3]{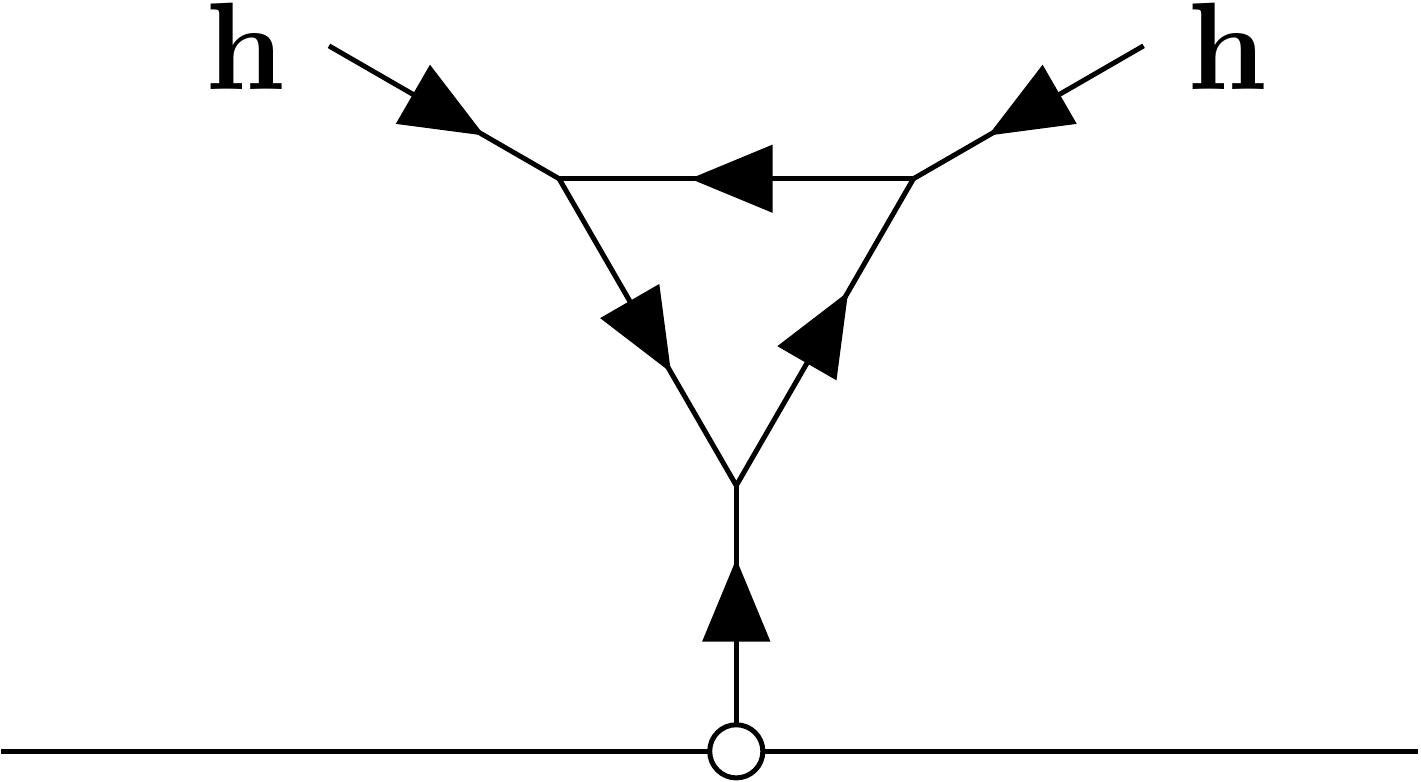}
		\subcaption{Bulk triangle}
		\label{subfig:triangle}
	\end{subfigure}
	\caption{\emph{1-loop Feynman diagrams whose BRST variation could in principle necessitate a correction to the $w,w$ OPEs. The undirected horizontal line represents the defect supported on $\cL_0$, and the solid and empty dots represent couplings to $w,\tilde w$ on the defect respectively. We have labelled the vertices of the first diagram for later convenience.}}
	\label{fig:BRST-1loop}
\end{figure}

The $w,\tilde w$ OPEs also receive corrections at 1-loop from essentially the same diagram as in subfigure \ref{subfig:box}. The differences are that one external vertex is replaced by $\bg$, and both legs on the defect couple to $\tilde w$. Again, the symmetry arguments from subsection \ref{subsec:possible-corrections} apply in this context, and this simplest corrections take the form given in \eqref{eq:tilde-first-simple-poles}.

The $\tilde w,\tilde w$ OPEs are not deformed at any loop order, since forgetting the defect from a diagram in the coupled bulk-defect system leaves a bulk diagram, which can have at most one negative helicity external leg.


\subsection{Gauge fixing} \label{subsec:gauge-fix}

In the remainder of this section we explicitly compute the BRST variation of the diagrams in figure \ref{fig:BRST-1loop} for particular specialisations of the external legs, allowing us to evaluate the constants $\alpha,\beta_{4,4},\beta_{5,5}^{2,2}$ from equations \eqref{eq:first-simple-pole} and \eqref{eq:first-double-pole}.\\

First we need to fix the gauge. To do so let's restrict to the patch $\{z\neq\infty\}\cong\bbC^3\subset\bbPT$, on which we continue to use the holomorphic coordinates $\{Z^a\} = \{z,v^\da\}$. We assume the defect wraps $\{v^\da=0\}\cong\bbC\subset\bbC^3$, which can be viewed as a subset of the real twistor line $\cL_0$. Since anomalies are local, working in this patch is sufficient.

In this patch the line bundles $\cO(n)$ trivialize, and in particular $\Dif^3Z$ can be identified with the untwisted holomorphic $(3,0)$-form
\be \dif^3Z = \dif z\dif v^1\dif v^2 = \frac{1}{3!}\varepsilon_{abc}\dif Z^a\dif Z^b\dif Z^c\,. \ee
Specifying the standard Hermitian form
\be \|\dif Z\|^2 = \delta_{a\bar b}\dif Z^a\dif\bar Z^{\bar b} =  \dif z\dif\bar z + \dif v^1\dif\bar v^1 + \dif v^2\dif\bar v^2 \ee
we can impose the gauge fixing conditions
\be \bar\p^\dag\bh = - \star\bar\p\star\bh = 0\,,\qquad \bar\p^\dag\tilde\bh = - \star\bar\p\star\tilde\bh = 0 \ee
for $\star$ the anti-linear Hodge star.\footnote{We emphasise that this is not the Woodhouse gauge familiar to twistor theorists \cite{Woodhouse:1985id}.} Technically in the BV formalism this specifies a Lagrangian subspace in the extended space of fields. The operator $\bar\p^\dag$ has the property
\be \{\bar\p,\bar\p^\dag\} = \Delta_{\bar\p}\ee
for $\Delta_{\bar\p}$ the Dolbeault Laplacian. The propagator is the pullback by the difference map
\be D_{12}:(\bbC^3)_1\times(\bbC^3)_2\to\bbC^3\,,\qquad (Z_1,Z_2)\mapsto Z_{12} = Z_1-Z_2 \ee
of the $(0,2)$-form $P$ obeying
\be \bar\p P = 2\pi\im\bar\delta^3 \ee
for $\bar\delta^3$ the $(0,3)$-form $\delta$-function with support at $Z=0$ such that
\be \int_{\bbC^3}\dif^3Z\,\bar\delta^3 = -1\,. \ee
Explicitly
\be P = 2\pi\im\bar\p^\dag\Delta_{\bar\p}^{-1}\bar\delta^3 = \frac{1}{8\pi^2}\bar\p^\dag\bigg(\frac{\dif^3\bar Z}{\|Z\|^4}\bigg) = \frac{\varepsilon_{\bar a\bar b\bar c}\bar Z^{\bar a}\dif\bar Z^{\bar b}\dif\bar Z^{\bar c}}{4\pi^2\|Z\|^6} \,. \ee
Fortunately it's possible to evaluate the necessary loop diagrams using a point-splitting regularisation on the defect, so there is no need to introduce a regulated propagator.


\subsection{Explicit diagram computation} \label{subsec:evaluate-diagram}

We are now in position to evaluate the constants $\alpha,\beta_{4,4},\beta_{5,5}^{2,2}$. In principle they can receive contributions from both diagrams in figure \ref{fig:BRST-1loop}, but fortunately the bulk triangle diagram vanishes for algebraic reasons in a 6d holomorphic theory \cite{Williams:2018ows}. We can therefore concentrate on the box diagram with one edge representing the defect as illustrated in subfigure \ref{subfig:box}.\\

We label the vertices by $i=1,\dots,4$ as shown in figure \ref{subfig:box}, and denote their positions by $Z_i$. In particular, $Z_1 = (z_1,0)$ locates the coupling of $\bh$ to the defect, and we move around the diagram clockwise so that $Z_4=(z_4,0)$ locates the coupling of $\bg$ to the defect. The vertex at which an external field is evaluated is denoted by a subscript. A point-splitting regularisation on the defect, $|z_{14}|\geq\epsilon$, ensures that our integrals are finite. At the end of the calculation we will take $\epsilon\to0$.

A linearised BRST transformation $\delta\bh = \bar\p\bh$ can act on either of the external legs, leading to two terms. Iteratively integrating by parts and employing the identity
\be \bar\p P = 2\pi\im\bar\delta^3 \ee
we can express the variation as a sum over possible `contractions' of the internal edges, where we replace the propagator on a given edge by a $\delta$-function, together with a boundary term on the defect.

Consider contracting the propagator connecting the two bulk vertices at $Z_2,Z_3$ resulting in the diagram illustrated in figure \ref{fig:23-contraction}. Eliminating the $Z_3$ variable in favour of $Z_2$, this diagram involves holomorphic derivatives acting upon the product of propagators $P_{12}P_{24}$. Holomorphic derivatives do not modify the antiholomorphic form structure, so the product of forms
\be (\varepsilon_{\bar a\bar b\bar c}\bar Z_{12}^{\bar a}\dif\bar Z_{12}^{\bar b}\dif\bar Z_{12}^{\bar c})(\varepsilon_{\bar d\bar e\bar f}\bar Z_{24}^{\bar d}\dif\bar Z_{24}^{\bar e}\dif\bar Z_{24}^{\bar f}) \ee
appears in the integrand. Expanding
\be
\varepsilon_{\bar a\bar b\bar c}Z^{\bar a}\dif\bar Z^{\bar b}\dif\bar Z^{\bar c} = \bar z\dif\bar v^\da\dif\bar v_\da + 2\bar v^\da\dif\bar v_\da\dif\bar z\,,
\ee
recalling that $Z_1 = (z_1,0)$, $Z_4 = (z_4,0)$ and writing $Z_2=(z_2,v^\da)$ we have
\be (\varepsilon_{\bar a\bar b\bar c}\bar Z_{12}^{\bar a}\dif\bar Z_{12}^{\bar b}\dif\bar Z_{12}^{\bar c})(\varepsilon_{\bar d\bar e\bar f}\bar Z_{24}^{\bar d}\dif\bar Z_{24}^{\bar e}\dif\bar Z_{24}^{\bar f}) = - 4\bar v^\da\bar v^\dd\dif\bar v_\da\dif\bar v_\dd\dif\bar z_{12}\dif\bar z_{24} = -2 [\bar v\,\bar v][\dif\bar v\,\dif\bar v]\dif\bar z_{12}\dif\bar z_{24}\,, \ee
which vanishes since $[\bar v\,\bar v]=0$. The same argument shows that there are no contributions to the anomaly from any propagator contractions of the diagram in subfigure \eqref{subfig:box}.

\begin{figure}[h!]
	\centering
	\includegraphics[scale=0.3]{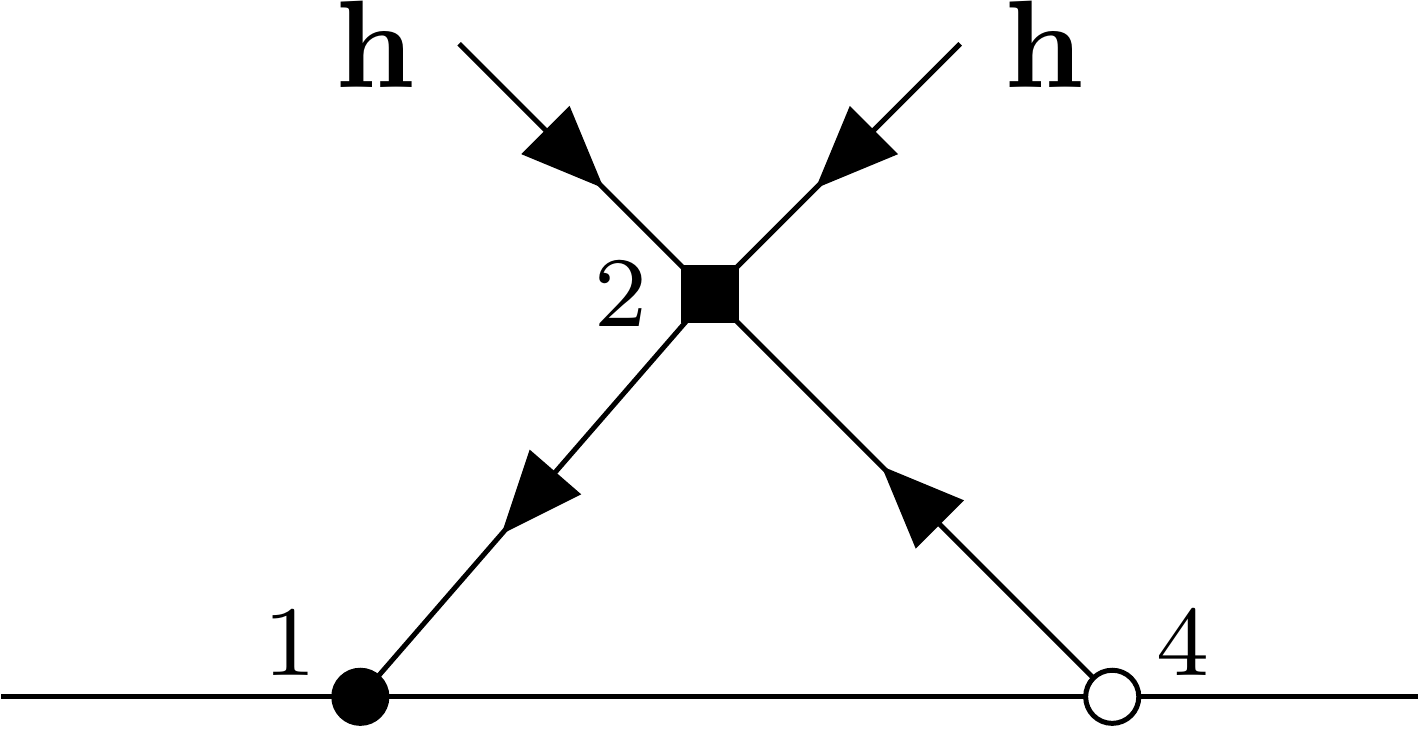}
	\caption{\emph{1-loop Feynman diagram obtained by contracting the edge connecting the two bulk vertices in subfigure \ref{subfig:box}. The square dot represents the combined bulk vertices.}}
	\label{fig:23-contraction}
\end{figure}

This leaves the boundary contribution. Writing $z_0 = (z_1+z_4)/2$, it takes the form
\be \label{eq:residue} \bigg(\frac{1}{2\pi\im}\bigg)^2\sum_{a,b,c,d\in\bbZ_{\geq0}}\int_\bbC\dif z_0 \oint_{|z_{14}|=\epsilon} \dif z_{14}\,w[a,b](z_1) \tilde w[c,d](z_4)\,\cI_{a,b,c,d}\big(z_1,z_4;\bh,\bh\big) \ee
where
\bea \label{eq:Iabcd}
&\cI_{a,b,c,d}\big(z_1,z_4;\bh,\bh^\prime\big) \\
&= \bigg(\frac{\im}{2\pi}\bigg)^2\int_{(\bbC^3)_{2}\times(\bbC^3)_{3}}\{\cD_{a,b,1}P_{12},\bh_2\}_2\,\dif^3 Z_2\,P_{23}\,\dif^3Z_3\,\{\bh^\prime_3,\cD_{c,d,4}P_{34}\}_3\Big\vert_{\dif\bar z_1=\dif\bar z_4=\dif\bar z_0}\,.
\eea
The subscripts on the Poisson brackets dictate on which copy of $\bbC^3$ they act. Prior to restricting $\dif\bar z_1=\dif\bar z_4=\dif\bar z_0$ the integral defines a  $(0,\bullet)$-polyform on $\{|z_{14}|\geq\epsilon\}\subset\bbC_1\times\bbC_4$. In equation \eqref{eq:residue} we pullback to the boundary of this region. Since the integral over $|z_{14}|=\epsilon$ is already saturated, only the $\dif\bar z_0$ component can contribute.

\begin{figure}[h!]
	\centering
	\begin{subfigure}[t]{0.49\textwidth}
		\centering
		\includegraphics[scale=0.3]{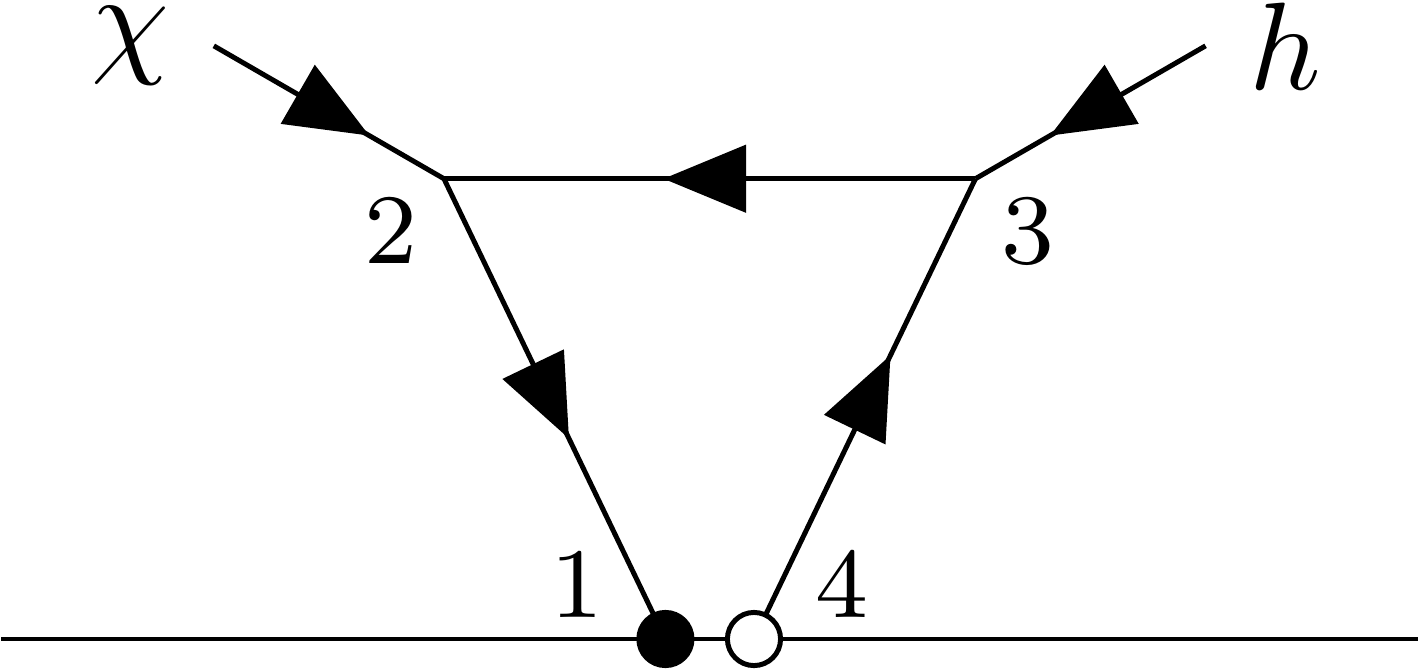}
		\subcaption{Anticlockwise}
		\label{subfig:anticlockwise}
	\end{subfigure}
	\begin{subfigure}[t]{0.49\textwidth}
		\centering
		\includegraphics[scale=0.3]{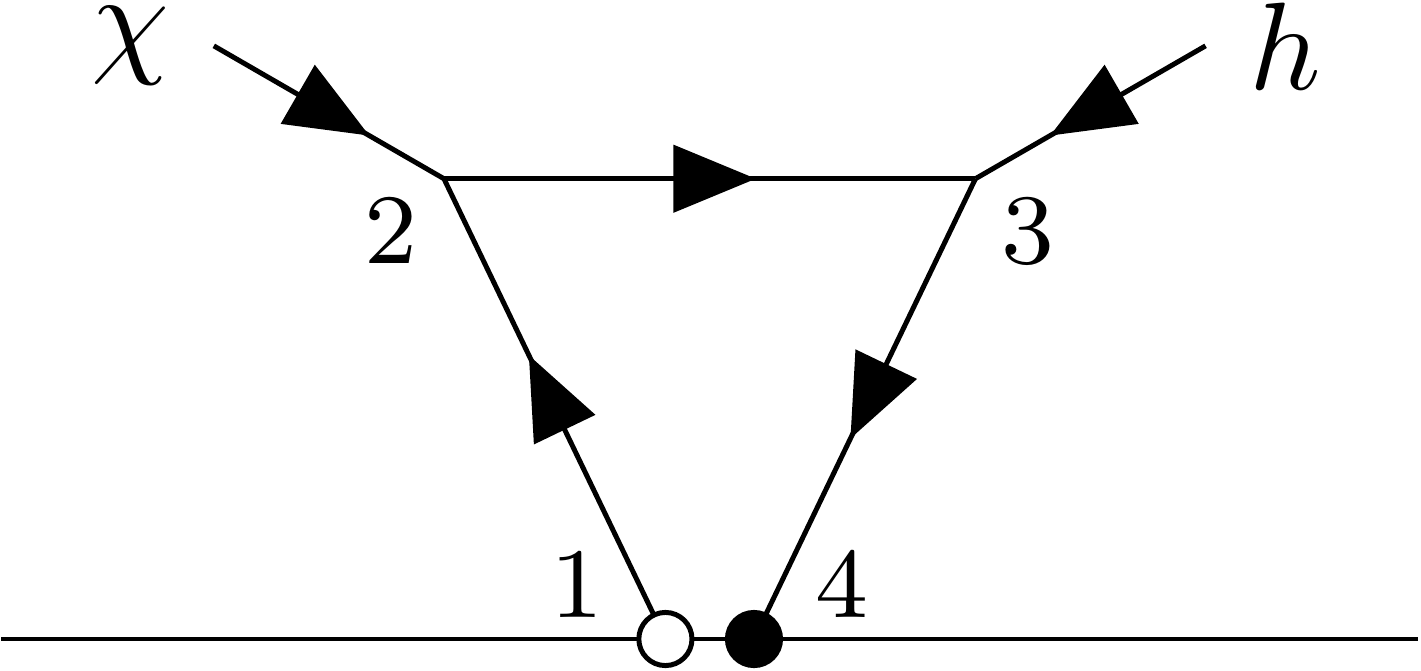}
		\subcaption{Clockwise}
		\label{subfig:clockwise}
	\end{subfigure}
	\caption{\emph{The BRST variation of subfigure \ref{subfig:box} gives a boundary term on the defect. Since we are taking the external legs to be distinct, there are now two contributing diagrams.}}
	\label{fig:box-bdry}
\end{figure}

By specialising the external fields we can extract the corrections to particular OPEs. Since we're computing an anomaly, we should take one external leg to be a ghost $\chi$ and the other a physical field $h$. These are distinct, so we must sum over the two possible orderings around the loop, illustrated in figure \ref{fig:box-bdry}. Fortunately these contribute almost identically: the second can be obtained from the first by exchanging $w\leftrightarrow\tilde w$ (before taking any OPEs on the defect).

We note that the diagrams in \ref{fig:box-bdry} can be interpreted as twistor uplifts of the non-factorizing diagrams appearing in the calculation of 1-minus graviton amplitudes by recursion methods \cite{Dunbar:2010xk,Alston:2015gea}, with the defect replacing the MHV gravity currents. This is evidence that the chiral algebra supported on the defect really does describe the holomorphic collinear limits of amplitudes in 1\textsuperscript{st}-order deformations of SD gravity.\\

Taking $\chi\sim (v^1)^3$, $h\sim (v^2)^3\dif\bar z$ we obtain the coefficient $\beta_{4,4}$ of the simple pole in $w[3,0],w[0,3]$ OPE. In this case $\fsl_2(\bbC)_+$ invariance and dimension matching imply that only non-vanishing term in the sum has $(a,b,c,d) = (0,0,0,0)$. Similarly, by taking $\chi\sim z(v^1)^4$, $h\sim (v^2)^4\dif\bar z$ we can extract the coefficients $\alpha,\beta_{5,5}^{2,2},\beta_{5,5}^{3,1},\beta_{5,5}^{1,3}$ appearing in the $w[4,0],w[0,4]$ OPE \eqref{eq:first-double-pole}. $\fsl_2(\bbC)_+$ invariance and dimension matching imply that the only non-zero contributions are from $(a,b,c,d) = (1,0,0,1)$, $(0,1,1,0)$, $(1,1,0,0)$, $(0,0,1,1)$. The first two options contribute to $\alpha,\beta^{2,2}_{5,5}$, and the latter two to $\beta^{3,1}_{5,5}$ and $\beta^{1,3}_{5,5}$ respectively. We'll concentrate on computing $\alpha$, since from the discussion in section \ref{sec:splitting} it can be identified with the coefficient of the 1-loop holomorphic splitting amplitude. We will therefore focus on the terms $(a,b,c,d) = (1,0,0,1)$, $(0,1,1,0)$ in the sum.\\

Recalling the explicit form of the propagators, the antiholomorphic form structure is
\bea
&(\varepsilon_{\bar a\bar b\bar c}\bar Z_{12}^{\bar a}\dif\bar Z_{12}^{\bar b}\dif\bar Z_{12}^{\bar c})(\varepsilon_{\bar d\bar e\bar f}\bar Z_{23}^{\bar d}\dif\bar Z_{23}^{\bar e}\dif\bar Z_{23}^{\bar f})(\varepsilon_{\bar g\bar h\bar i}\bar Z_{34}^{\bar g}\dif\bar Z_{34}^{\bar h}\dif\bar Z_{34}^{\bar i}) \\ 
&= 8[\bar v_2\,\bar v_3](\bar z_{12}\dif\bar z_{23}\dif\bar z_{34} - \bar z_{23}\dif\bar z_{12}\dif\bar z_{34} + \bar z_{34}\dif\bar z_{12}\dif\bar z_{23})\dif \bar v^1_2\dif\bar v^2_2\dif\bar v^1_3\dif\bar v^2_3\,.
\eea
Substituting the above into \eqref{eq:Iabcd}, and performing the restriction to $\dif\bar z_1 = \dif\bar z_4 = \dif\bar z_0$ gives
\bea
&\cI_{1,0,0,1}\big(z_1,z_4;z(v^1)^4,(v^2)^4\dif\bar z\big) \\
&= \frac{2^23^2\bar z_{14}\dif\bar z_0}{\pi^8}\int_{(\bbC^3)_2\times(\bbC^3)_3}\dif^6Z_2\,\dif^6Z_3\frac{z_2(v^1_2)^3(v^2_3)^3[\bar v_2\,\bar v_3]\bar v^1_2\bar v^2_2\bar v^1_3\bar v^2_3}{\|Z_{12}\|^{10}\|Z_{23}\|^6\|Z_{34}\|^{10}}\,, \\
&\cI_{0,1,1,0}\big(z_1,z_4;z(v^1)^4,(v^2)^4\dif\bar z\big) \\
&= \frac{2^23^2\bar z_{14}\dif\bar z_0}{\pi^8}\int_{(\bbC^3)_2\times(\bbC^3)_3}\dif^6Z_2\,\dif^6Z_3\frac{z_2(v^1_2)^3(v^2_3)^3[\bar v_2\,\bar v_3](\bar v^2_2)^2(\bar v^1_3)^2}{\|Z_{12}\|^{10}\|Z_{23}\|^6\|Z_{34}\|^{10}}\,.
\eea
We evaluate these integrals in appendix \ref{app:integrals}
\bea \label{eq:Feynman-integrals} &\cI_{1,0,0,1}\big(z_1,z_4;z(v^1)^4,(v^2)^4\dif\bar z\big) = - \frac{(9 z_{14} + 70 z_0)\dif\bar z_0}{70\pi^2z_{14}}\,, \\ &\cI_{0,1,1,0}\big(z_1,z_4;z(v^1)^4,(v^2)^4\dif\bar z\big) = - \frac{(2z_{14} + 28 z_0)\dif\bar z_0}{70\pi^2z_{14}}\,.
\eea
Plugging these expressions into equation \eqref{eq:residue} gives
\bea \label{eq:before-residue}
&\bigg(\frac{1}{2\pi\im}\bigg)^2\bigg(-\frac{1}{70\pi^2}\bigg)\int_\bbC\dif^2 z_0\oint_{|z_{14}|=\epsilon}\frac{\dif z_{14}}{z_{14}}\,\big((9z_{14}+70z_0)w[1,0](z_1)\tilde w[0,1](z_4) \\
& + (2z_{14}+28z_0)w[0,1](z_1)\tilde w[1,0](z_4)\big)\,.
\eea
At this point the dependence on the regulator $\epsilon$ drops out, as the integral over $z_{14}$ computes a residue. It receives two types of contributions: first from the explicit simple poles, and second from the tree OPEs between $w,\tilde w$. These are determined by the tree splitting amplitudes, so the second class of terms arise from a kind of triangle diagram. It can be interpreted as the twistor uplift of the off-shell triangle diagram on spacetime computing the 1-loop holomorphic splitting amplitude.

Using the classical OPEs \eqref{eq:chiral-algebra} the contour integral in equation \eqref{eq:before-residue} gives
\bea \label{eq:anticlockwise}
&\bigg(\frac{1}{2\pi\im}\bigg)\bigg(-\frac{1}{10\pi^2}\bigg)\int_\bbC\dif^2 z_0\,\big(\tilde w[0,0](z_0) \\
&+ 5z_0\normord{\{w[1,0],\tilde w[0,1]\}}\!(z_0) + 2z_0\normord{\{w[0,1],\tilde w[1,0]\}}\!(z_0)\big)\,,
\eea
where here $\{\ ,\ \}$ denotes the anticommutator. The value of the diagram in subfigure \ref{subfig:clockwise} is obtained from equation \eqref{eq:before-residue} by exchanging $w\leftrightarrow\tilde w$, so it contributes
\bea \label{eq:clockwise}
&\bigg(\frac{1}{2\pi\im}\bigg)\bigg(-\frac{1}{10\pi^2}\bigg)\int_\bbC\dif^2 z_0\,\big(\tilde w[0,0](z_0) \\
&+ 2z_0\normord{\{w[1,0],\tilde w[0,1]\}}\!(z_0) + 5z_0\normord{\{w[0,1],\tilde w[1,0]\}}\!(z_0)\big)\,.
\eea
Combining equations \eqref{eq:anticlockwise} and \eqref{eq:clockwise} gives
\bea \label{eq:total}
&\bigg(\frac{1}{2\pi\im}\bigg)\bigg(-\frac{1}{5\pi^2}\bigg)\int_\bbC\dif^2 z_0\,\big(\tilde w[0,0](z_0) + 7z_0(\normord{w[1,0]\tilde w[0,1]} + \normord{w[0,1]\tilde w[1,0]})(z_0)\big)\,.
\eea
This should cancel against the BRST variation of the bilocal term on the defect
\be \bigg(\frac{1}{2\pi\im}\bigg)^2\int_\bbC\dif z_0\oint_{|z_{14}|=\epsilon}\dif z_{14}\,w[4,0](z_1)w[0,4](z_4)\cD_{4,0,1}\chi_1\cD_{0,4,4}h_4  \ee
Taking the external legs to be $\chi\sim z(v^1)^4$, $h\sim(v^2)^4\dif\bar z$ and assuming the quantum corrected OPE takes the form in equation \eqref{eq:first-double-pole} this simplifies to
\be \label{eq:cancel-total} \frac{1}{4\pi\im}\int_\bbC\dif^2z_0\,\big(\alpha\tilde w[0,0](z_0) + 2\beta^{2,2}_{5,5}z_0(\normord{w[1,0]\tilde w[0,1]} + \normord{w[0,1]\tilde w[1,0]})(z_0)\big)\,.
\ee
Comparing equations \eqref{eq:total} and \eqref{eq:cancel-total}, we conclude that
\be \label{eq:coefficients-direct} \alpha = \frac{2}{5\pi^2}\,,\qquad\qquad \beta^{2,2}_{5,5} = \frac{7}{5\pi^2}\,. \ee
The value of $\alpha$ precisely matches the coefficient we found in section \ref{sec:splitting}. Associativity arguments in section \ref{sec:cure} provide a second check on these coefficients.

Instead taking $\chi\sim(v^1)^3$, $h\sim (v^2)^3\dif\bar z$ extracts the coefficient $\beta_{4,4}$ of the simple pole in the 1-loop correction to the $w[3,0],\tilde w[0,3]$ OPE. The value of the relevant Feynman integral is given appendix \ref{app:integrals} and it implies $\beta_{4,4} = 3/8\pi^2$. Similarly, taking $\chi\sim(v^1)^4$, $h\sim(v^2)^4\dif\bar z$ extracts the coefficient $\beta^{2,2}_{5,5}$, and the result matches equation \eqref{eq:coefficients-direct}.\\

In summary, we've obtained through the combination of symmetry arguments and a direct calculation precisely the deformed OPEs induced by the 1-loop holomorphic splitting amplitude from section \ref{sec:splitting}. This is strong evidence that the universal chiral algebra supported on holomorphic surface defects should be identified with the quantum deformation of the extended CCA. We've also seen that the 1-loop corrected OPEs involve subleading terms which are bilinear in generators. 


\subsection{Twistor interpretation of the associativity failure}

We saw in section \ref{sec:associativity} that 1-loop corrections to the extended CCA violate associativity of the operator product, and interpreted this as signalling the non-universality of the holomorphic collinear limits of amplitudes in 1\textsuperscript{st}-order deformations of SD gravity. We further argued that a necessary condition for universality is the vanishing of the 1-loop all-plus amplitudes.

Here we present an alternative interpretation based on the realisation of the quantum deformation of the extended CCA as the universal holomorphic surface defect.\\

It was recently argued that the twistor uplift of SD gravity, which is Poisson-BF theory, suffers from a gravitational anomaly. As reviewed in subsection \ref{subsec:quantum-SDGR}, it can be attributed to an anomalous 1-loop box diagram on twistor space. Given that the twistor formulation of the theory is anomalous, there is no reason to expect that holomorphic surface defects supported on real twistor lines should support consistent chiral algebras.

This is compatible with the perspective adopted in section \ref{sec:associativity}: the anomalous box diagram on twistor space can be identified with the 1-loop 4-point all-plus amplitude on spacetime.\\

This interpretation suggests the cancelling the twistorial anomaly is a sufficient condition for obtaining a quantum extended CCA with associative operator product, since consistency of the chiral algebras supported on holomorphic surface defects is expected if the twistorial theory exists. A number of methods cancelling the anomaly were identified in \cite{Bittleston:2022nfr}.


\section{Correcting associativity} \label{sec:cure}

In this section we explore methods of cancelling the twistorial anomaly in SD gravity, and hence of obtaining consistent extended CCAs. On spacetime this amounts to eliminating the 1-loop all-plus amplitudes, which we've seen is necessary for the collinear singularities of amplitudes in 1\textsuperscript{st}-order deformations to be universal.

The main approach we shall consider is by coupling to a $\p$-closed $(2,1)$-form field on twistor space, representing a 4\textsuperscript{th}-order gravitational axion on spacetime. By tuning the axion coupling the twistorial anomaly is cancelled via a kind of Green-Schwarz mechanism. Equivalently, on spacetime the 1-loop all-plus amplitudes are all cancelled by axion exchange \cite{Bittleston:2022nfr}. We show that incorporating the axion states into the chiral algebra restores associativity.

We will also briefly consider theories of SD gravity coupled to suitable matter in which the anomaly naturally vanishes. This occurs, e.g., in the case of SD supergravity theories.


\subsection{Incorporating axion states into the chiral algebra} \label{subsec:axion-chiral-algebra}

To obtain a consistent quantum extended CCA we adapt the proposal of \cite{Costello:2022upu}. Therein, the authors argue that a similar failure of associativity in the quantum deformation of the extended CCA of SD Yang-Mills can be remedied by cancelling a twistorial gauge anomaly. One method of doing so, originally proposed in \cite{Costello:2021bah}, is via a kind of Green-Schwarz mechanism. It involves coupling to a $\p$-closed $(2,1)$-form field on twistor space, describing a 4\textsuperscript{th}-order axion on spacetime. An advantage of this approach is that it cancels 1-loop and tree effects against one another, so can be used to bootstrap the 1-loop corrected chiral algebra from its tree OPEs.\\

In \cite{Bittleston:2022nfr} a variant of this mechanism was developed for SD gravity, which involves coupling to the same field on twistor space. This was reviewed in subsection \ref{subsec:quantum-SDGR}, but in brief, on twistor space we couple to a new polyform field $\bseta\in\Omega^{2,\bullet}_\mathrm{cl}(\bbPT)[1]$. It's physical part is a $\p$-closed $(2,1)$-form. For the BV action we take
\be \label{eq:eta-action} S[\bh,\bg;\bseta] = S_\mathrm{PBF}[\bh,\bg] + \frac{1}{4\pi\im}\int_\bbPT\big(\p^{-1}\bseta\,\nbar\bseta + \mu\,\bseta\,\p^\db\p_\da\bh\,\p^\da\p_\db\p\bh\big)\,. \ee
The physical part of $\bseta$ descends to a scalar $\rho$ on spacetime, and the physical part of the above action becomes
\be \label{eq:SDGR+axion-again} S_{\mathrm{SDGR}+\rho}[\rho;\Gamma,e] = S_\mathrm{SDGR}[\Gamma,e] + \int_{\bbR^4}\bigg(\vol_g\,\frac{1}{2}(\Delta_g\rho)^2 + \frac{\mu}{\sqrt{2}}\,\rho \,R^\mu_{~\,\nu}\wedge R^\nu_{~\,\mu}\bigg)\,, \ee
revealing $\rho$ to be kind of 4\textsuperscript{th}-order gravitational axion. The twistorial anomaly is cancelled by a Green-Schwarz mechanism if the axion coupling is tuned so that
\be \label{eq:tune-mu} \mu^2 = \frac{1}{5!}\bigg(\frac{\im}{2\pi}\bigg)^2\,. \ee
Since the twistorial theory exists, its holomorphic surface defects should support a consistent universal chiral algebra, which we expect can be identified with the quantum extended CCA. Mathematically this chiral algebra is Koszul dual to the algebra of local operators in Poisson-BF theory on twistor space. On spacetime, if $\mu$ satisfies equation \eqref{eq:tune-mu} then tree exchange of gravitational axions cancels the 1-loop all-plus amplitudes. The coupled theory has trivial amplitudes, so the non-universal collinear singularities discussed in subsection \ref{subsec:non-universal} do not arise.

In this subsection we verify that the associativity failure identified in proposition \ref{prop:associativity-failure} is cured in the 1-loop corrected extended CCA of the theory \eqref{eq:SDGR+axion-again}.\\

$\bseta$ contributes two new towers of generators to the chiral algebra, which have been identified in \cite{Costello:2022upu}. We denote them by $e[p,q]$ for $p+q>0$ and $f[r,s]$ for $r,s\geq0$. The quantum numbers of these states are listed in table \ref{tab:e-f}. 

\begin{table}[h!]
\centering
  \begin{tabular}{c c c c c}
    \toprule
	Generator & Field & Spin & $\fsl_2(\bbC)_+$ representation & Dimension \\ \midrule
	$e[m,n]\,,~m+n>0$ & $\bseta$ & $-(m+n)/2$ & $\mathbf{m+n+1}$ & $-m-n$ \\
	$f[m,n]\,,~m,n\geq0$ & $\bseta$ & $-(m+n)/2$ & $\mathbf{m+n+1}$ & $-2-m-n$ \\
    \bottomrule
  \end{tabular}
\caption{Quantum numbers of $e,f$.}
\label{tab:e-f}
\end{table}

The tree OPEs of these generators are determined by the vertices in the action \eqref{eq:eta-action}. They can be computed in a number of ways, but we'll adopt the approach from subsection \ref{subsec:classical-chiral-algebra}. The $e,f$ states are interpreted as parametrizing the coupling of $\bseta$ to the universal holomorphic defect supported on the real twistor line $\cL_0$:
\bea \label{eq:axion-defect}
&\frac{1}{2\pi\im}\sum_{m,n\in\bbZ_{\geq0},\,m+n>0}\int_{\cL_0}\dif z\,\bigg(e[m,n](z)\cD_{m,n}\bsgamma_z + \p_ze[m,n](z)\frac{\cD_{m-1,n}\bsgamma_1 + \cD_{m,n-1}\bsgamma_2}{m+n}\bigg) \\ 
&+ \frac{1}{2\pi\im}\sum_{m,n\in\bbZ_{\geq0}}\int_{\cL_0}\dif z\,f[m,n](z)\cD_{m,n}\bseta_{12}\,.
\eea
Here $\bsgamma=\p^{-1}\bseta\in\Omega^{1,\bullet}(\bbPT)[1]$ and the indices on $\bsgamma,\bseta$ refer to their holomorphic components. Note that the coupling to $e$ respects the ambiguity in $\bsgamma$ under addition of exact terms.

BRST invariance of the coupled bulk-defect system determines the OPEs of the $e,f$ states. At tree level we can simply take the variation of the defect \eqref{eq:axion-defect} directly, and cancel against the linearised variation of a bilocal term on the defect. Example calculations of this type are included in appendix \ref{app:axion-OPEs}. There are two new classes of OPEs. The first arise in the case of no counterterm coupling, i.e., $\mu=0$, and are given by
\begin{subequations} \label{eq:OPEs-tree-ef}
\begin{align}
&w[p,q](z)e[r,s](0)\sim\frac{ps-qr}{z}e[p+r-1,q+s-1](0)\,, \label{eq:OPE-tree-we} \\
\begin{split}
&w[p,q](z)f[r,s](0)\sim\frac{ps-qr}{z}f[p+r-1,q+s-1](0) \\ &+\frac{p+q}{r+s+2}\bigg(\frac{1}{z^2}e[p+r,q+s](0)+ \frac{1}{z}\frac{p+q-2}{p+q+r+s}\p_ze[p+r,q+s](0)\bigg)\,, \label{eq:OPE-tree-wf}
\end{split} \\
&e[p,q](z)f[r,s](0)\sim\frac{ps-qr}{z}\tilde w[p+r-1,q+s-1](0)\,, \label{eq:OPE-tree-ef} \\
\begin{split}
&f[p,q](z)f[r,s](0)\sim\frac{1}{z^2}\bigg(2 + \frac{p+q}{r+s+2} + \frac{r+s}{p+q+2}\bigg)\tilde w[p+r,q+s](0) \\ 
&+ \frac{1}{z}\bigg(1 + \frac{p+q}{r+s+2}\bigg)\p_z\tilde w[p+r,q+s](0)\,. \label{eq:OPE-tree-ff}
\end{split}
\end{align}
\end{subequations}
Together with the standard classical $w,\tilde w$ OPEs appearing in equation \eqref{eq:chiral-algebra} these define an extended CCA with associative operator product.\\

In the quantum theory we're required to switch on the counterterm coupling $\mu$ in order to cancel the Poisson-BF anomaly. This induces a second class of counterterm OPEs
\begin{subequations} \label{eq:OPEs-counter}
\begin{align}
\begin{split}
&w[p,q](z)w[r,s](0) \sim \mu R_2(p,q,r,s)\bigg(\frac{1}{z^2}e[p+r-2,q+s-2](0) \\ 
&+ \frac{1}{z}\frac{p+q-2}{p+q+r+s-4}\p_ze[p+r-2,q+s-2](0)\bigg) - \frac{\mu R_3(p,q,r,s)}{z}f[p+r-3,q+s-3](0)\,, \label{eq:OPE-counter-ww}
\end{split} \\
&w[p,q](z)e[r,s](0) \sim - \frac{\mu R_3(p,q,r,s)}{z}\tilde w[p+r-3,q+s-3](0)\,, \label{eq:OPE-counter-we} \\
\begin{split}
&w[p,q](z)f[r,s](0) \sim - \mu R_2(p,q,r,s)\bigg(\frac{1}{z^2}\frac{p+q+r+s}{r+s+2}\tilde w[p+r-2,q+s-2](0) \\
&+ \frac{1}{z}\p_z\tilde w[p+r-2,q+s-2](0)\bigg)\,, \label{eq:OPE-counter-wf}
\end{split}
\end{align}
\end{subequations}
for $R_\ell(p,q,r,s)$ as defined in \eqref{eq:Rpqrs} intertwining $\fsl_2(\bbC)_+$ representations
\be (\mathbf{p+q+1})\otimes(\mathbf{r+s+1})\to (\mathbf{p+q+r+s+1-2\boldsymbol{\ell}})\,.
\ee
Ignoring loop contributions, the counterterm OPEs \eqref{eq:OPEs-counter} break associativity of the operator product. However, if $\mu$ obeys equation \eqref{eq:tune-mu} we should obtain a consistent quantum extended CCA.\\

We emphasise that although we computed these OPEs using holomorphic surface defects on twistor space, they describe the holomorphic collinear singularities of tree amplitudes involving the axion in 1\textsuperscript{st}-order deformations of the theory \eqref{eq:SDGR+axion-again}. This is because the tree OPEs simply encode the vertices appearing the action, and the results of \cite{Bittleston:2022nfr} show that all non-trivial trees in the theory itself vanish.\\

Unfortunately there is a catch. Since the 4\textsuperscript{th}-order axion $\rho$ has a classical coupling to SD gravity, it can run through loops. The chiral algebra therefore receives new quantum corrections.


\subsection{1-loop corrections to operator products from the axion} \label{subsec:axion-quantum-corrections}

Unlike in the case of SD Yang-Mills, coupling SD gravity to the 4\textsuperscript{th}-order axion $\rho$ introduces new 1-loop corrections to OPEs, which are important to account when checking that associativity is restored. They can be constrained by symmetry as in subsection \ref{subsec:possible-corrections}. From the defect perspective we expect they can be attributed to anomalous 1-loop diagrams in the bulk-defect system: an example diagram is illustrated in figure \ref{fig:BRST-1loop-eta}. We will be less exhaustive in determining the possible 1-loop corrections to OPEs than in subsection \ref{subsec:possible-corrections}, concentrating only on those which are needed to verify that the associativity failure identified in proposition \ref{prop:associativity-failure} is cured.\\

The symmetries we use to constrain corrections are the same as those in subsection \ref{subsec:possible-corrections}. Reintroducing the parameter $\hbar$, we can simultaneously rescale $\hbar\mapsto\lambda\hbar$, $\bh\mapsto\bh$, $\bg\mapsto\lambda\bg$ and $\bseta\mapsto\lambda^{1/2}\bseta$ for $\lambda\in\bbC^*$. With this assignment, the coupling $\mu$ scales as $\mu\mapsto\lambda^{1/2}\mu$, consistent with its interpretation as a counterterm. We then further impose conformal invariance on the celestial sphere, $\gSL_2(\bbC)_+$ invariance and dimension matching.

\begin{figure}[h!]
	\centering
	\includegraphics[scale=0.3]{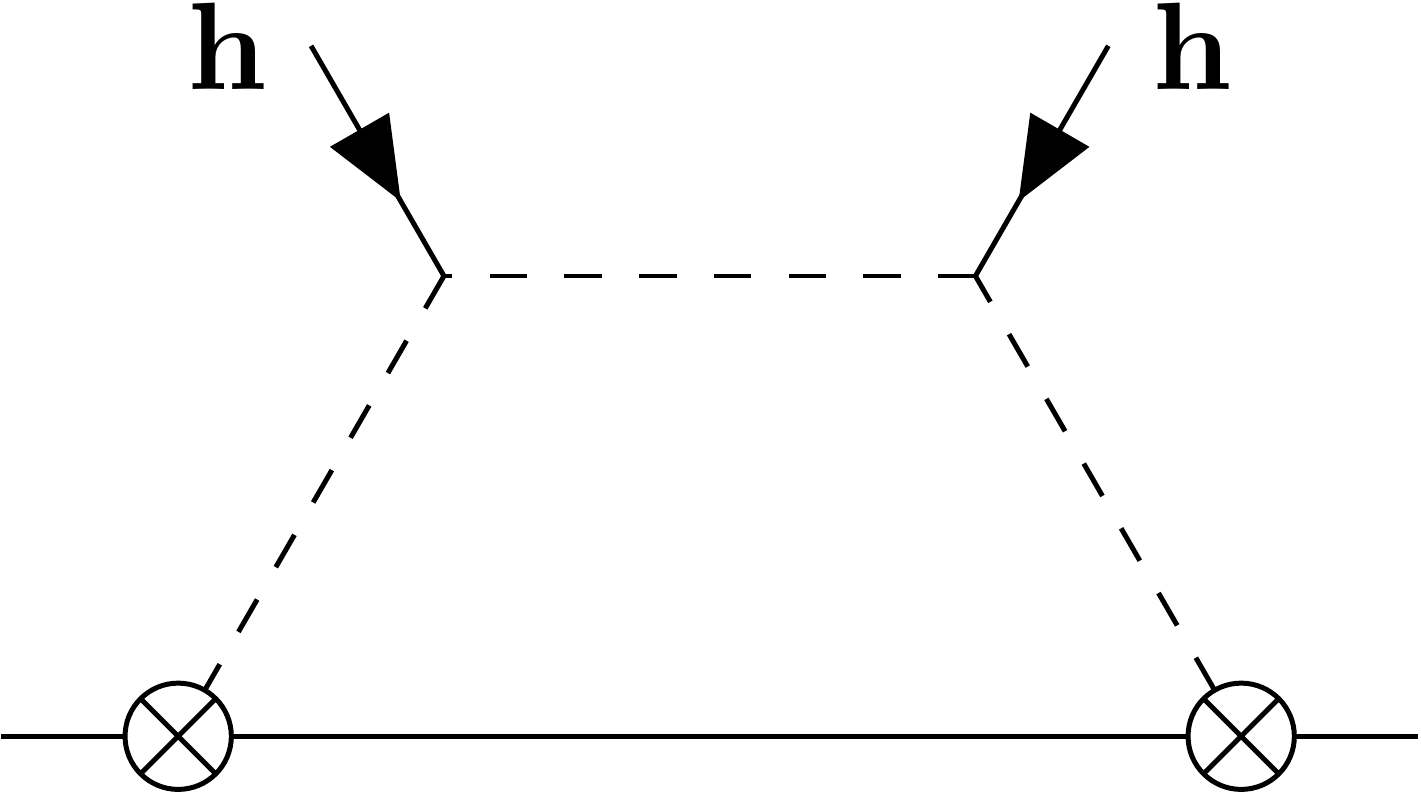}
	\caption{\emph{1-loop Feynman diagram in the bulk-defect system whose BRST variation necessitates corrections to the $w,w$ OPEs. Dotted lines designate $\bseta$ propagators, and crossed dots represent couplings of $\bseta$ to the $e,f$ generators.}}
	\label{fig:BRST-1loop-eta}
\end{figure}

The most important 1-loop corrections are to the OPEs of $w[m,n]$ generators with $m+n=4$. They are determined by
\bea
&w[4,0](z)w[0,4](0)\sim \gamma\bigg(\frac{1}{z^2}{\tilde w}[0,0] + \frac{1}{2z}\p_z\tilde w[0,0]\bigg)(0) \\
&+ \frac{1}{z}\big(2\delta^{3,1}_{5,5}e[1,1]f[0,0] + \delta^{2,2}_{5,5}(\normord{e[1,0]f[0,1]} + \normord{e[0,1]f[1,0]})\big)(0) + \dots
\eea
for constants $\gamma,\delta^{3,1}_{5,5},\delta^{2,2}_{5,5}\in\bbC$. Here $+\dots$ denotes terms which are quadratic in $e$ states and won't play a role in our associativity check. From the defect perspective, the potential double pole arises because the diagram in figure \ref{fig:BRST-1loop-eta} can couple to generators $e,f$ which themselves have the non-trivial tree OPEs given in equation \eqref{eq:OPE-tree-ef}.

In section \ref{sec:splitting} a double pole of this type was identified with the 1-loop holomorphic splitting amplitude. Here we're seeing that the axion $\rho$ is itself contributing to the splitting amplitude by running through the loop. This presents an added complexity as compared to the situation in SD Yang-Mills \cite{Costello:2022upu}. We wish to show that associativity is restored for $\alpha=2/5\pi^2$, but the coefficient of the double pole is now $\alpha + \gamma$. In subsection \ref{subsec:correct-associativity} we will explain how $\alpha,\gamma$ are related.\\

The remaining 1-loop corrections which must be accounted for are as follows. As was the case in subsection \ref{subsec:possible-corrections}, OPEs involving $w[m,n]$ for $m+n\leq2$ are not modified. The first relevant correction is determined by
\bea \label{eq:first-axion-simple-pole}
&w[3,0](z)w[0,3](0)\sim \varepsilon^{2,2}_{4,4}\bigg(\frac{2}{z^2}e[1,0]e[0,1] + \frac{1}{z}\p_z(e[1,0]e[0,1])\bigg)(0) \\
&+ \frac{\zeta^{2,2}_{4,4}}{z}(e[1,0]\p_ze[0,1] - e[0,1]\p_ze[1,0])(0)
\eea
for $\varepsilon^{2,2}_{4,4},\zeta^{2,2}_{4,4}\in\bbC$. Note that since the $e$ states have trivial OPEs among themselves, there is no need for normal ordering on the r.h.s. The next is to the OPEs between $w[p,q]$ for $p+q=4$ and $w[r,s]$ for $r+s=3$. Fortunately, only the corrections which factor through the $\mathfrak{sl}_2(\bbC)_+$ representation $\mathbf{4}$ play a role, and these are determined by
\bea
&w[2,2](z)w[3,0](0)\sim \frac{\varepsilon^{3,2}_{5,4}}{z^2}e[2,0]e[1,0](0) + \frac{1}{z}\big(\zeta^{3,2}_{5,4}e[2,0]\p_ze[1,0] + \zeta^{2,3}_{5,4}e[1,0]\p_ze[2,0]\big)(0) \eea
for constants $\varepsilon^{3,2}_{5,4},\zeta^{3,2}_{5,4},\zeta^{2,3}_{5,4}\in\bbC$. Finally, we'll also need
\bea
&w[4,1](z)w[0,3](0)\sim \frac{\delta_{6,4}^{3,1}}{z}e[1,1]f[0,0](0) + \frac{\delta_{6,4}^{2,2}}{z}(\normord{e[1,0]f[0,1]} + \normord{e[0,1]f[1,0]})(0) + \dots \eea
for $\delta^{3,1}_{6,4},\delta^{2,2}_{6,4}\in\bbC$. Here $+\dots$ denotes terms quadratic in $e$ which will not feature in the calculation. In particular, quantum corrections to the above OPE cannot introduce a double pole: on dimensional grounds it would have to accompany $\tilde w[0,0]$, but this transforms trivially under $\fsl_2(\bbC)_+$.\\

Before continuing, we emphasise that the $w,w$ OPEs are not the only ones to be corrected at 1-loop. The $w,e/f$ and $e/f,f$ OPEs will also be modified. There may also be 1-loop diagrams involving the counterterm coupling, but since these should be viewed as higher order quantum corrections.


\subsection{Correcting the failure of associativity} \label{subsec:correct-associativity}

We expect that SD gravity coupled to a 4\textsuperscript{th}-order gravitational axion admits a consistent quantum extended CCA. Here we verify that the associativity failure identified in subsection \ref{subsec:inconsistency} is cured.

\begin{proposition} \label{prop:associativity-check}
	Associativity of the operator product \eqref{eq:contour-identity} holds for the choice of operators \eqref{eq:operators-445} only if
	\be \alpha + \gamma = \frac{4}{5\pi^2}\,. \ee
\end{proposition}

\begin{proof}
For simplicity we'll evaluate both sides of equation \eqref{eq:contour-identity} modulo terms quadratic in the $e$ states. We should not, however, ignore terms of this type which appear at intermediate steps of the calculation, though in practice these do not contribute.

Let's proceed by evaluating the l.h.s of \eqref{eq:contour-identity}. The interior integral gives
\bea \label{eq:corrected-associativity-12}
&\frac{1}{2\pi\im}\oint_{|z_{12}|=1}\dif z_{12}\,z_2w[3,0](z_1)w[0,3](z_2) 
= z_2\big(9w[2,2] + \beta_{4,4}w[0,0]\tilde w[0,0] + 18\mu\p_ze[1,1] \\
&- 36\mu f[0,0] + \varepsilon_{4,4}^{2,2}\p_z(e[1,0]e[0,1]) + \zeta_{4,4}^{2,2}(e[1,0]\p_ze[0,1] - e[0,1]\p_ze[1,0])\big)(z_2)\,,
\eea
where the first term is the classical OPE \eqref{eq:chiral-algebra}, the second is a 1-loop correction in the pure gravitational theory \eqref{eq:first-simple-pole}, the next two terms follow from counterterm OPEs \eqref{eq:OPEs-counter}, and the final two terms are 1-loop corrections involving the axion states \eqref{eq:first-axion-simple-pole}. Since we've yet to take the second OPE, we cannot discard the terms quadratic in $e$. Next we perform the exterior integral, which extracts the coefficient of the simple pole in the OPE of \eqref{eq:corrected-associativity-12} with $w[2,2](0)$. The only non-vanishing contributions, modulo terms quadratic in $e$, are
\bea
&z_2w[2,2](z_2)w[2,2](0)\sim -\frac{24\mu}{z_2}e[2,2](0) + \frac{\alpha + \gamma}{6z_2}\tilde w[0,0](0)\,, \\
&z_2f[0,0](z_2)w[2,2](0)\sim \frac{2}{z_2}e[2,2](0)\,.
\eea
The l.h.s. of equation \eqref{eq:contour-identity} therefore evaluates to
\be
\label{eq:lhs-442} \frac{3}{2}(\alpha + \gamma)\tilde w[0,0](0) - 288\mu e[2,2](0) + \dots
\ee

Turning our attention to the first term on the r.h.s of equation \eqref{eq:contour-identity}, the interior integral evaluates to
\be \frac{1}{2\pi\im}\oint_{|z_2|=1}\dif z_2\,z_2w[0,3](z_2)w[2,2](0) = 12\mu e[0,3](0) + \varepsilon_{5,4}^{2,3}e[0,1]e[0,2](0)\,, \ee
where the only contributions are from the counterterm and 1-loop corrected axion OPEs. The exterior integral then extracts the residue in the OPE of $w[3,0](z_1)$ with the above. Modulo terms quadratic in $e$, the only contribution is from
\be w[3,0](z_1)e[0,3](0)\sim \frac{9}{z_1}(e[2,2] - 4\mu\tilde w[0,0])(0)\,, \ee
so the first term on the r.h.s. is
\be \label{eq:rhs1-442} 108\mu e[2,2](0) - 432\mu^2\tilde w[0,0](0) + \dots\,. \ee

At last we arrive at final term in equation \eqref{eq:contour-identity}. The interior integral gives
\bea &\frac{1}{2\pi\im}\oint_{|z_1|=1}\dif z_1\,w[3,0](z_1)w[2,2](0) \\
&=6w[4,1](0) + 4\mu\p_ze[3,0](0) \\ 
&+ (\varepsilon_{5,4}^{3,2} - \zeta_{5,4}^{3,2})e[2,0]\p_ze[1,0](0) + (\varepsilon_{5,4}^{3,2} -\zeta_{5,4}^{2,3})e[1,0]\p_ze[2,0](0)\,,
\eea
where the classical, counterterm and 1-loop corrected axion OPEs all contribute. The exterior integral over $z_2$ extracts the simple pole in the OPE of $z_2w[0,3](z_2)$ with the above. Working again modulo terms quadratic in $e$ states, the only non-vanishing contributions are
\bea
&z_2w[0,3](z_2)w[4,1](0)\sim \frac{72\mu}{z_2}e[2,2](0)\,, \\
&z_2w[0,3](z_2)\p_ze[3,0](0)\sim - \frac{9}{z_2}e[2,2](0) + \frac{36\mu}{z_2}\tilde w[0,0](0)\,.
\eea
Therefore the final term in equation \eqref{eq:contour-identity} evaluates to
\be \label{eq:rhs2-442} 396\mu e[2,2](0) + 144\mu^2\tilde w[0,0](0) + \dots\,. \ee

Collecting equations \eqref{eq:lhs-442}, \eqref{eq:rhs1-442} and \eqref{eq:rhs2-442} we learn that
\bea
&\frac{3}{2}(\alpha + \gamma)\tilde w[0,0](0) - 288\mu e[2,2](0) \\
&= 108\mu e[2,2](0) - 432\mu^2\tilde w[0,0](0) - 396\mu e[2,2](0) - 144\mu^2\tilde w[0,0](0)\,. \eea
The terms involving $e[2,2]$ cancel, as is expected from classical BRST invariance of the twistorial theory at first order in $\mu$. In order for the $\tilde w[0,0]$ terms to cancel we must have
\be \alpha + \gamma = - 384\mu^2\,. \ee
As explained in subsection \ref{subsec:quantum-SDGR}, cancellation of the twistorial anomaly requires
\be \mu^2 = \frac{1}{5!}\bigg(\frac{\im}{2\pi}\bigg)^2\,. \ee
So we find that the failure of associativity identified in \ref{subsec:inconsistency} is remedied only if
\be \label{eq:associativity-constraint} \alpha + \gamma = \frac{4}{5\pi^2} \ee
as claimed.
\end{proof}

We note that the real difficulty in the above calculation, which has largely been suppressed, lies in showing that 1-loop corrections aside from the double pole do not contribute to the final answer.

Repeating the above computation for the choice of operators in equation \eqref{eq:operators-553} we find that associativity of the operator product implies
\be \beta^{2,2}_{5,5} + \delta^{2,2}_{5,5} = 2(\alpha+\gamma) + \frac{6}{5\pi^2} = \frac{14}{5\pi^2}\,. \ee
An important consequence of this is that in order for quantum extended CCA to be associative it must be non-linear.\\

All that remains to be done is to separate the coefficients $\alpha,\gamma$. On spacetime this can be understood as determining the relative contributions of gravitons and 4\textsuperscript{th}-order axions to the 1-loop holomorphic splitting amplitude.

Imagine coupling SD gravity to two 4\textsuperscript{th}-order scalars, and switching on a gravitational axion coupling for one of them. On twistor space they are represented by $\bseta_1,\bseta_2\in\Omega^{2,\bullet}_\mathrm{cl}(\bbPT)[1]$, and the non-vanishing counterterm coupling is $\mu_1$. The gravitons and 4\textsuperscript{th}-order scalars all contribute to the 1-loop 4-point all-plus amplitude, or equivalently the twistorial gravitational anomaly. To cancel it, we must tune \cite{Bittleston:2022nfr}
\be \label{eq:mu1} \mu_1^2 = \frac{3}{2}\mu^2 = \frac{3}{5!\,2}\bigg(\frac{\im}{2\pi}\bigg)^2\,. \ee
The chiral algebra includes two copies of the axion states, $e_i,f_i$ for $i=1,2$. These have the tree OPEs \eqref{eq:OPEs-tree-ef} and separately induce the 1-loop corrections from subsection \ref{subsec:axion-quantum-corrections}. However, only the $e_1,f_1$ states have the counterterm OPEs from equation \eqref{eq:OPEs-counter} with parameter $\mu_1$. In particular we have the 1-loop correction
\be w[4,0](z)w[0,4](0) \sim (\alpha + 2\gamma)\bigg(\frac{1}{z^2}{\tilde w}[0,0] + \frac{1}{2z}\p_z\tilde w[0,0]\bigg)(0) + \dots \ee
where $+\dots$ hides the bilinear terms.

Proposition \ref{prop:associativity-check} shows that associativity of the operator product now requires
\be \alpha + 2\gamma = -384\mu_1^2 = \frac{6}{5\pi^2}\,. \ee
Comparing to \eqref{eq:associativity-constraint} we conclude that
\be \alpha = \gamma = \frac{2}{5\pi^2}\,. \ee
This matches the value obtained in sections \ref{sec:splitting} and \ref{sec:defect}. Similarly demanding associativity for the choice of operators in equation \eqref{eq:operators-553} shows that $\beta^{2,2}_{5,5} = \delta^{2,2}_{5,5}$. We infer that $\beta_{5,5} = 7/5\pi^2$, the value we found by direct calculation in subsection \ref{subsec:evaluate-diagram}.\\

In summary, the associativity failure identified in subsection \ref{subsec:inconsistency} is not present in SD gravity coupled to a 4\textsuperscript{th}-order gravitational axion, so long as the counterterm coupling is tuned to cancel the twistorial gravitational anomaly. This has the effect of cancelling the 1-loop all-plus amplitudes, which are responsible for non-universal holomorphic collinear singularities of amplitudes in 1\textsuperscript{st}-order deformations of the theory.


\subsection{Subleading terms in 1-loop operator products} \label{subsec:subleading}

We've seen that in order for the quantum extended CCA of SD gravity coupled to a 4\textsuperscript{th}-order gravitational axion to be associative, its 1-loop corrected operator products must involve subleading simple poles which are bilinear in generators. These cannot be interpreted straightforwardly as splitting amplitudes, so what do they represent physically?\\

From the discussion in \ref{subsec:interpretation}, they should describe the subleading holomorphic collinear singularities of amplitudes in 1\textsuperscript{st}-order deformations \emph{of the axion coupled theory}. In particular, they will be visible in the holomorphic collinear limits of 1-minus amplitudes in the infinitesimal deformation to full Einstein gravity.\\

It would be interesting to verify this explicitly. Particularly given that the poles beneath double poles present a significant complication in the application of recursion methods to 1-loop graviton amplitudes \cite{Brandhuber:2007up,Dunbar:2010xk,Alston:2012xd,Alston:2015gea}.


\subsection{Alternative methods of anomaly cancellation} \label{subsec:alternatives}

An alternative means of cancelling the twistorial anomaly, or equivalently the 1-loop all-plus amplitudes, is by coupling SD gravity to appropriate matter.\\

Suppose we couple to SD Yang-Mills with gauge group $\fg$, Weyl fermions in the representation $R_f$ and quadratic scalars in the representation $R_s$. Then for the twistorial gravitational anomaly to vanish we must have
\be \label{eq:gravitational-coefficient} \dim R_s - 2\dim R_f + 2\dim \fg + 2 = 0\,. \ee
However, by coupling to SD Yang-Mills we introduce potential twistorial gauge and mixed gauge-gravitational anomalies. For these to vanish we require
\bea \label{eq:gauge-coefficients}
&\tr_{R_s}(X^4) - 2\tr_{R_f}(X^4) + 2\tr_\fg(X^4) = 0\,, \\
&\tr_{R_s}(X^2) - 2\tr_{R_f}(X^2) + 2\tr_\fg(X^2) = 0
\eea
respectively for all $X\in\fg$ \cite{Costello:2021bah,Bittleston:2022nfr}. Of course, the chiral anomaly on spacetime must not be present either. If all of these conditions hold, then we expect that the spacetime theory admits a consistent quantum extended CCA. Other important examples of anomaly free models are SD supersymmetric theories, although we emphasise that anomaly cancellation is a weaker constraint than supersymmetry.\footnote{Even in the supersymmetric case, the extended CCA describes the holomorphic collinear singularities of amplitudes in non-supersymmetric deformations.}\\

Supersymmetric Ward identities necessitate that the all- and mostly-plus graviton amplitudes in theories of gravity coupled to matter are proportional to the l.h.s. of equation \eqref{eq:gravitational-coefficient} \cite{Grisaru:1976vm,Grisaru:1977px,Parke:1985pn}. Given that the 1-loop holomorphic splitting amplitude identified in section \ref{sec:splitting} can be obtained from the holomorphic collinear limit of the 5-point mostly-plus amplitude, in a general theory of SD gravity coupled to matter it's given by
\be \mathrm{Split}^\mathrm{1-loop}_+(1^+,2^+;1/2) = \frac{(\dim R_s-2\dim R_f+2\dim\fg+2)}{360(4\pi)^2}\frac{[12]^4}{\la12\ra^2}\,. \ee
If equation \eqref{eq:gravitational-coefficient} holds then the 1-loop holomorphic splitting amplitude vanishes, and the 1-loop corrected $w,w$ OPEs are free from double poles.\footnote{We emphasise that in the case of SD QCD vanishing of the twistorial gauge anomaly does not necessarily remove double poles from the 1-loop corrected chiral algebra \cite{Costello:2022upu}. Curiously, the double poles in the 1-loop OPEs of gluon states are actually proportional to the coefficient of the twistorial mixed gauge-gravitational anomaly.} The associativity failure identified in proposition \ref{prop:associativity-failure} is proportional to the coefficient of the double poles, so is remedied. It's natural to ask whether the $w,w$ OPEs are deformed at all? Our Feynman diagram computation in section \ref{sec:defect} suggests that the bilinear terms in 1-loop $w,w$ OPEs are still present. Since these can in principle induce associativity failures, it would be worthwhile understanding why cancellation of the twistorial anomaly is sufficient to ensure they do not.


\section{Discussion} \label{sec:discussion}

In this work we've explored whether the holomorphic collinear singularities of amplitudes in 1\textsuperscript{st}-order deformations of SD gravity define a consistent chiral algebra. At tree level this is certainly the case, but at 1-loop its OPEs are deformed. The simplest corrections can be attributed to the 1-loop effective vertex describing double poles in 1-loop graviton amplitudes.

We found associativity of the operator product was violated in the 1-loop deformation of the extended CCA. This failure signals that the holomorphic collinear singularities of amplitudes in 1\textsuperscript{st}-order deformations of SD gravity are not universal, i.e., they depend on the choice of deformation. The non-vanishing 1-loop all-plus amplitudes are responsible for such non-universalities.

The universal holomorphic surface defect in the twistor formulation of classical SD gravity can be identified with the classical extended CCA \cite{Costello:2022wso}. We argued that this applies equally in the quantum setting. Anomalous 1-loop diagrams in the coupled bulk-defect system necessitate corrections to OPEs, which we showed by direct computation match those induced by the 1-loop effective vertex. These 1-loop corrections also involve subleading simple poles which are bilinear in generators. From this perspective, the associativity failure can be traced to the recently discovered gravitational anomaly in the twistor formulation of SD gravity \cite{Bittleston:2022nfr}.

The anomaly can be cancelled via a kind of Green-Schwarz mechanism by coupling to a $\p$-closed $(2,1)$-form field on twistor space describing a 4\textsuperscript{th}-order gravitational axion on spacetime. Incorporating the states of this new field into the chiral algebra, we found that the previously identified failure of associativity was remedied. We also briefly discussed cancelling the twistorial anomaly by coupling to suitable matter.\\

Thus far we have only considered the simplest quantum corrections to extended CCAs. Whilst these are sufficient to see the failure of associativity, it would be interesting to determine the 1-loop corrected extended CCA for a theory of SD gravity coupled to matter for which the twistorial gravitational anomaly vanishes, and perhaps to go beyond 1-loop.\\

There are also number of natural questions raised by this work:
\begin{itemize}
	\item[-] In \cite{Costello:2022wso} (see also \cite{Bu:2022dis}) a correspondence between BRST invariant local operators in SD Yang-Mills coupled to a 4\textsuperscript{th}-order axion and the conformal blocks of its extended CCA was demonstrated. Furthermore, amplitudes in the presence of a local operator were proven to be equal to correlators in the corresponding conformal block. In this way the authors were able to recover MHV, NMHV and 1-loop all-plus amplitudes in Yang-Mills theory. In the gravitational case we've argued that the natural counterparts of local operators are 1\textsuperscript{st}-order deformations of the theory. These are determined by 4-form local operators whose BRST variation is de Rham exact. Applying anomaly ascent gives local operators of higher ghost number. It would be interesting to extend the correspondence of \cite{Costello:2022wso} to such operators, and to compute graviton amplitudes using the quantum extended CCA.
	\item[-] In recent works \cite{Monteiro:2022lwm,Bu:2022iak,Guevara:2022qnm} an integrable deformation of SD gravity has been considered in which (the loop algebra of) $\mathrm{Ham}(\bbC^2)$ is deformed to (the loop algebra of) the symplecton. Whilst this is not the quantum deformation studied here, it would be fascinating to understand the full space of deformations of the extended CCA. Another option is to turn on a cosmological constant, which on twistor space amounts to deforming the contact structure so that it becomes non-degenerate \cite{Ward:1980am,Mason:2007ct}.
	\item[-] One difference between the extended CCAs we've considered here and the CCAs appearing in much of the literature is that they incorporate two towers of states, corresponding to positive- and negative-helicity modes. It's curious that there does exist a natural candidate for a twistorial theory whose extended CCA consists of just one tower of states, describing only positive-helicity modes. This is the theory of the K\"{a}hler scalar, which is believed to be the target space description of the $\cN=2$ string \cite{Ooguri:1990ww,Ooguri:1991fp}. We expect it's described by Poisson-Chern-Simons theory on twistor space. At the quantum level this theory suffers from an anomaly, and its extended CCA will also be deformed. In particular there is an anomalous 1-loop diagram involving holomorphic surface defects of the same form as figure \ref{subfig:box}. Unfortunately at this time we have no means of cancelling the anomaly, and hence correcting associativity of the chiral algebra.
	\item[-] $M$-theory in a twisted $\Omega$-background is described by a 5d non-commutative mixed topological-holomorphic Chern-Simons theory \cite{Costello:2016nkh}, which is closely related to Poisson-Chern-Simons. In the twisted $\Omega$-background M5 branes become holomorphic surface operators supporting generalized $W_{1+\infty}$ algebras. On the other hand, M2 branes are described by topological line operators supporting closely related associative algebras \cite{Costello:2017fbo,Gaiotto:2019wcc,Oh:2020hph,Gaiotto:2020vqj}. It's possible that these results could be leveraged to obtain the extended CCA of an anomaly free theory of SD gravity with matter.
	\item[-] A remarkable new holographic duality between the 4d WZW model in asymptotically flat Burns space and a particular 2d chiral algebra has recently been obtained in \cite{Costello:2022jpg}. It would be intriguing to find a similar duality for the theory of the K\"{a}hler scalar mentioned above. One candidate for the chiral algebra side of the duality is the twistor $\sigma$-model of \cite{Adamo:2021bej}.
\end{itemize}


\acknowledgments 
 
It is a pleasure to thank K. Costello, S. Heuveline, A. Sharma and D. Skinner for many helpful discussions, and K. Costello, N. Paquette, A. Sharma and A. Strominger for comments on a draft version of this paper. Research at Perimeter Institute is supported in part by the Government of Canada through the Department of Innovation, Science and Economic Development and by the Province of Ontario through the Ministry of Colleges and Universities.


\begin{appendix}
	

\section{Holomorphic collinear limit of a 1-loop mostly-plus amplitude} \label{app:5pt-amplitude}

In this appendix we show that at leading order in the holomorphic collinear limit $\la23\ra\to0$ the 1-loop 5-point mostly-plus graviton amplitude factorizes as
\be \cM^\mathrm{1-loop}_5(1^-;2^+,3^+,4^+.5^+) \sim \mathrm{Split}^\mathrm{1-loop}_+(2^+,3^+;t)\cM^\mathrm{tree}_4(1^-,P^-_{23};4^+,5^+) \ee
for the 1-loop holomorphic splitting amplitude given in equation \eqref{eq:holomorphic-splitting}.\\

The motivation for considering this amplitude is as follows: it's well known that the all-plus and mostly-plus tree amplitudes vanish in Einstein gravity.\footnote{There is a caveat here: in the complexified or ultrahyperbolic setting there's a non-vanishing 1-plus, 2-minus tree with $[12]=[23]=[31]=0$ but angle brackets non-vanishing \cite{Dixon:2013uaa}. It might therefore seem natural to consider the holomorphic collinear limits of the 1-loop mostly-plus 4-point amplitude. However, at 4-points momentum conservation necessitates both $\la12\ra[12]=\la34\ra[34]$ and $\la12\ra[24]+\la13\ra[34]=0$, suggesting that if we take $\la12\ra\to0$ we must also send $[34]\to0$.} The simplest non-vanishing tree is therefore the 2-plus, 2-minus amplitude
\be \label{eq:4pt-MHV} \cM^\mathrm{tree}_4(1^-,2^-;3^+,4^+) = -\im\frac{\la12\ra^6[34]}{\la13\ra\la14\ra\la 23\ra\la 24\ra\la34\ra}\,. \ee
The simplest 1-loop amplitude whose holomorphic collinear limits see the 1-loop holomorphic splitting amplitude will have one extra negative-helicity state compared to the above, and so is 5-point mostly-plus amplitude. The 1-loop splitting amplitude in Yang-Mills theory was first computed in \cite{Bern:1994zx} by analysing the 1-loop amplitude with precisely this helicity configuration.\\

The 1-loop mostly-plus 5-point amplitude was first computed in \cite{Dunbar:2010xk}, and is easiest to write down in pieces. First the amplitude is expressed as a sum over cyclic permutations of three positive-helicity legs
\be \label{eq:1-4+1loop}
\cM^\mathrm{1-loop}_5(1^-;2^+,3^+,4^+,5^+) = \frac{\im}{(4\pi)^2}\sum_{\sigma\in C(\{3,4,5\})}\cR(1;2,\sigma(3),\sigma(4),\sigma(5))\,.
\ee
Each term in this sum is then a sum of three further terms
\be
\cR(1;2,3,4,5) = \sum_{i=1}^3\cR^{(i)}(1;2,3,4,5)\,,
\ee
each corresponding to a different class of diagrams in \cite{Dunbar:2010xk}. The first and second of these are relatively straightforward
\begin{align}
	& \label{eq:R1} \cR^{(1)}(1;2,3,4,5) = - \frac{1}{180}\frac{\la14\ra^2\la15\ra^2[23][45]^4(\la34\ra^2\la15\ra^2 + \la13\ra\la34\ra\la45\ra\la15\ra + \la13\ra^2\la45\ra^2)}{\la12\ra^2\la23\ra\la35\ra^2\la34\ra^2\la45\ra^2}\,, \\
	&\label{eq:R2} \cR^{(2)}(1;2,3,4,5) = \frac{1}{60}\frac{\la15\ra[25]^4([13]^2[45]^2 + [23][34][45][25] + [34]^2[25]^2)}{\la34\ra^2[12]^2[15]}\,.
\end{align}
Whereas the third is a little more involved. Writing
\bea \label{eq:Delta}
&\Delta(1;2,3,4,5) = \frac{\la14\ra\la23\ra}{\la12\ra\la34\ra} + \frac{\la15\ra\la23\ra}{\la12\ra\la35\ra} + 6\frac{[42][52]\la23\ra\la45\ra}{\la43\ra\la53\ra[23][45]} + 6\frac{[43][53]\la23\ra\la45\ra\la31\ra^2}{\la43\ra\la53\ra[23][45]\la21\ra^2} \\
&+ 7\frac{[43][52]\la23\ra\la45\ra\la31\ra}{\la43\ra\la53\ra[23][45]\la21\ra} + 7\frac{[42][53]\la23\ra\la45\ra\la31\ra}{\la43\ra\la53\ra[23][45]\la21\ra}\,,
\eea
we have
\be \label{eq:R3} \cR^{(3)}(1;2,3,4,5) = -\frac{1}{180}\frac{\la12\ra^2\la13\ra^4[23]^4[45]}{\la14\ra\la15\ra\la23\ra^2\la34\ra\la35\ra\la45\ra}\bigg(1-\frac{\Delta(1;2,3,4,5)}{2}\bigg)\,. \ee
This expression for the amplitude is not manifestly invariant under permutations of the external positive-helicity states. It does, however, make taking the holomorphic collinear limit $\la23\ra\to0$ fairly easy. We work through the terms in the sum \eqref{eq:1-4+1loop} in sequence, starting with the trivial permutation.\\

In the holomorphic collinear limit of $\cR^{(1)}(1;2,3,4,5)$ we get a decomposition
\be -\frac{1}{4t(1-t)}\frac{[23]}{\la23\ra}\bigg(\frac{\la14\ra^2\la15\ra^2[45]^4(\la P_{23}4\ra^2\la15\ra^2 + \la1P_{23}\ra\la P_{23}4\ra\la45\ra\la15\ra + \la1P_{23}\ra^2\la45\ra^2)}{180\la1P_{23}\ra^2\la P_{23}4\ra^2\la 45\ra^2\la5P_{23}\ra^2}\bigg)\,.
\ee
The first factor is easily identified with the tree graviton splitting amplitude given in \cite{Bern:1998xc},
\be \mathrm{Split}^\mathrm{tree}_-(1^+,2^+;t) = - \frac{1}{4t(1-t)}\frac{[12]}{\la12\ra}\,. \ee
The second can be massaged into the 1-loop 4-point mostly-plus amplitude \cite{Bern:1993wt,Dunbar:1994bn}
\bea
&\cM^\mathrm{1-loop}_4(1^-;P_{23}^+,4^+,5^+) \\
&=\frac{\im}{(4\pi)^2}\frac{\la1P_{23}\ra^2\la14\ra^2\la15\ra^2[14]^2(\la1P_{23}\ra^2[1P_{23}]^2 + \la14\ra^2[14]^2 + \la15\ra[15]^2)}{360\la P_{23}4\ra^2\la45\ra^2\la5P_{23}\ra^4}\,.
\eea
This is the expected splitting in the true collinear limit. $\cR^{(2)}(1;2,3,4,5)$ is non-singular as $\la23\ra\to0$, leaving $\cR^{(3)}(1;2,3,4,5)$. Now all terms in $\Delta$ involve $\la23\ra$, so in the holomorphic limit the only double pole originates from the constant term in the bracket in equation \eqref{eq:R3}. This double pole is
\be \frac{4t(1-t)}{180}\frac{[23]^4}{\la23\ra^2}\bigg(-\frac{\la1P_{23}\ra^6[45]}{\la14\ra\la15\ra\la P_{23}4\ra\la P_{23}5\ra\la45\ra}\bigg)\,. \ee
The second factor matches the 4-point MHV amplitude in equation \eqref{eq:4pt-MHV}, and we identify the first with the 1-loop holomorphic splitting amplitude
\be \mathrm{Split}_+^\mathrm{1-loop}(1^+,2^+;t) = \frac{4t(1-t)}{180(4\pi)^2}\frac{[12]^4}{\la12\ra^2}\,. \ee
This is precisely the expression we obtain from the 1-loop effective vertex in subsection \ref{subsec:effective-vertex}.\\

We then move on to the remaining two terms in equation \eqref{eq:1-4+1loop}, corresponding to the permutations $\sigma = (345)$, $(354)$. The holomorphic limit $\la23\ra\to0$ in these terms corresponds to taking $\la25\ra\to0$, $\la24\ra\to0$ respectively in $\cR(1;2,3,4,5)$ prior to acting with the permutations. A glance at the expressions for $\cR^{(i)}(1;2,3,4,5)$ show that they have no singularities whatsoever in these limits.\\

The simple pole beneath the double pole is determined by $\Delta$ from equation \eqref{eq:Delta}. Relating this object to the simple poles in the 1-loop corrected OPEs of the chiral algebra would be of interest.
	

\section{More details on constraining quantum corrections} \label{app:extended-sl2C}

In this appendix we provide arguments constraining quantum corrections to the chiral algebra. In particular, we show that OPEs involving $w[p,q]$ for $p+q\leq2$ cannot be deformed, and further that the double poles are completely determined by the central term in equation \eqref{eq:first-correction}.\\

We find it convenient to introduce new notation $w^m_p = w[p,m-1-p]$, $\tilde w^m_p = \tilde w[p,m-1-p]$, so that $m$ is the dimension of the $\fsl_2(\bbC)_+$ representation in which the state transforms, and $p$ indexes a basis of this representation. We also suppress $z$ derivatives and dependence, writing the tree level OPEs schematically as
\be w^m_pw^n_r\sim K^{m,n}_{p,r}w^{m+n-2}_{p-1,r-1}\,,\qquad w^m_p\tilde w^n_r\sim K^{m,n}_{p,r}\tilde w^{m+n-2}_{p-1,r-1} \ee
for $K^{p+q+1,r+s+1}_{p,q} = ps-qr$. The poles on the right hand side are fixed by $z$ scaling symmetry. Consider the 1-loop corrections to the $w,w$ OPEs. From the discussion in subsection \ref{subsec:possible-corrections} they must take the schematic form
\be \label{eq:schematic} w^m_p w^n_r \sim X^{m,n;j}_{p,r;a}\,\tilde w^j_a +  Y^{m,n;k,\ell}_{p,r;b,c}\,:w^k_b\tilde w^\ell_c:\,, \ee
where summation convention over repeated indices is implicit on the right hand side. On dimensional grounds $j=m+n-9$, $k+\ell=m+n-6$. Now
\be \mathbf{m}\otimes\mathbf{n}\cong (\mathbf{|m-n|+1}) \oplus \dots \oplus (\mathbf{m+n-3}) \oplus (\mathbf{m+n-1})\,, \ee
so in order for $X^{m,n;j}_{p,r;a}$ to be non-vanishing we must have
\be m+n-9\geq|m-n|+1\,, \ee
i.e., $m,n\geq5$. This shows that linear deformations can arise only in the OPEs of generators $w[p,q]$ with $p+q\geq4$. Similarly,
\be \mathbf{k}\otimes(\mathbf{m+n-6-k})\cong (\mathbf{|m+n-6-2k|+1})\oplus\dots\oplus(\mathbf{m+n-7})\,, \ee
so in order for $Y^{m,n;k,\ell}_{p,r;b,c}$ to be non-vanishing we must have
\be m+n-7\geq|m+n-6-2k|+1\geq|m-n|+1\,, \ee
necessitating $m,n\geq4$. Hence deformations can arise only in the OPEs of generators $w[p,q]$ with $p+q\geq3$.

It's straightforward to apply these same arguments to the $w,\tilde w$ OPEs, and the 1-loop corrections involving the axion discussed in subsection \ref{subsec:axion-quantum-corrections}.\\

Now let's investigate the linear deformation from equation \eqref{eq:schematic} a little more carefully. We saw in subsection \ref{subsec:splitting-deformation} that the 1-loop holomorphic splitting amplitude induces terms of this type. We have already seen that on dimensional grounds it specialises to
\be \label{eq:schematic-linear} w^m_p w^n_r \sim X^{m,n;m+n-9}_{p,r;a}\,\tilde w^{m+n-9}_a\,. \ee
Invariance under the Cartan of $\fsl_2(\bbC)_+$ further restricts $m+n-10\geq a=p+r-4\geq0$. Full $\fsl_2(\bbC)_+$ invariance allows us to write
\be \label{eq:sl2C-invariance} X^{m,n;m+n-9}_{p,r;p+r-4} = \alpha_{m,n}R_4(p,m-1-p,r,n-1-r)\,, \ee
for $R_\ell(p,q,r,s)$ as in equation \eqref{eq:Rpqrs} vanishing for integers outside the range $p+q,p+r,s+q,s+r\geq\ell$.

To constrain the $\alpha_{m,n}$ note that $\fsl_2(\bbC)_+$ admits a split extension by the Heisenberg algebra generated by the zero-modes of $w[p,q]$ for $p+q\leq1$. Imposing invariance under this extended symmetry (or equivalently acting on both sides of \eqref{eq:schematic-linear} with $w[1,0],w[0,1]$) and substituting in equation \eqref{eq:sl2C-invariance} we find
\bea
&(p+r-4)R_4(p,m-p,r,n-r)\alpha_{m+1,n+1} \\
&= pR_4(p-1,m-p,r,n-r)\alpha_{m,n+1}
+ rR_4(p,m-p,r-1,n-r)\alpha_{m+1,n}\,.
\eea
Under the assumption that $\alpha_{m,n} = \alpha_{5,5}$ for $M+N>m+n$, $m,n\geq5$, we find that
\be
\alpha_{M,N} = \begin{cases}
\alpha_{M-1,N}\ \mathrm{for}\ M>5\,, \\
\alpha_{M,N-1}\ \mathrm{for}\ N>5\,, \end{cases}
\ee
so that inductively $\alpha_{m,n} = \alpha_{5,5}$ for all $m,n\geq 5$. This shows that linear deformations to the $w,w$ OPEs must take the schematic form
\be w^m_pw^n_r \sim \alpha_{5,5}R_4(p,m-1-p,r,n-1-r)\tilde w^{m+n-9}_{p+r-4}\,. \ee
In particular, the deformation
\be w[4,0](z)w[0,4](0)\sim \alpha\bigg(\frac{1}{z^2}\tilde w[0,0] + \frac{1}{2z}\p_z\tilde w[0,0]\bigg)(0) \ee
necessitates
\bea
&w[p,q](z)w[r,s](0) \\
&\sim \alpha\frac{R_4(p,q,r,s)}{(4!)^2}\bigg(\frac{1}{z^2}\tilde w[p+r-4,q+s-4] + \frac{1}{2z}\p_z\tilde w[p+r-4,q+s-4]\bigg)(0)\,,
\eea
as claimed in subsections \ref{subsec:splitting-deformation} and \ref{subsec:possible-corrections}.\\

The above argument can be adapted to show that $R_\ell(p,q,r,s)$ does not merely intertwine representations of $\fsl_2(\bbC)_+$, but also of its split extension by the Heisenberg algebra. This explains the ubiquity of this object in the OPEs of the chiral algebras appearing in the main text. Note that it already appears in the defining relations of $\mathrm{Ham}(\bbC^2)$ since $R_1(p,q,r,s) = ps-qr$.\footnote{It also appears in the defining relations of $W_\infty$ and its many variants. In the notation of \cite{Pope:1991ig} we have $R_\ell(p,q,r,s) = N_{\ell+1}^{(p+q)/2-1,(r+s)/2-1}((p-q)/2,(r-s)/2)$. Any linear deformation of $\mathrm{Ham}(\bbC^2)$ preserving the action of the split extension of $\fsl_2(\bbC)_+$ must have linear combinations of these intertwiners as structure constants. From this it follows that of the continuous family of algebras $\fsl_\infty^{(s)}(\bbC)$ with $s\geq-1/2$ deforming the wedge subalgebra of $w_{1+\infty}$, only the symplecton with $s=-1/4$ defines a deformation of $\mathrm{Ham}(\bbC^2)$. This has also been noted in \cite{Bu:2022iak}.}


\section{Feynman integrals} \label{app:integrals}

In this appendix we evaluate Feynman integrals appearing in the text. We will employ Feynman's trick
\be \frac{1}{c_1^{\alpha_1}\dots c_n^{\alpha_n}} = \frac{\Gamma(\alpha_1+\dots+\alpha_n)}{\Gamma(\alpha_1)\dots\Gamma(\alpha_n)}\int_{[0,1]^n}\dif t_1\dots\dif t_n\,\frac{t_1^{\alpha_1-1} \dots t_n^{\alpha_n-1}\delta(1-t_1-\dots -t_n)}{(t_1c_1 + \dots + t_nc_n)^{\alpha_1+ \dots +\alpha_n}}\,, \ee
which for $c_1=\dots=c_n=1$ can be rewritten in the form
\be \int_{[0,\infty)^{n-1}}\dif r_1\,\dots\,\dif r_{n-1}\,\frac{r_1^{\alpha_1-1}\dots r_{n-1}^{\alpha_{n-1}-1}}{(1+r_1+\dots+r_{n-1})^{\alpha_1+\dots+\alpha_n}} = \frac{\Gamma(\alpha_1)\dots\Gamma(\alpha_n)}{\Gamma(\alpha_1+\dots+\alpha_n)}\,. \ee\\

In subsection \eqref{subsec:evaluate-diagram} we are required to evaluate
\bea \label{eq:app-int-1}
&\cI_{1,0,0,1}\big(w,z;x^0(x^1)^4,(y^2)^4\big) \\
&= \frac{2^23^2(\bar w-\bar z)}{\pi^8}\int_{(\bbC^3)^2}\dif^6X\,\dif^6Y\,\frac{x^0[\bar x\,\bar y](x^1)^3\bar x^1\bar x^2(y^2)^3\bar y^1\bar y^2}{\|W-X\|^{10}\|X-Y\|^6\|Y-Z\|^{10}}\,,
\eea
where $W = (w,0,0)$, $Z=(z,0,0)$, $X = (x^0,x^\da)$ and $Y = (y^0,y^\da)$

We begin by performing the integral over $Y$
\bea
&\int_{\bbC^3}\dif^6Y\,\frac{[\bar x\,\bar y](y^2)^3\bar y^1\bar y^2}{\|X-Y\|^6\|Y-Z\|^{10}} \\
&= \frac{7!}{2!\,4!}\int_{[0,1]}\dif t\,t^2(1-t)^4\int_{\bbC^3}\dif^6\tilde Y\,\frac{[\bar x\,\bar {\tilde y}]({\tilde y}^2 + t x^2)^3(\bar{\tilde y}^1 + t\bar x^1)(\bar {\tilde y}^2 + t\bar x^2)}{(\|\tilde Y\|^2 + t(1-t)\|X-Z\|^2)^8}\,,
\eea
where we have employed the Feynman trick and defined $\tilde Y = Y - tX - (1-t)Z$. The integral over $\tilde Y$ only receives contributions from terms which have vanishing charge under phase rotations of $\tilde y^\da$. The invariant piece of
\be [\bar x\,\bar y](y^1)^3(\bar y^2)^2 = [\bar x\,\bar {\tilde y}]({\tilde y}^2 + t x^2)^3(\bar{\tilde y}^1 + t\bar x^1)(\bar {\tilde y}^2 + t\bar x^2) \ee
is
\be 3x^2(\bar x^1)^2t^2|\tilde y^2|^2(|\tilde y^2|^2 + t^2|x^2|^2)\,, \ee
allowing us to simplify the above to
\be \label{eq:I2-match} \frac{7!\,3}{2!\,4!}x^2(\bar x^1)^2\int_{[0,1]}\dif t\,t^4(1-t)^4\int_{\bbC^3}\dif^6\tilde Y\,\frac{|\tilde y^2|^2(|\tilde y^2|^2 + t^2|x^2|^2)}{(\|\tilde Y\|^2 + t(1-t)\|X-Z\|^2)^8}\,. \ee
Writing $r_0 = |\tilde y^0|^2$, $r_1=|\tilde y^1|^2$ and $r_2=|\tilde y^2|^2$ this is
\bea
&\frac{7!\,3}{2!\,4!}x^2(\bar x^1)^2(-2\pi\im)^3\int_{[0,1]}\dif t\,t^4(1-t)^4\int_{[0,\infty)^3}\dif r_0\,\dif r_1\,\dif r_2\,\frac{r_2(r_2 + t^2|x^2|^2)}{(r_0+r_1+r_2+t(1-t)\|X-Z\|^2)^8} \\
&= \frac{3}{2!\,4!}\frac{x^2(\bar x^1)^2}{\|X-Z\|^8}(-2\pi\im)^3\int_{[0,1]}\dif t\,\big((2!)^2t(1-t)\|X-Z\|^2 + 3! t^2|x^2|^2\big) \\
&= \frac{(-2\pi\im)^3x^2(\bar x^1)^2}{4!}\bigg(\frac{1}{\|X-Z\|^6} + \frac{3|x^2|^2}{\|Z-X\|^8}\bigg)\,.
\eea
The integral over $X$ in equation \eqref{eq:app-int-1} is then
\bea &\frac{(-2\pi\im)^3}{4!}\int_{\bbC^3}\dif^6X\,\frac{x^0|x^1|^6|x^2|^2}{\|W-X\|^{10}}\bigg(\frac{1}{\|X-Z\|^6} + \frac{3|x^2|^2}{\|X-Z\|^8}\bigg) \\
&= (-2\pi\im)^3\int_{[0,1]}\dif s\,s^4(1-s)^2\int_{\bbC^3}\dif^6\tilde X\,({\tilde x}^0 + s w + (1-s)z)|x^1|^6|x^2|^2 \\ &\bigg(\frac{7!}{2!(4!)^2}\frac{1}{(\|\tilde X\|^2+s(1-s)|w-z|^2)^8} + \frac{8!\,3}{3!(4!)^2}\frac{(1-s)|x^2|^2}{(\|\tilde X\|^2+s(1-s)|w-z|^2)^9}\bigg)\,,
\eea
where $\tilde X = X - sW - (1-s)Z$. The integral over $\tilde X$ may only receive contributions from terms which have vanishing charge under phase rotations of $\tilde x^0$. Writing $r_0=|\tilde x^0|^2$, $r_1=|\tilde x^1|^2$ and $r_2=|\tilde x^2|^2$ we have
\bea \label{eq:I1-penultimate}
&(-2\pi\im)^6\int_{[0,1]}\dif s\,s^4(1-s)^2(sw + (1-s)z)\int_{[0,\infty)^3}\dif r_0\,\dif r_1\,\dif r_2\,(r_1)^3(r_2) \\ &\bigg(\frac{7!}{2!(4!)^2}\frac{1}{(r_0+r_1+r_2+s(1-s)|w-z|^2)^9} + \frac{8!\,3}{3!(4!)^2}\frac{(1-s)r_2}{(r_0+r_1+r_2+s(1-s)|w-z|^2)^8}\bigg) \\
&= \frac{(-2\pi\im)^6}{|w-z|^2}\int_{[0,1]}\dif s\,s^3(1-s)(sw+(1-s)z)\bigg(\frac{3!}{2!(4!)^2} + \frac{3!}{(4!)^2}(1-s)\bigg) \\
&= \frac{(-2\pi\im)^6}{|w-z|^2}\bigg(\frac{3!}{2!(4!)^2}\bigg(\frac{4!}{6!}w + \frac{2!\,3!}{6!}z\bigg) + \frac{3!}{(4!)^2}\bigg(\frac{2!\,4!}{7!}w+\frac{(3!)^2}{7!}z\bigg)\bigg) = - \frac{\pi^6(22w+13z)}{1260|w-z|^2}\,.
\eea
We conclude that
\be \label{eq:I1-final} \cI_{1,0,0,1}\big(w,z;x^0(x^1)^4,(y^2)^4\big) = - \frac{22w+13z}{35\pi^2(w-z)}\,. \ee
Which is the form assumed in equation \eqref{eq:Feynman-integrals}.\\

In the same subsection we also need the value of
\bea \label{eq:app-int-2}
&\cI_{0,1,1,0}\big(w,z;x^0(x^1)^4,(y^2)^4\big) \\
&= \frac{2^23^2(\bar w-\bar z)}{\pi^8}\int_{(\bbC^3)^2}\dif^6X\,\dif^6Y\,\frac{x^0[\bar x\,\bar y](x^1)^3(\bar x^2)^2(y^2)^3(y^1)^2}{\|W-X\|^{10}\|X-Y\|^6\|Y-Z\|^{10}}\,.
\eea
This proceeds similarly to the evaluation of $\cI_{1,0,0,1}\big(w,z;x^0(x^1)^4,(y^2)^4\big)$, so we shall be brief. Having performed the Feynman trick, the integral over $\tilde Y = \tilde Y-tX-(1-t)Z$ only receives contributions from terms which are invariant under phase rotations of $\tilde y^\da$. The invariant piece of
\be [\bar x\,\bar y](y^1)^3(\bar y^2)^2 = [\bar x\,\bar y]({\tilde y}^2 + tx^2)^3(\bar{\tilde y}^1 + t\bar x^1)^2 \ee
is $3t^2(x^2)^2(\bar x^1)^3|\tilde y^2|^2$. At this point the calculation reduces to the second term in equation \eqref{eq:I2-match}. Chasing this through the rest of the calculation to equation \eqref{eq:I1-final} we find that
\be \cI_{0,1,1,0}\big(w,z;x^0(x^1)^4,(y^2)^4\big) = \frac{2^3 3^2(\bar w-\bar z)}{\pi^8}\bigg(-\frac{\pi^6(4w+3z)}{630|w-z|^2}\bigg) = - \frac{2(4w+3z)}{35\pi^2(w-z)}\,. \ee
This is the form quoted in equation \eqref{eq:Feynman-integrals}. Similar calculations show that
\bea
\cI_{0,0,0,0}(w,z;(x^1)^3,(y^2)^3) &= -\frac{3}{16\pi^2(w-z)} \,, \\
\cI_{1,0,0,1}(w,z;(x^1)^4,(y^2)^4) &= -\frac{1}{\pi^2(w-z)} \,, \\
\cI_{0,1,1,0}(w,z;(x^1)^3,(y^2)^3) &= -\frac{2}{5\pi^2(w-z)} \,. \\
\eea


\section{Classical and counterterm operator products involving the axion} \label{app:axion-OPEs}

In this appendix we determine some of the classical and counterterm OPEs involving the $e,f$ states appearing in equations \eqref{eq:OPEs-tree-ef}, \eqref{eq:OPEs-counter} by the method of Koszul duality. For the sake of brevity, we suppress summation symbols.\\

We begin by determining OPEs induced by the classical interaction
\be \frac{1}{4\pi\im}\int_\bbPT\p^{-1}\bseta\,\cL_{\{\bh,\ \}}\bseta = \frac{1}{4\pi\im}\int_\bbPT\bseta\,\{\bh,\ \}\ip\bseta\,. \ee
Note that this term has no counterpart in the case of holomorphic BF theory \cite{Costello:2022wso,Costello:2022upu}. It is responsible for the BRST transformations
\be \label{eq:BRST-transformations-tree} \delta\bg = \frac{1}{2}\cL_{\p^\da}(\bseta\,\p_\da\ip\bseta)\,,\qquad \delta\bseta = - \p(\p^\da\bh\,\p_\da\ip\bseta)\,, \ee
or, in terms of $\bsgamma$,
\be \label{eq:BRST-transformations-tree-gamma} \delta\bsgamma = \p^\da\bh\,\p_\da\ip\p\bsgamma\,. \ee
First consider the coupling of $\bseta$ to the $e$ states in the chiral algebra
\be \frac{1}{2\pi\im}\int_{\cL_0}\dif z\,\bigg(e[m,n](z)\cD_{m,n}\bsgamma_z + \p_ze[m,n](z)\frac{\cD_{m-1,n}\bsgamma_1 + \cD_{m,n-1}\bsgamma_2}{m+n}\bigg)\,. \ee
Under the transformation \eqref{eq:BRST-transformations-tree-gamma}
\bea \label{eq:BRST-tree-e}
&\delta\Bigg(\frac{1}{2\pi\im}\int_{\cL_0}\dif z\,\bigg(e[m,n](z)\cD_{m,n}\bsgamma_z + \p_ze[m,n](z)\frac{\cD_{m-1,n}\bsgamma_1 + \cD_{m,n-1}\bsgamma_2}{m+n}\bigg)\Bigg) \\
&- \frac{1}{2\pi\im}\int_{\cL_0}\dif z\,\bigg((ps-qr)e[p+r-1,q+s-1](z)\cD_{p,q}\bh\cD_{r,s}\bsgamma_z \\ 
&- pe[p+r,q+s](z)\cD_{p,q}\bh\cD_{r+1,s}\p_z\bsgamma_2 + qe[p+r,q+s](z)\cD_{p,q}\bh\cD_{r,s+1}\p_z\bsgamma_1 \\
&- \frac{p+q}{p+q+r+s}\p_ze[p+r,q+s](z)\cD_{p,q}\bh\cD_{r,s}\bseta_{12}\bigg)\,.
\eea
This partially cancelled by the linearised variation of the bilocal term
\bea
&\bigg(\frac{1}{2\pi\im}\bigg)^2\int_{\cL_{0,1}}\dif z_1\,w[p,q](z_1)\cD_{p,q}\bh_1\int_{\cL_{0,2}}\dif z_2\,\bigg(e[r,s](z_2)\cD_{r,s}\bsgamma_{z,2} \\
&+ \p_ze[r,s](z_2)\frac{\cD_{r-1,s}\bsgamma_{1,2} + \cD_{r,s-1}\bsgamma_{2,2}}{r+s}\bigg)\,,
\eea
which is
\bea
&\bigg(\frac{1}{2\pi\im}\bigg)^2\int_{\cL_{0,2}}\Bigg(\bigg(\lim_{\epsilon\to0}\oint_{|z_{12}|=\epsilon}\dif z_{12}\,w[p,q](z_1)e[r,s](z_2)\bigg)\dif z_2\,\cD_{p,q}\bh_1\cD_{r,s}\bsgamma_{z,2} \\
&+ \bigg(\lim_{\epsilon\to0}\oint_{|z_{12}|=\epsilon}\dif z_{12}\,w[p,q](z_1)\p_ze[r,s](z_2)\bigg)\dif z_2\,\cD_{p,q}\bh_1\frac{\cD_{r-1,s}\bsgamma_{1,2}+\cD_{r,s-1}\bsgamma_{2,2}}{r+s}\Bigg)\,.
\eea
By taking
\be \label{eq:OPE-tree-we-app} w[p,q](z)e[r,s](0)\sim\frac{ps-qr}{z}e[p+r-1,q+s-1](0) \ee
we cancel the first term in \eqref{eq:BRST-tree-e}, but this introduces a new contribution
\be - \frac{1}{2\pi\im}\int_{\cL_0}\frac{ps-qr}{r+s}e[p+r-1,q+s-1](z)\cD_{p,q}\bh(\cD_{r-1,s}\p_z\bsgamma_1 + \cD_{r,s-1}\p_z\bsgamma_2)\,. \ee
Therefore the total uncancelled variation is
\bea \label{eq:BRST-tree-uncancelled-e} &\frac{1}{2\pi\im}\int_{\cL_0}\Bigg(\bigg(p-\frac{p(s+1)-q(r+1)}{r+s+2}\bigg)e[p+r,q+s](z)\cD_{p,q}\bh\cD_{r+1,s}\p_z\bsgamma_2 \\
&- \bigg(q+\frac{p(s+1)-q(r+1)}{r+s+2}\bigg)e[p+r,q+s](z)\cD_{p,q}\bh\cD_{r,s+1}\p_z\bsgamma_1 \\
&+ \frac{p+q}{p+q+r+s}\p_ze[p+r,q+s](z)\cD_{p,q}\bh\cD_{r,s}\bseta_{12}\Bigg)\,, \\
&= \frac{1}{2\pi\im}\int_{\cL_0}\bigg(\frac{p+q}{r+s+2}e[p+r,q+s](z)\cD_{p,q}\bh\cD_{r,s}\p_z\bseta_{12} \\
&+ \frac{p+q}{p+q+r+s}\p_ze[p+q,r+s](z)\cD_{p,q}\bh\cD_{r,s}\bseta_{12}\bigg)\,.
\eea
This is compensated by the linearised variation of the bilocal term
\be \bigg(\frac{1}{2\pi\im}\bigg)^2\int_{\cL_{0,1}}\dif z_1\,w[p,q](z_1)\cD_{p,q}\bh_1\int_{\cL_{0,2}}\dif z_2\,f[r,s](z_2)\cD_{r,s}\bseta_{12,2}\,. \ee
Assuming terms in the $w,f$ OPEs of the form
\be w[p,q](z)f[r,s](0)\sim\frac{1}{z^2}\cO^{(2)}_{p,q,r,s}(0) + \frac{1}{z}\cO^{(1)}_{p,q,r,s}(0)\,, \ee
this is equal to
\bea
&\frac{1}{2\pi\im}\int_{\cL_0}\dif z\,\big(\cO^{(2)}_{p,q,r,s}(z)\cD_{p,q}\p_z\bh\cD_{r,s}\bseta_{12} + \cO^{(1)}_{p,q,r,s}(z)\cD_{p,q}\bh\cD_{r,s}\bseta_{12}\big) \\
&= - \frac{1}{2\pi\im}\int_{\cL_0}\dif z\,\big(\cO^{(2)}_{p,q,r,s}(z)\cD_{p,q}\bh\cD_{r,s}\p_z\bseta_{12} + (\p_z\cO^{(2)}_{p,q,r,s}(z) - \cO^{(1)}_{p,q,r,s}(z))\cD_{p,q}\bh\cD_{r,s}\bseta_{12}\big)\,.
\eea
Comparing to equation \eqref{eq:BRST-tree-uncancelled-e} we find that
\bea
&\cO^{(2)}_{p,q,r,s}(z) = \frac{p+q}{r+s+2}e[p+r,q+s](z)\,, \\ &\cO^{(1)}_{p,q,r,s}(z) = \frac{(p+q)(p+q-2)}{(r+s+2)(p+q+r+s)}\p_ze[p+r,q+s](z)\,,
\eea
so that
\bea \label{eq:OPE-tree-wf-1-app}
&w[p,q](z)f[r,s](0)\sim \frac{p+q}{r+s+2}\bigg(\frac{1}{z^2}e[p+q,r+s](0) \\
&+ \frac{1}{z}\frac{p+q-2}{p+q+r+s}\p_ze[p+r,q+s](0)\bigg)\,.
\eea
The OPEs \eqref{eq:OPE-tree-we-app}, \eqref{eq:OPE-tree-wf-1-app} coincide with equations \eqref{eq:OPE-tree-ef}, \eqref{eq:OPE-tree-wf} in the main text. The remaining tree OPEs \eqref{eq:OPE-tree-ef}, \eqref{eq:OPE-tree-ff} can be evaluated using similar arguments.\\

Next let's determine OPEs induced by the counterterm interaction
\be \frac{\mu}{4\pi\im}\int_\bbPT\bseta\,\tr(\bs\p\bs) = \frac{\mu}{4\pi\im}\int_\bbPT\bseta\,\p^\da\p_\db\bh\,\p^\db\p_\da\p\bh\,, \ee
which is responsible for the following BRST transformations
\be \label{eq:BRST-transformations-counter} \delta\bg = - \mu\,\cL_{\p^\da}\cL_{\p_\db}\bseta\,\p^\db\p_\da\p\bh\,,\qquad \delta\bseta = - \frac{\mu}{2}\tr(\p\bs^2) = -\frac{\mu}{2}\p^\da\p_\db\p\bh\,\p^\db\p_\da\p\bh\,, \ee
or, in terms of $\bsgamma$,
\be \label{eq:BRST-transformations-counter-gamma} \delta\bsgamma = \frac{\mu}{2}\p^\da\p_\db\bh\,\p^\db\p_\da\p\bh\,. \ee

Consider the coupling of $\bg$ to a $\tilde w$ generator
\be \frac{1}{2\pi\im}\int_{\cL_0}\dif z\,\tilde w[m,n](z)\cD_{m,n}\tilde\bh\,. \ee
A tedious calculation shows that under the transformations \eqref{eq:BRST-transformations-counter}
\bea \label{eq:BRST-counter-g-1}
&\delta\bigg(\frac{1}{2\pi\im}\int_{\cL_0}\dif z\,\tilde w[m,n](z)\cD_{m,n}\tilde\bh\bigg) \\
&= \frac{\mu}{2\pi\im}\int_{\cL_0}\dif z\,\Big(\tilde w[p+r-2,q+s-2](z)\big(- R_2(p,q,r,s)\cD_{p,q}\bseta_{12}\cD_{r,s}\p_z\bh \\
&+ (2p(q+1)rs(s-1) - p(p-1)s(s-1)(s-2) - (q+1)qr(r-1)s)\cD_{p,q+1}\p_z\bsgamma_1\cD_{r,s}\bh \\
&- (2(p+1)qr(r-1)s - q(q-1)r(r-1)(r-2) - (p+1)prs(s-1))\cD_{p+1,q}\p_z\bsgamma_2\cD_{r,s}\bh\big) \\
&+ \tilde w[p+r-3,q+s-3](z)R_3(p,q,r,s)\cD_{p,q}\bsgamma_z\cD_{r,s}\bh\Big)\,.
\eea
The final term involving $\bsgamma_z$ can be cancelled by the linearised variation of the bilocal term
\bea \label{eq:bilocal-ew} 
&\bigg(\frac{1}{2\pi\im}\bigg)^2\int_{\cL_{0,1}}\dif z_1\,\bigg(e[p,q](z_1)\cD_{p,q}\bsgamma_{z,1} + \p_ze[p,q](z_1)\frac{\cD_{p-1,q}\bsgamma_{1,1} + \cD_{p,q-1}\bsgamma_{2,1}}{p+q}\bigg) \\
&\int_{\cL_{0,2}}\dif z_2\,w[r,s](z_2)\cD_{r,s}\bh_2\,,
\eea
which is
\bea \label{eq:BRST-bilocal-ew}
&\bigg(\frac{1}{2\pi\im}\bigg)^2\int_{\cL_{0,2}}\Bigg(\bigg(\lim_{\epsilon\to0}\oint_{|z_{12}|=\epsilon}\dif z_{12}\,e[p,q](z_1)w[r,s](z_2)\bigg)\dif z_2\,\cD_{p,q}\bsgamma_{z,1}\cD_{r,s}\bh_2\\
&+ \bigg(\lim_{\epsilon\to0}\oint_{|z_{12}|=\epsilon}\dif z_{12}\,\p_z e[p,q](z_1)w[r,s](z_2)\bigg)\dif z_2\,\frac{\cD_{p-1,q}\bsgamma_{1,1} + \cD_{p,q-1}\bsgamma_{2,1}}{p+q}\cD_{r,s}\bh_2\Bigg)\,.
\eea
Comparing to equation \eqref{eq:BRST-counter-g-1} it's clear that by taking
\be \label{eq:OPE-counter-we-app} e[p,q](z)w[r,s](0)\sim - \frac{\mu R_3(p,q,r,s)}{z}\tilde w[p+r-3,q+s-3](0) \ee
we can cancel the term involving $\bsgamma_z$. However, the second line in equation \eqref{eq:BRST-bilocal-ew} now contributes. A careful calculation shows that the total uncancelled variation is
\be
- \frac{\mu}{2\pi\im}R_2(p,q,r,s)\int_{\cL_0}\dif z\,\tilde w[p+r-2,q+s-2](z)\bigg(\cD_{p,q}\bseta_{12}\cD_{r,s}\p_z\bh - \frac{r+s-2}{p+q+2}\cD_{p,q}\p_z\bseta_{12}\cD_{r,s}\bh\bigg)\,.
\ee
Integrating by parts with $\p_z$ in the first term this is
\bea \label{eq:BRST-counter-g-2}
&\frac{\mu}{2\pi\im}R_2(p,q,r,s)\int_{\cL_0}\dif z\,\bigg(\frac{p+q+r+s}{p+q+2}\tilde w[p+r-2,q+s-2](z)\cD_{p,q}\p_z\bseta_{12}\cD_{r,s}\bh \\
&+ \p_z\tilde w[p+r-2,q+s-2](z)\cD_{p,q}\bseta_{12}\cD_{r,s}\bh\bigg)\,.
\eea
It can be cancelled by the linearised variation of the bilocal term
\be \bigg(\frac{1}{2\pi\im}\bigg)^2\int_{\cL_{0,1}}\dif z_1\,f[p,q](z_1)\cD_{p,q}\bseta_{12,1}\int_{\cL_{0,2}}\dif z_2\,w[r,s](z_2)\cD_{r,s}\bh_2\,, \ee
which is
\be \bigg(\frac{1}{2\pi\im}\bigg)^2\int_{\cL_{0,2}}\bigg(\lim_{\epsilon\to0}\oint_{|z_{12}|=\epsilon}\dif z_{12}\,f[p,q](z_1)w[r,s](z_2)\bigg)\dif z_2\,\cD_{p,q}\bseta_{12,1}\cD_{r,s}\bh_2\,. \ee
It's straightforward to see that for the above to offset equation \eqref{eq:BRST-counter-g-2} we should take
\bea \label{eq:OPE-counter-wf-app}
&f[p,q](z)w[r,s](0) \sim - \mu R_2(p,q,r,s)\bigg(\frac{1}{z^2}\frac{p+q+r+s}{p+q+2}\tilde w[p+r-2,q+s-2](0) \\
&+ \frac{1}{z}\p_z\tilde w[p+r-2,q+s-2](0)\bigg)\,.
\eea
The OPEs \eqref{eq:OPE-counter-we-app}, \eqref{eq:OPE-counter-wf-app} coincide with equations \eqref{eq:OPE-counter-we}, \eqref{eq:OPE-counter-wf} in the main text. The remaining counterterm OPEs \eqref{eq:OPE-counter-ww} can be computed in a similar manner.

\end{appendix}


\newpage

\bibliographystyle{JHEP}
\bibliography{references}


\end{document}